\documentclass[thmsa,12pt]{article}
\usepackage{amsthm,amsmath,amsfonts,amssymb}
\usepackage{natbib}
\usepackage[figuresright]{rotating}
\usepackage[colorlinks,citecolor=blue,urlcolor=blue]{hyperref}
\usepackage{array}

\input{tcilatex}
\newtheorem*{cor}{Corollary}
\newtheorem*{prop}{Proposition}
\newtheorem*{lemm}{Lemma}
\newtheorem*{examp}{Example}
\newcommand{\PreserveBackslash}[1]{\let\temp=\\#1\let\\=\temp}
\newcolumntype{C}[1]{>{\PreserveBackslash\centering}p{#1}}

\begin{document}

\title{\textbf{On semiparametric estimation of the intercept of the sample
selection model: \\
a kernel approach}}
\author{Zhewen Pan \\
School of Economics, Zhejiang University of Finance \& Economics \\
{\normalsize 18 Xueyuan Street, Xiasha, Hangzhou 310018, China. Email:
panzhew@zufe.edu.cn}}
\maketitle

\begin{abstract}
This paper presents a new perspective on the identification at infinity for
the intercept of the sample selection model as identification at the
boundary via a transformation of the selection index. This perspective
suggests generalizations of estimation at infinity to kernel regression
estimation at the boundary and further to local linear estimation at the
boundary. The proposed kernel-type estimators with an estimated
transformation are proven to be nonparametric-rate consistent and
asymptotically normal under mild regularity conditions. A fully data-driven
method of selecting the optimal bandwidths for the estimators is developed.
The Monte Carlo simulation shows the desirable finite sample properties of
the proposed estimators and bandwidth selection procedures.
\end{abstract}



\newpage

\section{Introduction}

\label{sec:intro}Since it was first introduced by the seminal papers of \cite%
{heckman1974shadow}, \cite{gronau1974wage}$,$ and \cite{lewis1974comments},
the sample selection model has been increasingly widely applied in empirical
studies to address potentially nonrandom samples that may arise from a
variety of causes, such as improper sampling design, self-selectivity,
nonresponse on survey questions, and attrition from social programs. Several
important applications of the sample selection model, such as the estimation
of average treatment effects and the decomposition of wage differentials,
rely on an estimate of the intercept.

The intercept of the sample selection model is conventionally estimated
along with the slope coefficients by means of the parametric maximum
likelihood approach or the likelihood-based two-step procedure %
\citep{heckman1979sample}. When the distribution of the model disturbance is
misspecified, however, the parametric likelihood-based estimators for
limited dependent variable models may possess evident bias and are likely
inconsistent for the true value of interest %
\citep[e.g.,][]{arabmazar1982investigation}. This finding spawned
considerable and influential literature on semiparametric identification and
estimation approaches that do not rest on parametric specification of the
disturbance distribution, among which \cite{chamberlain1986asymptotic}
invented the notion of \textquotedblleft identification at
infinity\textquotedblright , building identification upon the unbounded
support of the regressor distribution. According to the identification at
infinity, \cite{heckman1990varieties} proposed a semiparametric estimator
for the intercept of the sample selection model. However, %
\citeauthor{heckman1990varieties}'s estimator contains discontinuous
indicator functions, which greatly complicate the analysis of the
statistical properties. \cite{andrews1998semiparametric} suggested replacing
the indicator function with a smoothed version and established the
consistency and asymptotic normality of their modified estimator. Since
then, the identification-at-infinity method has been standard for
semiparametrically estimating the intercept of the sample selection model 
\citep[see][among others]{schafgans1998ethnic, schafgans2000gender,
hussinger2008r, mulligan2008selection, liu2009maternal, shen2013determinants}%
.

Nevertheless, identification-at-infinity estimation must choose a smoothing
parameter or bandwidth controlling for the proportion of observations used,
and no bandwidth selection algorithm that is both theoretically valid and
practically tractable in the context of the intercept estimation has yet
been developed. This is mainly because the asymptotic biases and variances
of the estimators of \cite{heckman1990varieties} and \cite%
{andrews1998semiparametric} are all implicit functions of the bandwidth. To
choose a theoretically valid bandwidth, one must impose high-level
assumptions on the tail behaviors of the regressor and disturbance
distributions in the selection equation and then carefully estimate some
sort of tail indices, as in \cite{klein2015estimation}. However, a precise
estimate of the tail index of the selection disturbance distribution is
difficult to obtain for two reasons. First, only the information in the
tails, which is limited, is useful for estimating tail indices. Second,
there is additional information loss when estimating the tail index of the
selection disturbance distribution because one can observe only the binary
selection outcome instead of the continuous selection disturbance. The
imprecise estimation of the tail index makes it difficult to choose an
eligible bandwidth in practice \citep{tan2018root}. In empirical studies,
researchers typically report several estimates of the intercept according to
different values of bandwidth selected by simple rules, such as by sample
quantiles of the selection linear index. However, these simple rules lack
theoretical justification, and an inappropriately selected bandwidth may
lead to sizable estimation bias and misleading inference results %
\citep{schafgans2004finite}. Additionally, confusion will arise if different
bandwidths lead to contradictory conclusions.

In this paper, I extend the identification-at-infinity estimators to a
kernel regression estimator at the boundary and provide a simple bandwidth
selection algorithm that is theoretically optimal in the sense of minimizing
the asymptotic mean squared estimation error. By transforming the
identification at infinity into identification at the boundary, I propose a
kernel regression estimator for the intercept of the sample selection model
that includes the estimators of \cite{heckman1990varieties} and \cite%
{andrews1998semiparametric} as special cases by taking particular forms of
transformation. I suggest using the cumulative distribution function (CDF)
of the selection linear index as the transformation, which makes the
proposed kernel estimator substantially different from the existing
identification-at-infinity estimators. Under this specific form of
transformation, the asymptotic bias and variance of the kernel estimator are
explicit functions of the bandwidth, motivating a simple plug-in procedure
of optimal bandwidth selection. In practice, however, the CDF of the linear
index is unknown. I adopt the empirical CDF estimation and show that the
induced estimation error is asymptotically negligible under regularity
conditions, implying that the kernel estimator with the empirical CDF
follows the same asymptotic distribution as if the true CDF were known. As a
result, the proposed bandwidth selection algorithm is theoretically valid.

Careful consideration of the asymptotic distribution of the kernel
regression estimator for the intercept indicates a boundary effect in that
the asymptotic bias has a larger order than the nonparametric kernel
regression estimator in the interior because the estimator is, in nature, a
Nadaraya-Watson or local constant estimator at the boundary. Hence, I
further consider the local linear estimation of the intercept for bias
reduction and show that it achieves a univariate nonparametric rate. The
consistency and asymptotic normality of the estimator are established, and
an optimal bandwidth selection algorithm tailored to the local linear case
is presented.

The main contributions of this paper are as follows. First, I provide a
novel interpretation of identification at infinity. Although infinity is
intractable or at least irregular in most econometric models, it is merely
an ordinary boundary point of the extended real line. Under the monotonic
transformation that serves to construct a metric, the identification at
infinity is converted into identification at the boundary, and kernel-type
estimators are naturally derived. This approach may be helpful in other
irregular identification problems involving infinity 
\citep[see,
e.g.,][]{khan2010irregular}. Second, I provide a theoretically justified
procedure for bandwidth selection that is easy to implement. By implementing
a special form of transformation, the asymptotic biases and variances of the
kernel-type estimators become explicit functions of the bandwidth, based on
which a fully data-driven algorithm is developed to choose the optimal
bandwidth. The regularization method in \cite{imbens2012optimal} is employed
to keep the random denominator of the estimated bandwidth away from zero.
Third, in the course of developing the statistical properties of the
estimators, I provide a\ rigorous treatment of the estimation error induced
by the empirical CDF. By means of a decomposition into an empirical process
component and a discontinuous error component, the estimation error of the
empirical CDF can be delicately controlled by the uniform rate of
convergence of the empirical process over shrinking intervals %
\citep[e.g.,][]{stute1982oscillation} and by the convergence results
developed by \cite{schafgans2002intercept}.

The rest of this paper is organized as follows. Section \ref{sec:model}
presents the sample selection model and reviews the semiparametric
estimators for the intercept in the literature. Section \ref{sec:KRE}
motivates the kernel regression estimator, provides the regularity
conditions under which the estimator is consistent and asymptotically
normal, discusses the choice of the transformation, and presents the optimal
bandwidth selection method. Section \ref{sec:LL} proposes local linear
estimation to eliminate the boundary effect. Section \ref{sec:simulation}
reports the results from a small simulation and Section \ref{sec:conclution}
concludes. All the technical proofs are left to the Appendix.

\section{The model and existing estimators}

\label{sec:model}

\subsection{The model}

The sample selection model has the following three-equation form: 
\begin{eqnarray}
D_{i} &=&1\left\{ X_{i}^{\prime }\beta _{0}>\varepsilon _{i}\right\} , 
\notag \\
Y_{i}^{\ast } &=&\mu _{0}+Z_{i}^{\prime }\theta _{0}+U_{i},  \label{model} \\
Y_{i} &=&Y_{i}^{\ast }D_{i},  \notag
\end{eqnarray}%
where $D_{i}$ is a binary selection variable, $Y_{i}^{\ast }$ is the latent
outcome, and $Y_{i}$ is the observed outcome. $X_{i}$ and $Z_{i}$ are random
vectors of regressors (excluding the constant term) of the selection and
outcome equations, respectively. $\varepsilon _{i}$ and $U_{i}$ are scalar
disturbances that are possibly correlated with each other. $\beta _{0}$ and $%
\theta _{0}$ are vectors of the slope coefficients, and $\mu _{0}$ is a
scalar intercept coefficient, both of which are the estimands of the model. $%
U_{i}$ is assumed to have zero mean to ensure identifiability of the
intercept.

This paper is primarily concerned with the estimation of $\mu _{0}$ in the
sample selection model (\ref{model}) without imposing a normal (or any other
parametric) distribution restriction on the disturbances. The estimation of $%
\mu _{0}$, however, relies on preliminary $\sqrt{n}$-consistent estimators
for $\beta _{0}$ and $\theta _{0}$, which are denoted by $\hat{\beta}$ and $%
\hat{\theta}$. Several such estimators are available in the literature. For
instance, $\hat{\beta}$ can be one of the semiparametric estimators for the
binary response model %
\citep[e.g.,][]{klein1993efficient,lewbel2000semiparametric} or for the
single-index model %
\citep[e.g.,][]{powell1989semiparametric,ichimura1993semiparametric}, and $%
\hat{\theta}$ can be one of the semiparametric slope estimators for the
sample selection model 
\citep[e.g.,][to name a few]{gallant1987semi, chen1998efficient,
powell2001semiparametric, lewbel2007endogenous, newey2009twostep, chen2010semiparametric}%
. For notational convenience, denote $W_{i}=X_{i}^{\prime }\beta _{0}$ and $%
\hat{W}_{i}=X_{i}^{\prime }\hat{\beta}$, hence $D_{i}=1\left\{
W_{i}>\varepsilon _{i}\right\} $.

\subsection{The identification-at-infinity estimators}

A natural strategy for identifying the intercept $\mu _{0}$ starts from the
expectation of $Y_{i}-Z_{i}^{\prime }\theta _{0}$ conditional on the
observation being selected:%
\begin{equation*}
E\left[ Y_{i}-Z_{i}^{\prime }\theta _{0}\left\vert D_{i}=1\right. \right] =E%
\left[ Y_{i}^{\ast }-Z_{i}^{\prime }\theta _{0}\left\vert D_{i}=1\right. %
\right] =\mu _{0}+E\left[ U_{i}\left\vert D_{i}=1\right. \right] .
\end{equation*}%
In the context of sample selection, where $U_{i}$ is correlated with $%
\varepsilon _{i}$, the selectivity bias term $E\left[ U_{i}\left\vert
D_{i}=1\right. \right] $ is generally nonzero and contaminates the
identification of $\mu _{0}$. Inspired by the fact that $E\left[
U_{i}\left\vert D_{i}=1\right. \right] $ will become arbitrarily close to
zero for those $W_{i}$ such that $\Pr \left( D_{i}=1\left\vert W_{i}\right.
\right) $ is arbitrarily close to unity, \cite{heckman1990varieties}
suggested identifying $\mu _{0}$ at infinity, i.e., when $W_{i}$ takes
arbitrarily large values. This idea can be heuristically formulated as%
\begin{equation*}
E\left[ Y_{i}-Z_{i}^{\prime }\theta _{0}\left\vert D_{i}=1,W_{i}=+\infty
\right. \right] =\mu _{0}+E\left[ U_{i}\left\vert D_{i}=1\right.
,W_{i}=+\infty \right] =\mu _{0}+E\left[ U_{i}\left\vert W_{i}=+\infty
\right. \right] =\mu _{0},
\end{equation*}%
provided the conditional mean independence is assumed. An intuitive
estimator based on the above identification strategy is%
\begin{equation}
\hat{\mu}^{H}=\frac{\sum_{i=1}^{n}\left( Y_{i}-Z_{i}^{\prime }\hat{\theta}%
\right) D_{i}\cdot 1\left\{ \hat{W}_{i}>\gamma _{n}\right\} }{%
\sum_{i=1}^{n}D_{i}\cdot 1\left\{ \hat{W}_{i}>\gamma _{n}\right\} },
\label{H1990Est}
\end{equation}%
where $\left\{ \gamma _{n}\right\} $ is a sequence of positive smoothing
parameters such that $\gamma _{n}\rightarrow \infty $ as $n\rightarrow
\infty $. The estimator $\hat{\mu}^{H}$ is essentially a sample average of $%
\mu _{0}+U_{i}$ over a decreasingly small fraction of all observations. The
effective sample size depends on the proportion of data censoring, the
degree of tail heaviness of $\hat{W}_{i}$'s distribution, and the choice of $%
\gamma _{n}$.

To facilitate the development of distribution theory of the
identification-at-infinity estimation, \cite{andrews1998semiparametric}
replaced the indicator function in \citeauthor{heckman1990varieties}'s
estimator (\ref{H1990Est}) with a smoothed function and proposed a weighted
sample average estimator:%
\begin{equation}
\hat{\mu}^{AS}=\frac{\sum_{i=1}^{n}\left( Y_{i}-Z_{i}^{\prime }\hat{\theta}%
\right) D_{i}\cdot s\left( \hat{W}_{i}-\gamma _{n}\right) }{%
\sum_{i=1}^{n}D_{i}\cdot s\left( \hat{W}_{i}-\gamma _{n}\right) },
\label{AS1998Est}
\end{equation}%
where $s\left( \cdot \right) $ is a nondecreasing $\left[ 0,1\right] $%
-valued function that has a third derivative bounded over $R$ and satisfies $%
s\left( w\right) =0$ for $w\leq 0$ and $s\left( w\right) =1$ for $w\geq b$,
for some $b>0$. By taking advantage of the smoothness of $s\left( \cdot
\right) $, \cite{andrews1998semiparametric} showed the asymptotic
negligibility of the estimation error induced by the preliminary estimators $%
\hat{\beta}$ and $\hat{\theta}$ and established the consistency and
asymptotic normality of $\hat{\mu}^{AS}$ under mild regularity conditions.
The rate of convergence of $\hat{\mu}^{AS}$ depends both upon the tail
heaviness of $W_{i}$'s distribution and upon the rate of divergence of $%
\gamma _{n}$. However, theoretically valid procedures for choosing $\gamma
_{n}$ have not yet been developed.

\subsection{Other existing estimators}

The semiparametric literature has neglected, to a certain extent, the
intercept estimation of the sample selection model. This is mainly because
the intercept is absorbed into the selectivity bias correction term in the
course of estimating the slope coefficients. The only exceptions, besides $%
\hat{\mu}^{H}$ and $\hat{\mu}^{AS}$, are the estimators of \cite%
{gallant1987semi}, \cite{lewbel2007endogenous}, and \cite%
{chen2010semiparametric}.

By employing Hermite series to approximate the unknown bivariate density
function of the disturbances, \cite{gallant1987semi} considered
semiparametric maximum likelihood estimation of the intercept along with the
slopes. The consistency of their estimator requires complicated continuity
conditions on the distributions of the disturbances and regressors that are
difficult to verify. Moreover, the asymptotic distribution of their
estimator has not been established. \cite{lewbel2007endogenous} achieved
identification of the intercept by the presence of a special regressor that
has a large support. When the support is infinite, his identification
strategy would be essentially equivalent to the identification at infinity. 
\cite{chen2010semiparametric} proposed a kernel-weighted
pairwise-difference-type estimator for the intercept and showed its
consistency and asymptotic normality. However, their estimation method
requires the disturbances to be jointly symmetrically distributed; this
joint symmetry condition may not hold in practice.

\section{Kernel regression estimation}

\label{sec:KRE}The kernel approach to estimating the intercept $\mu _{0}$ in
the sample selection model (\ref{model}) is grounded in the identification
at infinity:%
\begin{equation}
\mu _{0}=E\left[ Y_{i}-Z_{i}^{\prime }\theta _{0}\left\vert
D_{i}=1,W_{i}=+\infty \right. \right] .  \label{IdenAtInf}
\end{equation}%
Let $F\left( \cdot \right) $ be an absolutely continuous CDF that is
strictly increasing over $R$; then, the infinity condition $W_{i}=+\infty $
is equivalent to a boundary condition $F\left( W_{i}\right) =F\left( +\infty
\right) =1$. Therefore, the identification at infinity (\ref{IdenAtInf}) can
be written as identification at the boundary:%
\begin{equation}
\mu _{0}=E\left[ Y_{i}-Z_{i}^{\prime }\theta _{0}\left\vert D_{i}=1,F\left(
W_{i}\right) =1\right. \right] .  \label{IdenAtBoundary}
\end{equation}%
In other words, $\mu _{0}$ is identified by the conditional expectation of $%
Y_{i}-Z_{i}^{\prime }\theta _{0}$ given that the observation is selected and
that the transformation of $W_{i}$, $F\left( W_{i}\right) $, takes a
boundary value. The identification-at-boundary equation (\ref{IdenAtBoundary}%
) motivates the kernel regression estimation of $\mu _{0}$:%
\begin{equation}
\tilde{\mu}=\frac{\displaystyle\sum_{i=1}^{n}\left( Y_{i}-Z_{i}^{\prime }%
\hat{\theta}\right) D_{i}k\left( \frac{1-F\left( \hat{W}_{i}\right) }{h_{n}}%
\right) }{\displaystyle\sum_{i=1}^{n}D_{i}k\left( \frac{1-F\left( \hat{W}%
_{i}\right) }{h_{n}}\right) },  \label{LC}
\end{equation}%
where $\hat{W}_{i}=X_{i}^{\prime }\hat{\beta}$, $\hat{\beta}$ and $\hat{%
\theta}$ are preliminary estimators for $\beta _{0}$ and $\theta _{0}$, $%
k\left( \cdot \right) $ is a kernel function defined on $\left[ 0,\infty
\right) $, and $\left\{ h_{n}\right\} $ is a sequence of positive bandwidth
parameters such that $h_{n}\rightarrow 0$ as $n\rightarrow \infty $. The
values of the kernel over the negative reals are irrelevant since the
argument of it is always positive. Some examples of $k\left( \cdot \right) $
are given in Table \ref{table:kernel}. $\tilde{\mu}$ is a kernel regression
estimator at the boundary. Alternatively, it can be viewed as a kernel
regression estimator at infinity over the extended real line $\bar{R}=\left[
-\infty ,+\infty \right] =R\cup \left\{ -\infty ,+\infty \right\} $. A
distance function with which $\bar{R}$ would become a (compact) metric space
can be defined as $d\left( w_{1},w_{2}\right) =\left\vert F\left(
w_{1}\right) -F\left( w_{2}\right) \right\vert $\ for any $w_{1},w_{2}\in 
\bar{R}$. Accordingly, the numerator within the kernel is the distance
between $\hat{W}_{i}$ and infinity since $d\left( \hat{W}_{i},+\infty
\right) =\left\vert F\left( \hat{W}_{i}\right) -1\right\vert =1-F\left( \hat{%
W}_{i}\right) $. From this perspective, our estimation method is an
extension of kernel regression to the extended real line.

Similar to the identification-at-infinity estimators $\hat{\mu}^{H}$ and $%
\hat{\mu}^{AS}$ defined in (\ref{H1990Est}) and (\ref{AS1998Est}), the
kernel estimator $\tilde{\mu}$ is also a (weighted) sample average over a
vanishingly small subset of the entire data set. However, the local average
in $\tilde{\mu}$ is over observations for $F\left( \hat{W}_{i}\right) $ in a
left neighborhood of the right boundary point, rather than for $\hat{W}_{i}$
near positive infinity. Notably, $\tilde{\mu}$ generalizes the
identification-at-infinity estimators in the sense that it nests $\hat{\mu}%
^{H}$ and $\hat{\mu}^{AS}$ as special cases by taking specific forms of $%
F\left( \cdot \right) $, $k\left( \cdot \right) $ and $h_{n}$.

\begin{prop}
(i) For any strictly increasing $F\left( \cdot \right) $, if $k\left(
u\right) =1\left\{ 0\leq u\leq 1\right\} $ and $h_{n}=1-F\left( \gamma
_{n}\right) $, then $\tilde{\mu}=\hat{\mu}^{H}$. (ii) If $F\left( \cdot
\right) $ is the standard Laplacian CDF such that $F\left( w\right)
=1-\left( 1\left/ 2\right. \right) \exp \left( -w\right) $ for $w\geq 0$, if 
$k\left( u\right) =s\left( -\log u\right) $, and if $h_{n}=1-F\left( \gamma
_{n}\right) =\left( 1\left/ 2\right. \right) \exp \left( -\gamma _{n}\right) 
$, then $\tilde{\mu}=\hat{\mu}^{AS}$.
\end{prop}

\subsection{Asymptotic property}

This subsection investigates the consistency and asymptotic normality of the
kernel estimator $\tilde{\mu}$ under general $F\left( \cdot \right) $.
Several regularity assumptions are made first. Denote $k_{ni}=k\left( \left.
\left( 1-F\left( W_{i}\right) \right) \right/ h_{n}\right) $, where $%
W_{i}=X_{i}^{\prime }\beta _{0}$.

\bigskip

\noindent\textbf{Assumption 1.} (i) $\left\{ \left(
Y_{i},D_{i},X_{i},Z_{i}\right) \right\} _{i=1}^{n}$ is a random sample of $n$
observations from Model (\ref{model}). (ii) The model disturbance $\left(
\varepsilon _{i},U_{i}\right) $ is independent of the regressors $X_{i}$ and 
$Z_{i}$. (iii) $EU_{i}=0$. (iv) There exists a positive constant $c_{1}$
such that $E\left| U_{i}\right| ^{2+c_{1}}<\infty $, $E\left\| X_{i}\right\|
^{4+c_{1}}<\infty $, and $E\left\| Z_{i}\right\| ^{2+c_{1}}<\infty $.

\noindent\textbf{Assumption 2.} The kernel function $k\left( \cdot \right) $
is defined on $\left[ 0,\infty \right) $ and satisfies that (i) it is
nonnegative and supported on $\left[ 0,1\right] $, (ii) it is bounded above
by $\bar{k}>0$, and (iii) it is twice continuously differentiable over $%
\left[ 0,\infty \right) $ and its derivatives $k^{\prime }\left( \cdot
\right) $ and $k^{\prime \prime }\left( \cdot \right) $ are bounded above by 
$\bar{k^{\prime }}$ and $\bar{k^{\prime \prime }}$, respectively.

\noindent\textbf{Assumption 3.} (i) The transformation function $F\left(
\cdot \right) $ is a CDF for a continuously distributed random variable
whose support is right-unbounded. (ii) The probability density function
(PDF) $f\left( \cdot \right) $ corresponding to $F\left( \cdot \right) $ is
absolutely continuous with respect to the Lebesgue measure, with its
derivative $f^{\prime }\left( \cdot \right) $ bounded almost everywhere by $%
\bar{f^{\prime }}$. (iii) For a large positive constant $C$, the tail hazard
rate $H_{C}\left( \cdot \right) =1\left\{ \cdot >C\right\} f\left( \cdot
\right) \left/ \left[ 1-F\left( \cdot \right) \right] \right. $ and the tail
score (of location parameter) $S_{C}\left( \cdot \right) =1\left\{ \cdot
>C\right\} \left. f^{\prime }\left( \cdot \right) \right/ f\left( \cdot
\right) $ corresponding to $F\left( \cdot \right) $ satisfy $E\left[
\sup_{\beta \in \mathcal{N}\left( \beta _{0}\right) }H_{C}^{4}\left(
X_{i}^{\prime }\beta \right) \right] <\infty $ and $E\left[ \sup_{\beta \in 
\mathcal{N}\left( \beta _{0}\right) }S_{C}^{4}\left( X_{i}^{\prime }\beta
\right) \right] <\infty $ for a neighborhood $\mathcal{N}\left( \beta
_{0}\right) $ of $\beta _{0}$.

\noindent\textbf{Assumption 4.} The preliminary estimators $\hat{\beta}$ and 
$\hat{\theta}$ are $\sqrt{n}$-consistent, that is, $\sqrt{n}\left\| \hat{%
\beta}-\beta _{0}\right\| =O_{p}\left( 1\right) $ and $\sqrt{n}\left\| \hat{%
\theta}-\theta _{0}\right\| =O_{p}\left( 1\right) $.

\noindent \textbf{Assumption 5.} The bandwidth satisfies that $0<h_{n}\leq
1/2$ and that as $n\rightarrow \infty $, (i) $h_{n}\rightarrow 0$, (ii) $%
nEk_{ni}^{2}\rightarrow \infty $, and (iii) $\left. \left[ \Pr \left(
F\left( W_{i}\right) >1-h_{n}\right) \right] ^{1+c}\right/
Ek_{ni}^{2}\rightarrow 0$ for any $c>0$.

\bigskip

Assumption 1 describes the model and data. Assumptions 2 and 3 impose
smoothness and boundedness conditions on the kernel function $k\left( \cdot
\right) $ and the transformation function $F\left( \cdot \right) $. The
compactness of the kernel's support is assumed to simplify the technical
proofs, and the theoretical results in this paper are still supposed to hold
when using kernels that decay sufficiently fast in the tails. Assumption
3.(iii) is a joint condition on the tail behaviors of $F\left( \cdot \right) 
$ and $X_{i}$'s distribution. When $F\left( \cdot \right) $ has a power-type
upper tail, e.g., Pareto($\lambda $) tail such that\footnote{%
Here and below, $g\left( t\right) \sim h\left( t\right) $ means that the
ratios $g\left( t\right) \left/ h\left( t\right) \right. $ and $h\left(
t\right) \left/ g\left( t\right) \right. $ are $O\left( 1\right) $ as $%
t\rightarrow +\infty $.} $1-F\left( t\right) \sim t^{-\lambda }$ for some $%
\lambda >0$, we have $H_{C}\left( t\right) \sim t^{-1}$ and $S_{C}\left(
t\right) \sim t^{-1}$, and Assumption 3.(iii) holds for any distribution of $%
X_{i}$. When $F\left( \cdot \right) $ has an exponential-type upper tail,
e.g., Weibull($\lambda $) tail such that $1-F\left( t\right) \sim \exp
\left( -c_{0}t^{\lambda }\right) $ for some $\lambda >0$ and $c_{0}>0$, we
have $H_{C}\left( t\right) \sim t^{\lambda -1}$ and $S_{C}\left( t\right)
\sim t^{\lambda -1}$. If $0<\lambda \leq 1$, Assumption 3.(iii) still holds
for any distribution of $X_{i}$. For example, in the special case of \cite%
{andrews1998semiparametric}'s estimator that corresponds to a Laplacian $%
F\left( \cdot \right) $ (i.e., $\lambda =1$), Assumption 3.(iii) is
automatically satisfied. If $\lambda >1$, then Assumption 3.(iii) will
require $E\left\Vert X_{i}\right\Vert ^{4\left( \lambda -1\right) }<\infty $%
. For example, if $F\left( \cdot \right) $ is the normal CDF, we need $%
E\left\Vert X_{i}\right\Vert ^{4}<\infty $, which is already guaranteed by
Assumption 1.(iv). When the tail of $F\left( \cdot \right) $ decays more
rapidly as $1-F\left( t\right) \sim \exp \left( -\exp \left( c_{0}t^{\lambda
}\right) \right) $, a sufficient condition of Assumption 3.(iii) is that the
distribution of $\left\Vert X_{i}\right\Vert $ has a Weibull($\theta $) tail
with $\theta >\lambda $. Assumption 4 imposes the $\sqrt{n}$-consistency of
the preliminary slope estimators. Examples of such estimators are listed in
Subsection \ref{sec:model}.1.

According to Assumption 5, the bandwidth $h_{n}$ is required to go to zero
as the sample size goes to infinity, but the speed of the decline is not
allowed to be excessively fast. To see this, note that $Ek_{ni}^{2}\leq \bar{%
k}^{2}\Pr \left( F\left( W_{i}\right) >1-h_{n}\right) $ by Assumption 2. As $%
h_{n}$ goes to zero, $\Pr \left( F\left( W_{i}\right) >1-h_{n}\right) $ will
also go to zero. If the speed at which $h_{n}$ declines is so fast that $\Pr
\left( F\left( W_{i}\right) >1-h_{n}\right) $ goes to zero at a rate faster
than $n^{-1}$, we will have $nEk_{ni}^{2}\rightarrow 0$, which violates
Assumption 5.(ii). Another implicit requirement of Assumption 5.(ii) is the
upper unboundedness of $W_{i}$'s support, for if $W_{i}$ has bounded support
from above, $\Pr \left( F\left( W_{i}\right) >1-h_{n}\right) $ will be
exactly zero for sufficiently small $h_{n}$. Assumption 5.(iii) requires the
distribution of $W_{i}$ to be not too thin upper tailed relative to $F\left(
\cdot \right) $, as illustrated by the following example.

\begin{examp}
Suppose $F\left( \cdot \right) $ is the standard Laplacian CDF, as is the
case with $\hat{\mu}^{AS}$. Then, Assumption 5.(iii) holds if

(i) the distribution of $W_{i}$ has a power-type upper tail;

(ii) the distribution of $W_{i}$ has an exponential-type upper tail;

(iii) the distribution of $W_{i}$ has a ``double''-exponential-type but not
excessively thin upper tail, namely, $\Pr \left( W_{i}>t\right) \sim \exp
\left( -\exp \left( c_{0}t^{\lambda }\right) \right) $ with $\lambda <1$.
\end{examp}

\begin{theorem}
\label{Thm:lite}Under Assumptions 1-5, the kernel regression estimator $%
\tilde{\mu}$ defined in (\ref{LC}) is consistent and asymptotically normal:%
\begin{equation*}
\frac{\sqrt{n}Ek_{ni}}{\sqrt{Ek_{ni}^{2}}}\left( \tilde{\mu}-\mu _{0}-\frac{%
EU_{i}D_{i}k_{ni}}{ED_{i}k_{ni}}\right) \rightarrow N\left( 0,\sigma
_{U}^{2}\right) ,
\end{equation*}%
where $\sigma _{U}^{2}=EU_{i}^{2}$.
\end{theorem}

Theorem \ref{Thm:lite} generalizes Theorems 1-2 of \cite%
{andrews1998semiparametric} by allowing general forms of the transformation
function $F\left( \cdot \right) $. At first sight, introducing a general $%
F\left( \cdot \right) $ appears to help little in the choice of $h_{n}$
because the asymptotic bias and variance of $\tilde{\mu}$ are still unknown
functions of $h_{n}$. However, as will be seen below, a particular
data-dependent choice of $F\left( \cdot \right) $ can largely facilitate the
subsequent choice of $h_{n}$.

\subsection{Choice of the transformation $F\left( \cdot \right) $}

A natural choice of $F\left( \cdot \right) $ would be the CDF of $%
W_{i}=X_{i}^{\prime }\beta _{0}$, under which $F\left( W_{i}\right) $
follows a uniform distribution and several assumptions of Theorem \ref%
{Thm:lite} have a simpler form. For example, Assumption 3.(iii) will be
automatically fulfilled in this case as long as $W_{i}$ is continuously
distributed, and Assumption 3 is implied by

\bigskip

\noindent\textbf{Assumption 3'.} $W_{i}$ is continuously distributed with
right-unbounded support. In addition, $W_{i}$'s PDF $f_{W}\left( \cdot
\right) $ is absolutely continuous with a derivative that is bounded almost
everywhere.

\bigskip

Assumption 3' presumes that the underlying data is continuous, as in
standard kernel methods. When encountering discrete regressors, the kernel
estimator can be adapted using the frequency-based method or the smoothing
method \citep{racine2004nonparametric} at the cost of more complicated
notations. Assumption 5.(iii) will also be automatically fulfilled for $%
F\left( \cdot \right) =F_{W}\left( \cdot \right) $ because in this case $\Pr
\left( F\left( W_{i}\right) >1-h_{n}\right) =h_{n}$ and%
\begin{equation}
Ek_{ni}^{r}=\int_{0}^{1}k^{r}\left( \frac{1-t}{h_{n}}\right)
dt=h_{n}\int_{0}^{1\left/ h_{n}\right. }k^{r}\left( s\right)
ds=h_{n}\int_{0}^{1}k^{r}\left( s\right) ds  \label{knir}
\end{equation}%
for any $r>0$. Therefore, Assumption 5 is simplified to $h_{n}\rightarrow 0$
and $nh_{n}\rightarrow \infty $ as $n\rightarrow \infty $.

Now, consider the asymptotic bias of $\tilde{\mu}$, $EU_{i}D_{i}k_{ni}\left/
ED_{i}k_{ni}\right. $, under $F\left( \cdot \right) =F_{W}\left( \cdot
\right) $. Define%
\begin{equation}
G\left( t\right) =E\left[ U_{i}\left| D_{i}=1,F_{W}\left( W_{i}\right)
=t\right. \right] =E\left[ U_{i}\left| \varepsilon _{i}<F_{W}^{-1}\left(
t\right) \right. \right]  \label{G(t)}
\end{equation}%
for $t\in \left( 0,1\right) $, and define $G\left( 1\right)
=\lim_{t\rightarrow 1}G\left( t\right) =0$.

\begin{lemm}
Let $F\left( \cdot \right) =F_{W}\left( \cdot \right) =\Pr \left( W_{i}\leq
\cdot \right) $. Suppose $G\left( t\right) $ is continuously differentiable
over $t\in \left( 1-\delta ,1\right] $ for a small $\delta >0$, with its
derivative function $g\left( t\right) $ being bounded above over $t\in
\left( 1-\delta ,1\right] $. Then, we have%
\begin{equation*}
\frac{EU_{i}D_{i}k_{ni}}{ED_{i}k_{ni}}=-\frac{\kappa _{1}g\left( 1\right) }{%
\kappa _{0}}h_{n}+o\left( h_{n}\right) ,
\end{equation*}%
where $\kappa _{r}=\int_{0}^{1}t^{r}k\left( t\right) dt$.
\end{lemm}

Since by a calculation%
\begin{equation*}
g\left( t\right) =\frac{dG\left( t\right) }{dt}=\frac{f_{\varepsilon }\left(
F_{W}^{-1}\left( t\right) \right) }{f_{W}\left( F_{W}^{-1}\left( t\right)
\right) }\cdot \frac{E\left[ U_{i}\left\vert \varepsilon
_{i}=F_{W}^{-1}\left( t\right) \right. \right] -G\left( t\right) }{%
F_{\varepsilon }\left( F_{W}^{-1}\left( t\right) \right) },
\end{equation*}%
it can be seen that the finiteness of $g\left( 1\right) =\lim_{t\rightarrow
1}g\left( t\right) $ assumed by the Lemma essentially requires the selection
index $W_{i}$ to have a heavier upper tail than the selection disturbance $%
\varepsilon _{i}$. The relatively heavy tail of $W_{i}$ is also required by $%
\hat{\mu}^{AS}$, as illustrated by the examples of \cite%
{andrews1998semiparametric}.

\begin{cor}
Let $F\left( \cdot \right) =F_{W}\left( \cdot \right) =\Pr \left( W_{i}\leq
\cdot \right) $. Suppose (i) Assumptions 1-5 hold, (ii) the assumption of
the Lemma holds, and (iii) $nh_{n}^{3}=O\left( 1\right) $. Then, the kernel
regression estimator $\tilde{\mu}$ defined in (\ref{LC})\ is consistent and
asymptotically normal:%
\begin{equation*}
\sqrt{nh_{n}}\left( \tilde{\mu}-\mu _{0}+\frac{\kappa _{1}g\left( 1\right) }{%
\kappa _{0}}h_{n}\right) \rightarrow N\left( 0,\frac{\chi _{0}\sigma _{U}^{2}%
}{\kappa _{0}^{2}}\right) ,
\end{equation*}%
where $\chi _{0}=\int_{0}^{1}k^{2}\left( t\right) dt$.
\end{cor}

The corollary shows that, by setting $F\left( \cdot \right) =F_{W}\left(
\cdot \right) $, the asymptotic bias and variance of the kernel estimator
become explicit functions of $h_{n}$, based on which we can select a
theoretically optimal bandwidth. In practice, however, $F_{W}\left( \cdot
\right) $ is unknown and must be estimated beforehand. A simple estimator
for $F_{W}\left( \cdot \right) $ is the empirical CDF of $\hat{W}%
_{i}=X_{i}^{\prime }\hat{\beta}$. With this specific choice, the kernel
estimator becomes 
\begin{equation}
\hat{\mu}=\frac{\displaystyle\sum_{i=1}^{n}\left( Y_{i}-Z_{i}^{\prime }\hat{%
\theta}\right) D_{i}k\left( \frac{1-\hat{F}_{n}\left( \hat{W}_{i}\right) }{%
h_{n}}\right) }{\displaystyle\sum_{i=1}^{n}D_{i}k\left( \frac{1-\hat{F}%
_{n}\left( \hat{W}_{i}\right) }{h_{n}}\right) },  \label{FLC}
\end{equation}%
where $\hat{F}_{n}\left( \hat{W}_{i}\right) =\frac{1}{n-1}\sum_{j\neq
i}1\left\{ \hat{W}_{j}\leq \hat{W}_{i}\right\} $. To control the estimation
error induced by the empirical CDF, the existing assumptions should be
strengthened, and one new assumption should be imposed. Partition $%
X_{i}=\left( X_{i1},X_{i,(-1)}^{\prime }\right) ^{\prime }$, $\beta
_{0}=\left( \beta _{01},\beta _{0,(-1)}^{\prime }\right) ^{\prime }$, and $%
\hat{\beta}=\left( \hat{\beta}_{1},\hat{\beta}_{(-1)}^{\prime }\right)
^{\prime }$, with $X_{i1}$, $\beta _{01}$, and $\hat{\beta}_{1}$\ being the
respective first components.

\bigskip

\noindent\textbf{Assumption 1'.} Assumption 1 holds and $E\left\|
X_{i}\right\| ^{6}<\infty $; $\beta _{01}\neq 0$. Without loss of
generality, set $\beta _{01}=1$ to achieve scale normalization of $\beta
_{0} $.

\noindent \textbf{Assumption 2'.} Assumption 2 holds and $k\left( \cdot
\right) $ is six times continuously differentiable over $\left[ 0,\infty
\right) $ with all of its derivatives bounded above.

\noindent\textbf{Assumption 4'.} Assumption 4 holds and $\hat{\beta}_{1}=1$.

\noindent \textbf{Assumption 5'.} The bandwidth $h_{n}\in \left( 0,1/2\right]
$ satisfies (i) $h_{n}\rightarrow 0$, (ii) $nh_{n}^{3}=O\left( 1\right) $,
and (iii) $nh_{n}^{13/5}\rightarrow \infty $, as $n\rightarrow \infty $.

\noindent \textbf{Assumption 6.} The PDF of $W_{i}$, $f_{W}\left( \cdot
\right) $, and the conditional PDF of $W_{i}$ given $X_{i,(-1)}$, $%
f_{W\left| X_{(-1)}\right. }\left( \cdot \left| X_{i,(-1)}\right. \right) $,
satisfies that (i) $\left. f_{W}\left( F_{W}^{-1}\left( 1-h_{n}\right)
-n^{-1/10}h_{n}^{-1/18}\right) \right/ h_{n}^{2/3}\rightarrow 0$, (ii) $%
f_{W\left| X_{(-1)}\right. }\left( \cdot \left| X_{i,(-1)}\right. \right) $
is uniformly bounded by $\overline{f_{W\left| X_{(-1)}\right. }}$, and (iii)
there exists a large constant $C$ such that $f_{W\left| X_{(-1)}\right.
}\left( w\left| X_{i,(-1)}\right. \right) $ declines monotonically for $w>C$
almost surely.

\bigskip

Since $\beta _{0}$ is identified only up to scale, the normalization $\beta
_{01}=\hat{\beta}_{1}=1$ is postulated by most semiparametric estimators for
the binary response model. Assumption 6, which is comparable to Assumption A
of \cite{schafgans2002intercept}, is imposed to address the
non-differentiable indicator function contained in the empirical CDF.
Because the expectation of the generalized derivative of the indicator
function equals the value of the PDF, Assumption 6 involves restrictions on
the PDF and conditional PDF of $W_{i}$. Assumption 6.(i) relates to the
upper tail behavior of the distribution of $W_{i}$. Under the assumption of $%
nh_{n}^{13/5}\rightarrow \infty $ that implies $n^{-1/10}h_{n}^{-1/18}%
\rightarrow 0$, it can be shown that Assumption 6.(i) is satisfied if $W_{i}$
has a power- or exponential-type upper tail.

\begin{theorem}
\label{Thm:FLC}Under Assumptions 1'-5' and 6, the kernel regression
estimator $\hat{\mu}$ defined in (\ref{FLC}) is consistent and
asymptotically normal:%
\begin{equation*}
\sqrt{nh_{n}}\left( \hat{\mu}-\mu _{0}+\frac{\kappa _{1}g\left( 1\right) }{%
\kappa _{0}}h_{n}\right) \rightarrow N\left( 0,\frac{\chi _{0}\sigma _{U}^{2}%
}{\kappa _{0}^{2}}\right) ,
\end{equation*}%
where $\kappa _{r}=\int_{0}^{1}t^{r}k\left( t\right) dt$ and $\chi
_{r}=\int_{0}^{1}t^{r}k^{2}\left( t\right) dt$.
\end{theorem}

Theorem \ref{Thm:FLC} reveals that, under slightly stronger conditions, the
estimation error induced by the empirical CDF is asymptotically negligible;
thus, the asymptotic distribution of $\hat{\mu}$ is the same as if the true
CDF $F_{W}\left( \cdot \right) $ were known.

\subsection{Bandwidth selection}

An important implication of Theorem \ref{Thm:FLC} is that, under the
specific choice of $F\left( \cdot \right) =F_{W}\left( \cdot \right) $, the
kernel estimator for the intercept follows a standard asymptotic
distribution as an ordinary kernel regression estimator at the boundary
point. As a result, we can borrow approaches of bandwidth selection from the
nonparametric regression literature \citep[see, e.g.,][]{li2007nonparametric}%
. However, the widely used cross-validation method based on the integrated
mean squared error (MSE) criteria takes into account the global performance
of the regression function estimation and is thus not suitable for the
problem at hand. A closely related problem is the choice of bandwidth for
the nonparametric regression discontinuity estimator, which is the
difference between two regression estimators evaluated at boundary points. I
follow the plug-in method proposed by \cite{imbens2012optimal} and suggest a
data-dependent bandwidth selection procedure that is tailored to $\hat{\mu}$.

By Theorem \ref{Thm:FLC}, we know that%
\begin{equation*}
\text{Abias}\left( \hat{\mu}\right) =\left( -\frac{\kappa _{1}g\left(
1\right) }{\kappa _{0}}\right) h_{n},\text{ \ Avar}\left( \hat{\mu}\right)
=\left( \frac{\chi _{0}\sigma _{U}^{2}}{\kappa _{0}^{2}}\right) \frac{1}{%
nh_{n}},
\end{equation*}%
where $g\left( 1\right) =\left. \left. dG\left( t\right) \right/
dt\right\vert _{t=1}$ with $G\left( t\right) $ defined in (\ref{G(t)}), and $%
\sigma _{U}^{2}=Var\left( U_{i}\right) $. Provided $g\left( 1\right) \neq 0$%
, the optimal bandwidth for $\hat{\mu}$ is defined as a minimizer of its
asymptotic MSE:%
\begin{eqnarray*}
h_{opt} &=&\arg \min_{h}\left\{ \text{AMSE}_{h}\left( \hat{\mu}\right)
=\left( \frac{\kappa _{1}g\left( 1\right) }{\kappa _{0}}\right)
^{2}h^{2}+\left( \frac{\chi _{0}\sigma _{U}^{2}}{\kappa _{0}^{2}}\right) 
\frac{1}{nh}\right\} \\
&=&c_{k}\left( \sigma _{U}^{2}\left/ g^{2}\left( 1\right) \right. \right)
^{1/3}n^{-1/3},
\end{eqnarray*}%
where $c_{k}=\left( \chi _{0}\left/ \left( 2\kappa _{1}^{2}\right) \right.
\right) ^{1/3}$ is a functional of the kernel $k\left( \cdot \right) $. If $%
g\left( 1\right) =0$, the bias converges to zero faster, allowing for
estimation of the intercept at a faster rate of convergence. However, it is
difficult to exploit the improved convergence rate resulting from this
condition in practice; hence, I focus on the optimal bandwidth given $%
g\left( 1\right) \neq 0$.

A natural choice of the estimator for the optimal bandwidth $h_{opt}$ is to
replace $\sigma _{U}^{2}$ and $g\left( 1\right) $ with their consistent
estimators $\hat{\sigma}_{U}^{2}$ and $\hat{g}\left( 1\right) $,
respectively. One potential problem with this choice is that $\hat{g}\left(
1\right) $ may occasionally be very close to zero due to the stochastic
estimation error, even if $g\left( 1\right) \neq 0$. In such cases, the
estimated bandwidth will be imprecisely and unstably large, which may in
turn lead to large finite sample bias of $\hat{\mu}$ because observations
that are far from the boundary will be included in the kernel estimation. To
alleviate this problem, I employ the regularization method 
\citep[Subsection
4.1.1]{imbens2012optimal} and propose the following bandwidth estimator: 
\begin{equation}
\hat{h}_{opt}=c_{k}\left( \frac{\hat{\sigma}_{U}^{2}}{\hat{g}^{2}\left(
1\right) +3\widehat{Var}\left( \hat{g}\left( 1\right) \right) }\right)
^{1/3}n^{-1/3}.  \label{hat_h}
\end{equation}%
By means of regularization, $\hat{h}_{opt}$ will not become infinite even in
the case of $\hat{g}\left( 1\right) =0$. Moreover, the leading term of $E%
\left[ 1\left/ \left( \hat{g}^{2}\left( 1\right) +3Var\left( \hat{g}\left(
1\right) \right) \right) \right. -1\left/ \hat{g}^{2}\left( 1\right) \right. %
\right] $ cancels out the leading term of $E\left[ 1\left/ \hat{g}^{2}\left(
1\right) \right. -1\left/ g^{2}\left( 1\right) \right. \right] $. Therefore,
the bias of $1\left/ \left( \hat{g}^{2}\left( 1\right) +3Var\left( \hat{g}%
\left( 1\right) \right) \right) \right. $ for the reciprocal of $g^{2}\left(
1\right) $ is of lower order than the bias of the naive $1\left/ \hat{g}%
^{2}\left( 1\right) \right. $. It remains to construct consistent estimators
for the components of the plug-in bandwidth, namely, $\hat{\sigma}_{U}^{2}$, 
$\hat{g}^{2}\left( 1\right) $, and $\widehat{Var}\left( \hat{g}\left(
1\right) \right) $, which is deferred to the next section for expositional
convenience.

\section{Local linear estimation}

\label{sec:LL}Theorem \ref{Thm:FLC} finds that the kernel regression
estimator $\hat{\mu}$ for the intercept suffers from a boundary effect in
the sense that its order of bias, $O\left( h_{n}\right) $, is larger than
that of the nonparametric regression estimator in the interior, which is
typically $O\left( h_{n}^{2}\right) $. This is simply because $\hat{\mu}$ is
in nature a Nadaraya-Watson or local constant estimator at the boundary. In
this section, I resort to the local polynomial regression method %
\citep{fan1996local} for bias reduction. In particular, I focus on the local
linear estimation because of its asymptotic minimax efficiency properties %
\citep{cheng1997automatic} and attractive practical performance %
\citep{gelman2019why}.

Denote $\hat{F}_{n}\left( \hat{W}_{i}\right) =\frac{1}{n-1}\sum_{j\neq
i}1\left\{ \hat{W}_{j}\leq \hat{W}_{i}\right\} $ as the empirical estimate
of $F_{W}\left( W_{i}\right) $ as before. The local linear estimator $\hat{%
\mu}^{L}$ for the intercept is defined via locally weighted least squares
regression:%
\begin{equation}
\left( \hat{\mu}^{L},\hat{b}^{L}\right) =\arg \min_{\mu ,b}\sum_{i=1}^{n}%
\left[ Y_{i}-Z_{i}^{\prime }\hat{\theta}-\mu -\left( \hat{F}_{n}\left( \hat{W%
}_{i}\right) -1\right) b\right] ^{2}D_{i}k\left( \frac{1-\hat{F}_{n}\left( 
\hat{W}_{i}\right) }{h_{n}}\right) .  \label{FLL}
\end{equation}%
To establish the asymptotic properties of $\hat{\mu}^{L}$, Assumption 5'
must be modified to accommodate the local linear case.

\bigskip

\noindent\textbf{Assumption 5''.} The bandwidth $h_{n}\in \left( 0,1/2\right]
$ satisfies (i) $h_{n}\rightarrow 0$, (ii) $nh_{n}^{5}=O\left( 1\right) $,
and (iii) $nh_{n}^{3}\rightarrow \infty $, as $n\rightarrow \infty $.

\begin{theorem}
\label{Thm:FLL}Suppose Assumptions 1'-4', 5\textquotedblright , and 6 hold.
In addition, suppose $G\left( t\right) $ given in (\ref{G(t)}) is twice
continuously differentiable over $t\in \left( 1-\delta ,1\right] $ for a
small $\delta >0$, with its first and second derivative functions $g\left(
t\right) $ and $g^{\prime }\left( t\right) $ being bounded above over $t\in
\left( 1-\delta ,1\right] $. Then, the local linear estimator defined in (%
\ref{FLL}) is consistent for $\left( \mu _{0},g\left( 1\right) \right) $ and
asymptotically normal:%
\begin{equation*}
\left( 
\begin{array}{cc}
\sqrt{nh_{n}} &  \\ 
& \sqrt{nh_{n}^{3}}%
\end{array}%
\right) \left[ \left( 
\begin{array}{c}
\hat{\mu}^{L} \\ 
\hat{b}^{L}%
\end{array}%
\right) -\left( 
\begin{array}{c}
\mu _{0} \\ 
g\left( 1\right)%
\end{array}%
\right) +\left( 
\begin{array}{c}
\frac{\left( \kappa _{1}\kappa _{3}-\kappa _{2}^{2}\right) g^{\prime }\left(
1\right) }{2\left( \kappa _{0}\kappa _{2}-\kappa _{1}^{2}\right) }h_{n}^{2}
\\ 
\frac{\left( \kappa _{0}\kappa _{3}-\kappa _{1}\kappa _{2}\right) g^{\prime
}\left( 1\right) }{2\left( \kappa _{0}\kappa _{2}-\kappa _{1}^{2}\right) }%
h_{n}%
\end{array}%
\right) \right] \rightarrow N\left( 0,\sigma _{U}^{2}\Omega ^{L}\right) ,
\end{equation*}%
where $\sigma _{U}^{2}=EU_{i}^{2}$ and%
\begin{equation*}
\Omega ^{L}=\frac{1}{\left( \kappa _{0}\kappa _{2}-\kappa _{1}^{2}\right)
^{2}}\left( 
\begin{array}{cc}
\kappa _{2}^{2}\chi _{0}+\kappa _{1}^{2}\chi _{2}-2\kappa _{1}\kappa
_{2}\chi _{1} & \kappa _{1}\kappa _{2}\chi _{0}+\kappa _{0}\kappa _{1}\chi
_{2}-\left( \kappa _{0}\kappa _{2}+\kappa _{1}^{2}\right) \chi _{1} \\ 
\kappa _{1}\kappa _{2}\chi _{0}+\kappa _{0}\kappa _{1}\chi _{2}-\left(
\kappa _{0}\kappa _{2}+\kappa _{1}^{2}\right) \chi _{1} & \kappa
_{1}^{2}\chi _{0}+\kappa _{0}^{2}\chi _{2}-2\kappa _{0}\kappa _{1}\chi _{1}%
\end{array}%
\right) .
\end{equation*}
\end{theorem}

\medskip

Theorem \ref{Thm:FLL} shows that the local linear estimation of the
intercept automatically eliminates the boundary effect, and its asymptotic
bias is of the same order as that in the interior. More interestingly, the
local linear procedure generates a consistent estimate of $g\left( 1\right) $
as a byproduct, which is a key ingredient of the optimal bandwidth for the
kernel estimator $\hat{\mu}$. There would be no technical difficulty in
extending the results of Theorem \ref{Thm:FLL} to local polynomial
estimation with higher order, except for more complicated notation and more
involved mathematical derivations.

\subsection{Bandwidth selection}

Provided $g^{\prime }\left( 1\right) \neq 0$, the optimal bandwidth for $%
\hat{\mu}^{L}$ is analogously defined by minimizing the asymptotic MSE of $%
\hat{\mu}^{L}$:%
\begin{eqnarray*}
h_{opt}^{L} &=&\arg \min_{h}\left\{ \text{AMSE}_{h}\left( \hat{\mu}%
^{L}\right) =\left[ \frac{\left( \kappa _{1}\kappa _{3}-\kappa
_{2}^{2}\right) g^{\prime }\left( 1\right) }{2\left( \kappa _{0}\kappa
_{2}-\kappa _{1}^{2}\right) }\right] ^{2}h^{4}+\frac{\left( \kappa
_{2}^{2}\chi _{0}+\kappa _{1}^{2}\chi _{2}-2\kappa _{1}\kappa _{2}\chi
_{1}\right) \sigma _{U}^{2}}{\left( \kappa _{0}\kappa _{2}-\kappa
_{1}^{2}\right) ^{2}}\frac{1}{nh}\right\} \\
&=&c_{k}^{L}\left( \sigma _{U}^{2}\left/ \left[ g^{\prime }\left( 1\right) %
\right] ^{2}\right. \right) ^{1/5}n^{-1/5},
\end{eqnarray*}%
where%
\begin{equation}
c_{k}^{L}=\left[ \frac{\kappa _{2}^{2}\chi _{0}+\kappa _{1}^{2}\chi
_{2}-2\kappa _{1}\kappa _{2}\chi _{1}}{\left( \kappa _{1}\kappa _{3}-\kappa
_{2}^{2}\right) ^{2}}\right] ^{1/5}.  \label{cL_k}
\end{equation}%
As in Subsection \ref{sec:KRE}.3, I adopt the regularization method and
propose the following estimator for $h_{opt}^{L}$: 
\begin{equation}
\hat{h}_{opt}^{L}=c_{k}^{L}\left( \frac{\hat{\sigma}_{U}^{2}}{\left[ \hat{g}%
^{\prime }\left( 1\right) \right] ^{2}+3\widehat{Var}\left( \hat{g}^{\prime
}\left( 1\right) \right) }\right) ^{1/5}n^{-1/5}.  \label{hat_h_L}
\end{equation}

To implement the plug-in bandwidth selection procedures (\ref{hat_h}) and (%
\ref{hat_h_L}), one must estimate the limits of the derivative functions, $%
g\left( 1\right) $ and $g^{\prime }\left( 1\right) $, the regularization
terms, $Var\left( \hat{g}\left( 1\right) \right) $ and $Var\left( \hat{g}%
^{\prime }\left( 1\right) \right) $, and the variance of the disturbance, $%
\sigma _{U}^{2}$. I first construct $\hat{g}\left( 1\right) $ and $\hat{g}%
^{\prime }\left( 1\right) $ by fitting a quadratic function to the
observations near the boundary:%
\begin{equation}
\left( \hat{\mu}^{Q},\hat{g}\left( 1\right) ,\hat{g}^{\prime }\left(
1\right) \right) =\arg \min_{b_{0},b_{1},b_{2}}\sum_{i=1}^{n}\left[
Y_{i}-Z_{i}^{\prime }\hat{\theta}-\sum_{r=0}^{2}\frac{\left( \hat{F}%
_{n}\left( \hat{W}_{i}\right) -1\right) ^{r}}{r!}b_{r}\right]
^{2}D_{i}k\left( \frac{1-\hat{F}_{n}\left( \hat{W}_{i}\right) }{h_{1n}}%
\right) ,  \label{LQE}
\end{equation}%
where $h_{1n}$ is a pilot bandwidth. Similar to Theorem \ref{Thm:FLL},\ one
can show that under $nh_{1n}^{7}=O\left( 1\right) $, $nh_{1n}^{5}\rightarrow
\infty $, and the regularity conditions, the local quadratic estimator is
consistent and asymptotically normal:%
\begin{equation*}
\left( 
\begin{array}{ccc}
\sqrt{nh_{1n}} &  &  \\ 
& \sqrt{nh_{1n}^{3}} &  \\ 
&  & \sqrt{nh_{1n}^{5}}%
\end{array}%
\right) \left[ \left( 
\begin{array}{c}
\hat{\mu}^{Q} \\ 
\hat{g}\left( 1\right) \\ 
\hat{g}^{\prime }\left( 1\right)%
\end{array}%
\right) -\left( 
\begin{array}{c}
\mu _{0} \\ 
g\left( 1\right) \\ 
g^{\prime }\left( 1\right)%
\end{array}%
\right) +\left( 
\begin{array}{c}
B_{1}^{Q}g^{\prime \prime }\left( 1\right) h_{1n}^{3} \\ 
B_{2}^{Q}g^{\prime \prime }\left( 1\right) h_{1n}^{2} \\ 
B_{3}^{Q}g^{\prime \prime }\left( 1\right) h_{1n}%
\end{array}%
\right) \right] \rightarrow N\left( 0,\sigma _{U}^{2}\Omega ^{Q}\right) .
\end{equation*}%
As a result, the regularization terms can be estimated by%
\begin{equation*}
\widehat{Var}\left( \hat{g}\left( 1\right) \right) =\frac{1}{nh_{1n}^{3}}%
\hat{\sigma}_{U}^{2}\Omega _{22}^{Q},\text{ \ }\widehat{Var}\left( \hat{g}%
^{\prime }\left( 1\right) \right) =\frac{1}{nh_{1n}^{5}}\hat{\sigma}%
_{U}^{2}\Omega _{33}^{Q},
\end{equation*}%
where%
\begin{eqnarray}
\Omega _{22}^{Q} &=&\frac{\left\{ 
\begin{array}{c}
\left( \kappa _{1}\kappa _{4}-\kappa _{2}\kappa _{3}\right) ^{2}\chi
_{0}-2\left( \kappa _{1}\kappa _{4}-\kappa _{2}\kappa _{3}\right) \left(
\kappa _{0}\kappa _{4}-\kappa _{2}^{2}\right) \chi _{1} \\ 
+\left[ \left( \kappa _{0}\kappa _{4}-\kappa _{2}^{2}\right) ^{2}+2\left(
\kappa _{1}\kappa _{4}-\kappa _{2}\kappa _{3}\right) \left( \kappa
_{0}\kappa _{3}-\kappa _{1}\kappa _{2}\right) \right] \chi _{2} \\ 
-\left( \kappa _{0}\kappa _{4}-\kappa _{2}^{2}\right) \left( \kappa
_{0}\kappa _{3}-\kappa _{1}\kappa _{2}\right) \chi _{3}+\left( \kappa
_{0}\kappa _{3}-\kappa _{1}\kappa _{2}\right) ^{2}\chi _{4}%
\end{array}%
\right\} }{\left( \kappa _{0}\kappa _{2}\kappa _{4}-\kappa _{0}\kappa
_{3}^{2}-\kappa _{1}^{2}\kappa _{4}+2\kappa _{1}\kappa _{2}\kappa
_{3}-\kappa _{2}^{3}\right) ^{2}},  \label{omigaQ_22} \\
\Omega _{33}^{Q} &=&\frac{\left\{ 
\begin{array}{c}
\left( \kappa _{1}\kappa _{3}-\kappa _{2}^{2}\right) ^{2}\chi _{0}-2\left(
\kappa _{1}\kappa _{3}-\kappa _{2}^{2}\right) \left( \kappa _{0}\kappa
_{3}-\kappa _{1}\kappa _{2}\right) \chi _{1} \\ 
+\left[ \left( \kappa _{0}\kappa _{3}-\kappa _{1}\kappa _{2}\right)
^{2}+2\left( \kappa _{1}\kappa _{3}-\kappa _{2}^{2}\right) \left( \kappa
_{0}\kappa _{2}-\kappa _{1}^{2}\right) \right] \chi _{2} \\ 
-\left( \kappa _{0}\kappa _{3}-\kappa _{1}\kappa _{2}\right) \left( \kappa
_{0}\kappa _{2}-\kappa _{1}^{2}\right) \chi _{3}+\left( \kappa _{0}\kappa
_{2}-\kappa _{1}^{2}\right) ^{2}\chi _{4}%
\end{array}%
\right\} }{\left( \kappa _{0}\kappa _{2}\kappa _{4}-\kappa _{0}\kappa
_{3}^{2}-\kappa _{1}^{2}\kappa _{4}+2\kappa _{1}\kappa _{2}\kappa
_{3}-\kappa _{2}^{3}\right) ^{2}}.  \label{omigaQ_33}
\end{eqnarray}%
Last, I estimate the variance of the disturbance by 
\begin{equation}
\hat{\sigma}_{U}^{2}=\frac{\displaystyle\sum_{i=1}^{n}\left( Y_{i}-\hat{\mu}%
^{Q}-Z_{i}^{\prime }\hat{\theta}\right) ^{2}D_{i}k\left( \frac{1-\hat{F}%
_{n}\left( \hat{W}_{i}\right) }{h_{2n}}\right) }{\displaystyle%
\sum_{i=1}^{n}D_{i}k\left( \frac{1-\hat{F}_{n}\left( \hat{W}_{i}\right) }{%
h_{2n}}\right) },  \label{sigma2_hat}
\end{equation}%
where $\hat{\mu}^{Q}$ is the initial estimate of $\mu _{0}$ given in (\ref%
{LQE}) and $h_{2n}$ is another pilot bandwidth. Following Theorem \ref%
{Thm:FLC}, it can be shown that%
\begin{equation*}
\hat{\sigma}_{U}^{2}\overset{p}{\rightarrow }E\left[ \left. \left( Y_{i}-\mu
_{0}-Z_{i}^{\prime }\theta _{0}\right) ^{2}\right\vert D_{i}=1,F_{W}\left(
W_{i}\right) =1\right] =EU_{i}^{2}=\sigma _{U}^{2}.
\end{equation*}%
For the pilot bandwidths, simply setting%
\begin{equation*}
h_{1n}=n^{-1/7},\text{ }h_{2n}=n^{-1/3}
\end{equation*}%
is sufficient to ensure the consistency of $\hat{h}_{opt}$ for $h_{opt}$ and
of $\hat{h}_{opt}^{L}$ for $h_{opt}^{L}$. In practice, the suggested
bandwidth selection algorithm is fairly robust to the choice of pilot
bandwidth, which is not surprising given the presence of the power $1/3$ or $%
1/5$ in the expressions for the optimal bandwidths.

\section{Simulation}

\label{sec:simulation}This section examines the finite sample properties of
the kernel regression estimator $\hat{\mu}$ defined in (\ref{FLC}) and the
local linear estimator $\hat{\mu}^{L}$ defined in (\ref{FLL}), in comparison
with the parametric two-step estimator \citep{heckman1979sample} and the
semiparametric identification-at-infinity estimators %
\citep{heckman1990varieties, andrews1998semiparametric}. The parametric
two-step procedure implements probit estimation for the selection equation
in the first step and least squares estimation for the outcome equation with
a correction term using the uncensored observations in the second step. This
approach is commonly applied in empirical studies due to its computational
ease. However, it is likely to be inconsistent when the true distribution of
the model disturbance is nonnormal. In contrast, the consistency of \cite%
{heckman1990varieties}'s estimator (henceforth the Heckman estimator) and 
\cite{andrews1998semiparametric}'s estimator (henceforth the AS estimator)
does not rely on parametric specification of the disturbance distribution,
but it is difficult to choose an appropriate smoothing parameter for these
estimators. To investigate the robustness of their practical performance, a
wide range of smoothing parameters is considered. Following the literature,
the choices considered are based on the percentage of uncensored
observations used in the estimation. Specifically, the smoothing parameter
takes the values of various quantiles of the selection linear index in the
uncensored subsample.

The simulation setting mainly follows \cite{schafgans2004finite}, and the
data generating process is%
\begin{eqnarray*}
D_{i} &=&1\left\{ c_{0}+X_{1i}+X_{2i}>\varepsilon _{i}\right\} , \\
Y_{i}^{\ast } &=&\mu _{0}+U_{i}, \\
Y_{i} &=&Y_{i}^{\ast }D_{i},\text{ }i=1,2,\cdots ,n,
\end{eqnarray*}%
where $\mu _{0}=0$, $X_{1i}$ follows the standard normal distribution, $%
X_{2i}$ follows the standardized (zero-mean, unit-variance) Student's $t$
distribution with three degrees of freedom, $\varepsilon _{i}$ and $U_{i}$
are zero-mean random variables described below, and only $\left(
Y_{i},D_{i},X_{1i},X_{2i}\right) $ is observed. In the simulation, the
outcome equation does not contain any nonconstant regressors, implying that
the intercept $\mu _{0}$ of primary concern represents the population mean
of the latent outcome $Y_{i}^{\ast }$. Different designs are constructed by
varying the disturbance distribution and the value of $c_{0}$. $\varepsilon
_{i}$ follows three different distributions, namely, the standard normal
distribution, the standardized Student's $t$ distribution with three degrees
of freedom, and the standardized chi-square distribution with three degrees
of freedom. $U_{i}$ is generated by $U_{i}=\varepsilon _{i}+e_{i}$, where $%
e_{i}$ is a standard normal random variable independent of $\varepsilon _{i}$%
. The constant $c_{0}$ controls for the amount of censoring. In the
benchmark design, $c_{0}=0$, producing approximately 50\% censoring.
Different values of $c_{0}$ are chosen so that $\Pr \left(
c_{0}+X_{1i}+X_{2i}>\varepsilon _{i}\right) $ is equal to 0.8 and 0.2,
corresponding to 20\% and 80\% proportions of zero observations,
respectively. The sample size $n$ is set to 250, 1000, 4000, and the
simulation is replicated 1000 times for each design.

Before calculating the semiparametric estimators for $\mu _{0}$, a
distribution-free estimate for the selection equation is necessary. Since
the simulation results of \cite{schafgans2004finite} show little sensitivity
to the particular choice of this estimate, I employ the average derivative
estimation \citep{powell1989semiparametric} for computational convenience.
When implementing the Heckman and AS estimators, the smoothing parameter is
equal to the 0.99, 0.95, 0.9, 0.8, 0.7, and 0.5 quantiles of $\hat{W}_{i}=%
\hat{\beta}_{1}X_{1i}+\hat{\beta}_{2}X_{2i}$ in the uncensored subsample,
corresponding to 1\%, 5\%, 10\%, 20\%, 30\%, and 50\% uncensored
observations used in the estimation. The value of the smoothing parameter
declines with the proportion of uncensored observations. For consistency,
the smoothing parameter is required to approach infinity as $n$ goes to
infinity such that the estimation is based on only the observations for
which $\Pr \left( \left. D_{i}=1\right\vert X_{i}\right) $ is close to one
and in the limit is equal to one. Following the suggestion of \cite%
{andrews1998semiparametric}, the smoothed function in the AS estimator is%
\begin{equation*}
s\left( w\right) =\left\{ 
\begin{array}{c}
0 \\ 
1-\exp \left( -\frac{w}{b-w}\right) \\ 
1%
\end{array}%
\begin{array}{l}
\text{for }w\leq 0, \\ 
\text{for }0<w\leq b, \\ 
\text{for }w>b,%
\end{array}%
\right.
\end{equation*}%
where $b$ is set equal to 1. Note that the AS estimator with $b=0$ is
equivalent to the Heckman estimator. For the proposed kernel-type
estimators, the bandwidth is chosen by the plug-in algorithm given in
Subsection \ref{sec:LL}.1. The commonly used Gaussian and Epanechnikov
kernel functions are applied. However, the Gaussian kernel is not compactly
supported, and the Epanechnikov kernel is not smooth. Therefore, I also
consider the polynomial and polyweight kernels, both of degree seven, which
possess sufficient smoothness required by Assumption 2'. The definition of
these kernel functions and their relevant functionals are presented in Table %
\ref{table:kernel}.

\begin{sidewaystable}[ptb]
\caption{Several kernel functions and their relevant functionals}
\bigskip
\label{table:kernel}
\renewcommand\arraystretch{2}
\resizebox{\linewidth}{!}{
\begin{tabular}{C{2.2cm}ccc}
\hline\hline
Kernel & Definition & $\kappa _{r}=\int_{0}^{\infty }t^{r}k\left( t\right) dt
$, $r\in \mathbb{N}$ & $\chi _{r}=\int_{0}^{\infty }t^{r}k^{2}\left(
t\right) dt$, $r\in \mathbb{N}$ \\ \hline
Gaussian & $\displaystyle k\left( t\right) =\frac{1}{\sqrt{2\pi }}\exp
\left( -\frac{t^{2}}{2}\right)1\left\{t\geq 0\right\} $ & $\displaystyle\sqrt{\frac{2^{r-2}}{\pi }}%
\gamma \left( \frac{r+1}{2}\right) $ & $\displaystyle\frac{1}{4\pi }\gamma
\left( \frac{r+1}{2}\right) $ \\ 
Epanechnikov & $\displaystyle k\left( t\right) =\left\{ \QATOP{1-t^{2}\text{
if }0\leq t\leq 1}{0\text{ \ if }t>1}\right. $ & $\displaystyle\frac{2}{%
\left( r+1\right) \left( r+3\right) }$ & $\displaystyle\frac{8}{\left(
r+1\right) \left( r+3\right) \left( r+5\right) }$ \\ 
Polynomial of degree 7 & $\displaystyle k\left( t\right) =\left\{ \QATOP{%
\left( 1-t\right) ^{7}\text{ if }0\leq t\leq 1}{\text{ \ }0\text{\ \ \ if }%
t>1}\right. $ & $\displaystyle B\left( r+1,8\right) =\frac{7!r!}{\left(
r+8\right) !}$ & $\displaystyle B\left( r+1,15\right) =\frac{14!r!}{\left(
r+15\right) !}$ \\ 
Polyweight of degree 7 & $\displaystyle k\left( t\right) =\left\{ \QATOP{%
\left( 1-t^{2}\right) ^{7}\text{ if }0\leq t\leq 1}{\text{ \ }0\text{\ \ \ \
if }t>1}\right. $ & $\displaystyle\frac{1}{2}B\left( \frac{r+1}{2},8\right) =%
\frac{2^{7}7!\left( r-1\right) !!}{\left( r+15\right) !!}$ & $\displaystyle%
\frac{1}{2}B\left( \frac{r+1}{2},15\right) =\frac{2^{14}14!\left( r-1\right)
!!}{\left( r+29\right) !!}$ \\ \hline
Kernel & $c_{k}^{L}$ in (\ref{cL_k}) & $\Omega _{22}^{Q}$ in (\ref{omigaQ_22}%
) & $\Omega _{33}^{Q}$ in (\ref{omigaQ_33}) \\ \hline
Gaussian & $\left[ \frac{\left( \pi +1-2\sqrt{2}\right) \sqrt{\pi }}{\left(
4-\pi \right) ^{2}}\right] ^{1/5}\approx \allowbreak 1.\allowbreak 259$ & $%
\frac{\left( 4\pi +11-12\sqrt{2}\right) \sqrt{\pi }}{8\left( \pi -3\right)
^{2}}\approx \allowbreak 72.\allowbreak 89$ & $\frac{3\pi ^{2}-\left( 16-4%
\sqrt{2}\right) \pi +\left( 44-24\sqrt{2}\right) }{16\sqrt{\pi }\left( \pi
-3\right) ^{2}}\approx \allowbreak 12.\allowbreak 62$ \\ 
Epanechnikov & $3.200$ & $4913.0$ & $6327.8$ \\ 
Polynomial of degree 7 & $8.175$ & $8477.5$ & $48857.0$ \\ 
Polyweight of degree 7 & $5.396$ & $8645.0$ & $29139.5$ \\ \hline\hline
\end{tabular}
}
{\small {Note: $r!!$, $\gamma \left( r\right) $, and $B\left(
r_{1},r_{2}\right) $ denote the double factorial, gamma function, and beta
function, respectively.} }
\end{sidewaystable}

The summary statistics for the simulation are the estimators' bias, standard
deviation (SD), root mean squared error (RMSE) ratio, and rejection rate of
the t test, over 1000 replications. The RMSE can be calculated from the bias
and SD and is thus omitted from the tables. Instead, the RMSE ratio defined
by the RMSE of the semiparametric estimators over that of the parametric
two-step estimator is reported. Although the RMSE (ratio) is the primary
criterion used to compare consistent estimators, it is useless if there are
both consistent and inconsistent estimators. An inconsistent estimator
having a small RMSE implies that the estimator is highly concentrated in a
narrow interval centered at a biased value, consequently leading to
incorrect inference. As a complement to the RMSE (ratio), I consider the
simulated probability of rejecting the null hypothesis $H_{0}:\mu _{0}=0$
against $H_{1}:\mu _{0}\neq 0$ at a 5\% level of significance using t tests
(or, more precisely, z tests), based on the asymptotic variances given in 
\cite{heckman1979sample}, \cite{schafgans2002intercept}, \cite%
{andrews1998semiparametric}, and Theorems \ref{Thm:FLC}-\ref{Thm:FLL}.

Tables \ref{table:normal}-\ref{table:chi2} report the simulation results
when the model disturbance follows a normal distribution, a $t\left(
3\right) $ distribution that is symmetric but fat-tailed, and a $\chi
^{2}\left( 3\right) $ distribution that is skewed, respectively, under
approximately 50\% censoring. Table \ref{table:normal} shows that, under
normal disturbance, the parametric two-step estimator is asymptotically
unbiased, as expected, and converges at a $\sqrt{n}$ rate (its SD halves
when the sample size quadruples). In contrast, all the considered
semiparametric estimators converge at slower than $\sqrt{n}$ rates, as their
RMSE ratios increase with the sample size. For the Heckman and AS
estimators, the bias increases and the SD decreases when the proportion of
uncensored observations used in the estimation increases or, equivalently,
the smoothing parameter decreases. The optimal smoothing parameter in terms
of RMSE depends upon the estimator, the sample size, and, by comparing
across Tables \ref{table:normal}-\ref{table:chi2}, the disturbance
distribution. If the smoothing parameter is improperly chosen, the RMSE may
be several times larger than the smallest RMSE, and the rejection rate may
be far larger than the specified level of significance. Moreover, the
optimal smoothing parameter in terms of RMSE usually disagrees with that in
terms of the rejection rate. For instance, in the case of $n=1000$ for the
Heckman estimator and the case of $n=4000$ for the AS estimator, the optimal
smoothing parameter in terms of RMSE would use 30\% uncensored observations
in the estimation, but the corresponding simulated rejection rates are more
than twice the real level. In these two cases, a more sensible choice would
be to use 20\% uncensored observations, sacrificing a little RMSE but
leading to a rejection rate very close to 0.05. In practice, however, the
RMSE and rejection rate are not known; therefore, we are never aware of
whether we have made a good choice.

\begin{sidewaystable}[ptb]
\caption{Simulation results when $\varepsilon _{i}\sim N\left(0,1\right) $}
\medskip
\label{table:normal}
\resizebox{\linewidth}{!}{
\begin{tabular}{ccccC{1.2cm}C{1.5cm}cccC{1.2cm}C{1.5cm}cccC{1.2cm}C{1.5cm}}
\hline\hline
$\Pr \left( Y_{i}=0\right) =0.5$ &  & \multicolumn{4}{c}{$n=250$} &  & 
\multicolumn{4}{c}{$n=1000$} &  & \multicolumn{4}{c}{$n=4000$} \\ 
\cline{1-1}\cline{3-6}\cline{8-11}\cline{13-16}
Estimator &  & Bias & SD & RMSE ratio & Rejection rate &  & Bias & SD & RMSE
ratio & Rejection rate &  & Bias & SD & RMSE ratio & Rejection rate \\ \hline
\multicolumn{16}{l}{\cite{heckman1979sample}'s parametric two-step estimator}
\\ 
&  & -0.004 & 0.186 & 1 & 0.052 &  & -0.002 & 0.093 & 1 & 0.048 &  & 0.000 & 
0.044 & 1 & 0.043 \\ 
\multicolumn{16}{l}{\cite{heckman1990varieties}'s semiparametric estimator
with various smoothing parameters} \\ 
$1\%$ observations &  & -0.007 & 1.025 & 5.519 & 0.390 &  & 0.028 & 0.599 & 
6.444 & 0.128 &  & 0.006 & 0.301 & 6.812 & 0.062 \\ 
$5\%$ observations &  & -0.003 & 0.549 & 2.955 & 0.125 &  & 0.001 & 0.286 & 
3.074 & 0.076 &  & -0.006 & 0.140 & 3.173 & 0.056 \\ 
$10\%$ observations &  & -0.016 & 0.389 & 2.093 & 0.088 &  & -0.021 & 0.195
& 2.108 & 0.064 &  & -0.016 & 0.096 & 2.209 & 0.062 \\ 
$20\%$ observations &  & -0.046 & 0.275 & 1.503 & 0.070 &  & -0.049 & 0.136
& 1.558 & 0.059 &  & -0.044 & 0.067 & 1.801 & 0.088 \\ 
$30\%$ observations &  & -0.077 & 0.228 & 1.294 & 0.082 &  & -0.085 & 0.111
& 1.499 & 0.114 &  & -0.079 & 0.055 & 2.185 & 0.278 \\ 
$50\%$ observations &  & -0.168 & 0.175 & 1.306 & 0.178 &  & -0.164 & 0.084
& 1.982 & 0.518 &  & -0.164 & 0.041 & 3.821 & 0.983 \\ 
\multicolumn{16}{l}{\cite{andrews1998semiparametric}'s semiparametric
estimator with various smoothing parameters} \\ 
$1\%$ observations &  & -0.055 & 1.302 & 7.013 & 0.616 &  & 0.037 & 0.697 & 
7.498 & 0.144 &  & 0.010 & 0.341 & 7.710 & 0.065 \\ 
$5\%$ observations &  & -0.007 & 0.689 & 3.711 & 0.165 &  & 0.007 & 0.338 & 
3.632 & 0.077 &  & -0.005 & 0.165 & 3.737 & 0.048 \\ 
$10\%$ observations &  & -0.012 & 0.480 & 2.585 & 0.106 &  & -0.007 & 0.241
& 2.589 & 0.074 &  & -0.008 & 0.116 & 2.628 & 0.050 \\ 
$20\%$ observations &  & -0.027 & 0.331 & 1.787 & 0.076 &  & -0.029 & 0.163
& 1.775 & 0.064 &  & -0.024 & 0.079 & 1.857 & 0.051 \\ 
$30\%$ observations &  & -0.045 & 0.263 & 1.438 & 0.076 &  & -0.052 & 0.130
& 1.508 & 0.070 &  & -0.046 & 0.063 & 1.769 & 0.101 \\ 
$50\%$ observations &  & -0.106 & 0.198 & 1.211 & 0.111 &  & -0.110 & 0.096
& 1.569 & 0.214 &  & -0.106 & 0.047 & 2.629 & 0.596 \\ 
\multicolumn{16}{l}{Kernel regression (local constant) estimator with
various kernel functions} \\ 
Gaussian &  & -0.031 & 0.302 & 1.635 & 0.058 &  & -0.027 & 0.165 & 1.797 & 
0.072 &  & -0.016 & 0.090 & 2.073 & 0.058 \\ 
Epanechnikov &  & -0.009 & 0.550 & 2.959 & 0.038 &  & 0.004 & 0.312 & 3.357
& 0.052 &  & -0.005 & 0.171 & 3.874 & 0.042 \\ 
7th polynomial &  & -0.017 & 0.653 & 3.514 & 0.064 &  & 0.012 & 0.357 & 3.843
& 0.052 &  & -0.001 & 0.192 & 4.333 & 0.044 \\ 
7th polyweight &  & -0.012 & 0.605 & 3.259 & 0.046 &  & 0.009 & 0.344 & 3.696
& 0.049 &  & -0.004 & 0.187 & 4.222 & 0.039 \\ 
\multicolumn{16}{l}{Local linear estimator with various kernel functions} \\ 
Gaussian &  & 0.032 & 0.284 & 1.539 & 0.092 &  & 0.029 & 0.148 & 1.616 & 
0.091 &  & 0.032 & 0.077 & 1.894 & 0.102 \\ 
Epanechnikov &  & 0.019 & 0.424 & 2.284 & 0.057 &  & 0.014 & 0.235 & 2.535 & 
0.053 &  & 0.014 & 0.125 & 2.843 & 0.045 \\ 
7th polynomial &  & -0.001 & 0.553 & 2.979 & 0.089 &  & 0.015 & 0.312 & 3.353
& 0.076 &  & 0.006 & 0.167 & 3.791 & 0.053 \\ 
7th polyweight &  & 0.005 & 0.491 & 2.643 & 0.065 &  & 0.013 & 0.278 & 2.989
& 0.059 &  & 0.008 & 0.151 & 3.424 & 0.063 \\ \hline\hline
\end{tabular}
}
\end{sidewaystable}

\begin{sidewaystable}[ptb]
\caption{Simulation results when $\varepsilon _{i}\sim t\left(3\right) $}
\medskip
\label{table:t}
\resizebox{\linewidth}{!}{
\begin{tabular}{ccccC{1.2cm}C{1.5cm}cccC{1.2cm}C{1.5cm}cccC{1.2cm}C{1.5cm}}
\hline\hline
$\Pr \left( Y_{i}=0\right) =0.5$ &  & \multicolumn{4}{c}{$n=250$} &  & 
\multicolumn{4}{c}{$n=1000$} &  & \multicolumn{4}{c}{$n=4000$} \\ 
\cline{1-1}\cline{3-6}\cline{8-11}\cline{13-16}
Estimator &  & Bias & SD & RMSE ratio & Rejection rate &  & Bias & SD & RMSE
ratio & Rejection rate &  & Bias & SD & RMSE ratio & Rejection rate \\ \hline
\multicolumn{16}{l}{\cite{heckman1979sample}'s parametric two-step estimator}
\\ 
&  & 0.014 & 0.181 & 1 & 0.059 &  & 0.027 & 0.089 & 1 & 0.058 &  & 0.025 & 
0.047 & 1 & 0.114 \\ 
\multicolumn{16}{l}{\cite{heckman1990varieties}'s semiparametric estimator
with various smoothing parameters} \\ 
$1\%$ observations &  & -0.052 & 0.981 & 5.403 & 0.386 &  & 0.003 & 0.592 & 
6.363 & 0.155 &  & -0.001 & 0.296 & 5.594 & 0.069 \\ 
$5\%$ observations &  & -0.050 & 0.533 & 2.945 & 0.121 &  & -0.030 & 0.271 & 
2.934 & 0.070 &  & -0.029 & 0.135 & 2.603 & 0.056 \\ 
$10\%$ observations &  & -0.049 & 0.367 & 2.035 & 0.073 &  & -0.033 & 0.186
& 2.029 & 0.056 &  & -0.039 & 0.097 & 1.975 & 0.081 \\ 
$20\%$ observations &  & -0.061 & 0.260 & 1.468 & 0.063 &  & -0.052 & 0.133
& 1.532 & 0.065 &  & -0.056 & 0.067 & 1.654 & 0.137 \\ 
$30\%$ observations &  & -0.076 & 0.210 & 1.226 & 0.073 &  & -0.068 & 0.105
& 1.351 & 0.092 &  & -0.072 & 0.054 & 1.701 & 0.274 \\ 
$50\%$ observations &  & -0.120 & 0.163 & 1.115 & 0.113 &  & -0.109 & 0.080
& 1.452 & 0.257 &  & -0.111 & 0.042 & 2.238 & 0.770 \\ 
\multicolumn{16}{l}{\cite{andrews1998semiparametric}'s semiparametric
estimator with various smoothing parameters} \\ 
$1\%$ observations &  & -0.089 & 1.285 & 7.083 & 0.627 &  & 0.006 & 0.711 & 
7.646 & 0.197 &  & 0.005 & 0.346 & 6.541 & 0.072 \\ 
$5\%$ observations &  & -0.059 & 0.659 & 3.641 & 0.137 &  & -0.018 & 0.325 & 
3.503 & 0.069 &  & -0.022 & 0.157 & 2.999 & 0.060 \\ 
$10\%$ observations &  & -0.054 & 0.476 & 2.633 & 0.089 &  & -0.028 & 0.227
& 2.460 & 0.060 &  & -0.031 & 0.113 & 2.209 & 0.068 \\ 
$20\%$ observations &  & -0.051 & 0.311 & 1.733 & 0.069 &  & -0.039 & 0.156
& 1.727 & 0.058 &  & -0.045 & 0.080 & 1.727 & 0.087 \\ 
$30\%$ observations &  & -0.061 & 0.246 & 1.395 & 0.063 &  & -0.052 & 0.123
& 1.440 & 0.062 &  & -0.057 & 0.064 & 1.616 & 0.161 \\ 
$50\%$ observations &  & -0.089 & 0.184 & 1.123 & 0.075 &  & -0.080 & 0.093
& 1.319 & 0.142 &  & -0.083 & 0.048 & 1.813 & 0.438 \\ 
\multicolumn{16}{l}{Kernel regression (local constant) estimator with
various kernel functions} \\ 
Gaussian &  & -0.061 & 0.278 & 1.564 & 0.051 &  & -0.041 & 0.153 & 1.705 & 
0.066 &  & -0.042 & 0.086 & 1.798 & 0.090 \\ 
Epanechnikov &  & -0.053 & 0.547 & 3.021 & 0.041 &  & -0.021 & 0.303 & 3.264
& 0.048 &  & -0.020 & 0.163 & 3.099 & 0.058 \\ 
7th polynomial &  & -0.058 & 0.642 & 3.545 & 0.061 &  & -0.013 & 0.355 & 
3.824 & 0.059 &  & -0.014 & 0.190 & 3.596 & 0.055 \\ 
7th polyweight &  & -0.055 & 0.605 & 3.341 & 0.037 &  & -0.016 & 0.335 & 
3.610 & 0.053 &  & -0.017 & 0.181 & 3.427 & 0.054 \\ 
\multicolumn{16}{l}{Local linear estimator with various kernel functions} \\ 
Gaussian &  & -0.020 & 0.270 & 1.488 & 0.072 &  & -0.009 & 0.147 & 1.583 & 
0.069 &  & -0.018 & 0.082 & 1.587 & 0.067 \\ 
Epanechnikov &  & -0.034 & 0.399 & 2.200 & 0.046 &  & -0.016 & 0.222 & 2.397
& 0.046 &  & -0.022 & 0.125 & 2.403 & 0.060 \\ 
7th polynomial &  & -0.044 & 0.541 & 2.987 & 0.076 &  & -0.013 & 0.301 & 
3.238 & 0.071 &  & -0.015 & 0.164 & 3.112 & 0.070 \\ 
7th polyweight &  & -0.039 & 0.474 & 2.616 & 0.050 &  & -0.014 & 0.265 & 
2.850 & 0.056 &  & -0.019 & 0.147 & 2.803 & 0.068 \\ \hline\hline
\end{tabular}
}
\end{sidewaystable}

\begin{sidewaystable}[ptb]
\caption{Simulation results when $\varepsilon _{i}\sim \chi^2\left(3\right) $}
\medskip
\label{table:chi2}
\resizebox{\linewidth}{!}{
\begin{tabular}{ccccC{1.2cm}C{1.5cm}cccC{1.2cm}C{1.5cm}cccC{1.2cm}C{1.5cm}}
\hline\hline
$\Pr \left( Y_{i}=0\right) =0.5$ &  & \multicolumn{4}{c}{$n=250$} &  & 
\multicolumn{4}{c}{$n=1000$} &  & \multicolumn{4}{c}{$n=4000$} \\ 
\cline{1-1}\cline{3-6}\cline{8-11}\cline{13-16}
Estimator &  & Bias & SD & RMSE ratio & Rejection rate &  & Bias & SD & RMSE
ratio & Rejection rate &  & Bias & SD & RMSE ratio & Rejection rate \\ \hline
\multicolumn{16}{l}{\cite{heckman1979sample}'s parametric two-step estimator}
\\ 
&  & -0.146 & 0.167 & 1 & 0.180 &  & -0.145 & 0.082 & 1 & 0.459 &  & -0.146
& 0.041 & 1 & 0.913 \\ 
\multicolumn{16}{l}{\cite{heckman1990varieties}'s semiparametric estimator
with various smoothing parameters} \\ 
$1\%$ observations &  & -0.008 & 0.977 & 4.405 & 0.413 &  & 0.004 & 0.578 & 
3.465 & 0.134 &  & -0.024 & 0.299 & 1.978 & 0.076 \\ 
$5\%$ observations &  & -0.016 & 0.514 & 2.321 & 0.125 &  & -0.047 & 0.258 & 
1.569 & 0.067 &  & -0.052 & 0.128 & 0.913 & 0.073 \\ 
$10\%$ observations &  & -0.057 & 0.351 & 1.603 & 0.080 &  & -0.084 & 0.183
& 1.206 & 0.087 &  & -0.083 & 0.088 & 0.797 & 0.140 \\ 
$20\%$ observations &  & -0.117 & 0.243 & 1.215 & 0.099 &  & -0.129 & 0.124
& 1.072 & 0.189 &  & -0.126 & 0.059 & 0.917 & 0.530 \\ 
$30\%$ observations &  & -0.162 & 0.200 & 1.160 & 0.152 &  & -0.165 & 0.098
& 1.148 & 0.395 &  & -0.163 & 0.048 & 1.124 & 0.911 \\ 
$50\%$ observations &  & -0.229 & 0.149 & 1.232 & 0.354 &  & -0.230 & 0.073
& 1.445 & 0.869 &  & -0.229 & 0.037 & 1.535 & 1.000 \\ 
\multicolumn{16}{l}{\cite{andrews1998semiparametric}'s semiparametric
estimator with various smoothing parameters} \\ 
$1\%$ observations &  & -0.009 & 1.219 & 5.497 & 0.605 &  & 0.021 & 0.694 & 
4.162 & 0.180 &  & -0.009 & 0.350 & 2.315 & 0.078 \\ 
$5\%$ observations &  & -0.009 & 0.624 & 2.816 & 0.153 &  & -0.026 & 0.315 & 
1.893 & 0.067 &  & -0.040 & 0.156 & 1.061 & 0.068 \\ 
$10\%$ observations &  & -0.037 & 0.440 & 1.990 & 0.102 &  & -0.057 & 0.217
& 1.343 & 0.066 &  & -0.063 & 0.108 & 0.822 & 0.097 \\ 
$20\%$ observations &  & -0.084 & 0.291 & 1.365 & 0.081 &  & -0.099 & 0.148
& 1.068 & 0.113 &  & -0.097 & 0.072 & 0.799 & 0.251 \\ 
$30\%$ observations &  & -0.119 & 0.231 & 1.173 & 0.106 &  & -0.131 & 0.116
& 1.048 & 0.210 &  & -0.127 & 0.056 & 0.919 & 0.609 \\ 
$50\%$ observations &  & -0.182 & 0.171 & 1.124 & 0.206 &  & -0.185 & 0.084
& 1.217 & 0.582 &  & -0.185 & 0.042 & 1.250 & 0.991 \\ 
\multicolumn{16}{l}{Kernel regression (local constant) estimator with
various kernel functions} \\ 
Gaussian &  & -0.087 & 0.278 & 1.316 & 0.093 &  & -0.092 & 0.157 & 1.091 & 
0.131 &  & -0.079 & 0.088 & 0.783 & 0.199 \\ 
Epanechnikov &  & -0.019 & 0.526 & 2.372 & 0.051 &  & -0.028 & 0.302 & 1.817
& 0.067 &  & -0.038 & 0.167 & 1.128 & 0.068 \\ 
7th polynomial &  & -0.015 & 0.624 & 2.816 & 0.071 &  & -0.017 & 0.355 & 
2.131 & 0.072 &  & -0.030 & 0.196 & 1.308 & 0.062 \\ 
7th polyweight &  & -0.016 & 0.580 & 2.615 & 0.045 &  & -0.020 & 0.337 & 
2.024 & 0.069 &  & -0.033 & 0.186 & 1.245 & 0.066 \\ 
\multicolumn{16}{l}{Local linear estimator with various kernel functions} \\ 
Gaussian &  & -0.066 & 0.258 & 1.202 & 0.095 &  & -0.070 & 0.134 & 0.907 & 
0.122 &  & -0.065 & 0.069 & 0.627 & 0.183 \\ 
Epanechnikov &  & -0.027 & 0.387 & 1.751 & 0.056 &  & -0.044 & 0.222 & 1.355
& 0.066 &  & -0.039 & 0.122 & 0.846 & 0.071 \\ 
7th polynomial &  & -0.007 & 0.530 & 2.392 & 0.090 &  & -0.019 & 0.302 & 
1.813 & 0.072 &  & -0.026 & 0.167 & 1.116 & 0.071 \\ 
7th polyweight &  & -0.012 & 0.466 & 2.102 & 0.063 &  & -0.029 & 0.266 & 
1.605 & 0.059 &  & -0.031 & 0.148 & 0.997 & 0.072 \\ \hline\hline
\end{tabular}
}
\end{sidewaystable}

By contrast, the kernel regression and local linear estimators have no
difficulty in choosing the bandwidth parameter because fully data-driven
procedures of selecting the optimal bandwidths have been explicitly
proposed. Table \ref{table:normal} shows that the proposed bandwidths
perform satisfactorily in that the rejection rates for the kernel regression
and local linear estimators are all close to 0.05 for different sample sizes
and different kernel functions. In terms of RMSE, the local linear estimator
is superior to the kernel regression estimator, as expected. The RMSE of the
local linear estimator with Gaussian kernel is comparable to the optimal
RMSE of the Heckman and AS estimators. However, the Gaussian kernel may not
be the best choice because over-rejection for the t test appears to be a
problem. As an alternative, the local linear estimator with Epanechnikov
kernel achieves a satisfactory balance between RMSE and rejection rate. The
polynomial and polyweight kernels also do not lead to size distortion of the
t test but are clearly outperformed by the simple Epanechnikov kernel in
terms of RMSE.

Table \ref{table:t} investigates the finite sample behavior of the
estimators under $t\left( 3\right) $ distributed disturbance. The parametric
two-step estimator performs well for small sample sizes but deteriorates as
the sample size increases because its nonvanishing bias due to nonnormality
becomes large relative to its declining SD. By contrast, the performance of
the semiparametric estimators is robust to nonnormality, and the main
findings are almost the same as in the normal case. First, if the smoothing
parameter of the Heckman and AS estimators is improperly chosen, their RMSEs
may be increased by several factors and the t tests may yield misleading
inferences. Second, the smoothing parameter giving rise to the smallest RMSE
may lead to evident over-rejection for the t test. These facts indicate the
importance of selecting a proper smoothing parameter for the
identification-at-infinity estimators, but this task is difficult since the
RMSE and rejection rate are unobserved in practice. Third, the bandwidth
selection algorithms proposed for the kernel-type estimators perform well in
terms of rejection rate. Fourth, the local linear estimator dominates the
kernel regression estimator. Fifth, the RMSE of the local linear estimator
with Gaussian kernel is comparable to (when $n=4000$, even smaller than) the
smallest RMSE of the identification-at-infinity estimators. Sixth, using the
Epanechnikov kernel results in more accurate rejection rate.

Table \ref{table:chi2} considers the case of $\chi ^{2}\left( 3\right) $
distributed disturbance. In this case, the parametric two-step estimator has
notable bias, and the probability of making a type I error rapidly
approaches one. Comparison with Table \ref{table:t} indicates that the
skewness of the nonnormal disturbance exerts a worse influence on the
parametric estimator than does the fat tail. The Heckman and AS estimators
become less robust to the smoothing parameter under this design, in the
sense that the rejection rate is below 0.1 for only a narrow range of
smoothing parameters. On the other hand, the local linear estimator with
Epanechnikov kernel still has desirable finite sample properties in terms of
both RMSE and rejection rate.

Tables \ref{table:small}-\ref{table:large} investigate the effect of the
proportion of censoring. Table \ref{table:small} shows that, under mild
censoring, all the considered estimators behave better. The parametric
two-step estimator becomes less biased in nonnormal designs. The Heckman and
AS estimators are more robust to the smoothing parameter, and the rejection
rates for the kernel regression and local linear estimators are closer to
the true level, even when the Gaussian kernel is used. Table \ref%
{table:large} reveals the reverse side of the coin, with an unchanged
conclusion being the superiority of the local linear estimator with
Epanechnikov kernel.

\begin{sidewaystable}[ptb]
\caption{Simulation results when $n=1000$ and the amount of zero observations is small}
\medskip
\label{table:small}
\resizebox{\linewidth}{!}{
\begin{tabular}{ccccC{1.2cm}C{1.5cm}cccC{1.2cm}C{1.5cm}cccC{1.2cm}C{1.5cm}}
\hline\hline
$\Pr \left( Y_{i}=0\right) =0.2$ &  & \multicolumn{4}{c}{$\varepsilon
_{i}\sim N\left( 0,1\right) $} &  & \multicolumn{4}{c}{$\varepsilon _{i}\sim
t\left( 3\right) $} &  & \multicolumn{4}{c}{$\varepsilon _{i}\sim \chi
^{2}\left( 3\right) $} \\ \cline{1-1}\cline{3-6}\cline{8-11}\cline{13-16}
Estimator &  & Bias & SD & RMSE ratio & Rejection rate &  & Bias & SD & RMSE
ratio & Rejection rate &  & Bias & SD & RMSE ratio & Rejection rate \\ \hline
\multicolumn{16}{l}{\cite{heckman1979sample}'s parametric two-step estimator}
\\ 
&  & -0.001 & 0.065 & 1 & 0.055 &  & -0.015 & 0.068 & 1 & 0.070 &  & -0.071
& 0.066 & 1 & 0.239 \\ 
\multicolumn{16}{l}{\cite{heckman1990varieties}'s semiparametric estimator
with various smoothing parameters} \\ 
$1\%$ observations &  & -0.006 & 0.483 & 7.447 & 0.103 &  & -0.014 & 0.466 & 
6.719 & 0.099 &  & -0.004 & 0.477 & 4.920 & 0.102 \\ 
$5\%$ observations &  & 0.006 & 0.224 & 3.449 & 0.065 &  & -0.017 & 0.218 & 
3.145 & 0.061 &  & -0.021 & 0.206 & 2.134 & 0.057 \\ 
$10\%$ observations &  & 0.002 & 0.163 & 2.510 & 0.063 &  & -0.027 & 0.153 & 
2.231 & 0.049 &  & -0.036 & 0.145 & 1.545 & 0.054 \\ 
$20\%$ observations &  & -0.001 & 0.113 & 1.743 & 0.056 &  & -0.035 & 0.132
& 1.970 & 0.058 &  & -0.056 & 0.106 & 1.242 & 0.081 \\ 
$30\%$ observations &  & -0.012 & 0.095 & 1.474 & 0.057 &  & -0.042 & 0.102
& 1.589 & 0.082 &  & -0.075 & 0.087 & 1.183 & 0.153 \\ 
$50\%$ observations &  & -0.036 & 0.073 & 1.252 & 0.093 &  & -0.054 & 0.074
& 1.317 & 0.130 &  & -0.111 & 0.067 & 1.336 & 0.434 \\ 
\multicolumn{16}{l}{\cite{andrews1998semiparametric}'s semiparametric
estimator with various smoothing parameters} \\ 
$1\%$ observations &  & -0.006 & 0.573 & 8.841 & 0.126 &  & 0.002 & 0.536 & 
7.723 & 0.109 &  & -0.016 & 0.573 & 5.917 & 0.136 \\ 
$5\%$ observations &  & 0.002 & 0.264 & 4.077 & 0.061 &  & -0.012 & 0.258 & 
3.714 & 0.056 &  & -0.010 & 0.258 & 2.661 & 0.070 \\ 
$10\%$ observations &  & 0.006 & 0.190 & 2.924 & 0.061 &  & -0.021 & 0.182 & 
2.642 & 0.061 &  & -0.026 & 0.172 & 1.797 & 0.051 \\ 
$20\%$ observations &  & 0.000 & 0.132 & 2.037 & 0.060 &  & -0.030 & 0.130 & 
1.918 & 0.060 &  & -0.042 & 0.121 & 1.325 & 0.064 \\ 
$30\%$ observations &  & -0.005 & 0.107 & 1.651 & 0.054 &  & -0.036 & 0.116
& 1.755 & 0.074 &  & -0.057 & 0.100 & 1.185 & 0.095 \\ 
$50\%$ observations &  & -0.020 & 0.081 & 1.285 & 0.066 &  & -0.046 & 0.088
& 1.423 & 0.091 &  & -0.088 & 0.074 & 1.189 & 0.243 \\ 
\multicolumn{16}{l}{Kernel regression (local constant) estimator with
various kernel functions} \\ 
Gaussian &  & 0.002 & 0.160 & 2.466 & 0.051 &  & -0.025 & 0.155 & 2.259 & 
0.063 &  & -0.031 & 0.147 & 1.553 & 0.049 \\ 
Epanechnikov &  & -0.002 & 0.302 & 4.663 & 0.043 &  & -0.011 & 0.300 & 4.322
& 0.054 &  & -0.005 & 0.305 & 3.150 & 0.059 \\ 
7th polynomial &  & -0.004 & 0.345 & 5.321 & 0.056 &  & -0.007 & 0.342 & 
4.920 & 0.057 &  & -0.010 & 0.359 & 3.711 & 0.074 \\ 
7th polyweight &  & -0.004 & 0.330 & 5.094 & 0.043 &  & -0.012 & 0.329 & 
4.748 & 0.056 &  & -0.005 & 0.337 & 3.482 & 0.063 \\ 
\multicolumn{16}{l}{Local linear estimator with various kernel functions} \\ 
Gaussian &  & 0.017 & 0.143 & 2.222 & 0.070 &  & -0.024 & 0.147 & 2.151 & 
0.067 &  & -0.020 & 0.128 & 1.334 & 0.057 \\ 
Epanechnikov &  & 0.005 & 0.234 & 3.608 & 0.057 &  & -0.016 & 0.227 & 3.280
& 0.053 &  & -0.011 & 0.217 & 2.247 & 0.044 \\ 
7th polynomial &  & 0.003 & 0.301 & 4.635 & 0.056 &  & -0.008 & 0.299 & 4.310
& 0.063 &  & -0.005 & 0.304 & 3.135 & 0.077 \\ 
7th polyweight &  & 0.008 & 0.274 & 4.227 & 0.054 &  & -0.010 & 0.273 & 3.933
& 0.056 &  & -0.008 & 0.265 & 2.734 & 0.061 \\ \hline\hline
\end{tabular}
}
\end{sidewaystable}

\begin{sidewaystable}[ptb]
\caption{Simulation results when $n=1000$ and the amount of zero observations is large}
\medskip
\label{table:large}
\resizebox{\linewidth}{!}{
\begin{tabular}{ccccC{1.2cm}C{1.5cm}cccC{1.2cm}C{1.5cm}cccC{1.2cm}C{1.5cm}}
\hline\hline
$\Pr \left( Y_{i}=0\right) =0.8$ &  & \multicolumn{4}{c}{$\varepsilon
_{i}\sim N\left( 0,1\right) $} &  & \multicolumn{4}{c}{$\varepsilon _{i}\sim
t\left( 3\right) $} &  & \multicolumn{4}{c}{$\varepsilon _{i}\sim \chi
^{2}\left( 3\right) $} \\ \cline{1-1}\cline{3-6}\cline{8-11}\cline{13-16}
Estimator &  & Bias & SD & RMSE ratio & Rejection rate &  & Bias & SD & RMSE
ratio & Rejection rate &  & Bias & SD & RMSE ratio & Rejection rate \\ \hline
\multicolumn{16}{l}{\cite{heckman1979sample}'s parametric two-step estimator}
\\ 
&  & 0.001 & 0.168 & 1 & 0.045 &  & 0.195 & 0.231 & 1 & 0.193 &  & -0.234 & 
0.132 & 1 & 0.428 \\ 
\multicolumn{16}{l}{\cite{heckman1990varieties}'s semiparametric estimator
with various smoothing parameters} \\ 
$1\%$ observations &  & -0.040 & 0.889 & 5.285 & 0.316 &  & 0.030 & 0.879 & 
2.907 & 0.343 &  & -0.067 & 0.855 & 3.186 & 0.325 \\ 
$5\%$ observations &  & -0.035 & 0.403 & 2.404 & 0.086 &  & -0.044 & 0.415 & 
1.380 & 0.097 &  & -0.081 & 0.391 & 1.485 & 0.098 \\ 
$10\%$ observations &  & -0.066 & 0.289 & 1.762 & 0.080 &  & -0.056 & 0.292
& 0.984 & 0.064 &  & -0.133 & 0.277 & 1.141 & 0.114 \\ 
$20\%$ observations &  & -0.154 & 0.208 & 1.539 & 0.149 &  & -0.099 & 0.212
& 0.772 & 0.089 &  & -0.212 & 0.188 & 1.053 & 0.219 \\ 
$30\%$ observations &  & -0.237 & 0.166 & 1.716 & 0.324 &  & -0.143 & 0.171
& 0.737 & 0.145 &  & -0.269 & 0.148 & 1.140 & 0.445 \\ 
$50\%$ observations &  & -0.388 & 0.124 & 2.419 & 0.876 &  & -0.228 & 0.134
& 0.875 & 0.452 &  & -0.362 & 0.113 & 1.409 & 0.891 \\ 
\multicolumn{16}{l}{\cite{andrews1998semiparametric}'s semiparametric
estimator with various smoothing parameters} \\ 
$1\%$ observations &  & -0.021 & 1.085 & 6.449 & 0.428 &  & 0.024 & 1.011 & 
3.345 & 0.442 &  & -0.062 & 1.039 & 3.869 & 0.436 \\ 
$5\%$ observations &  & -0.022 & 0.497 & 2.955 & 0.105 &  & -0.032 & 0.487 & 
1.614 & 0.104 &  & -0.056 & 0.485 & 1.814 & 0.108 \\ 
$10\%$ observations &  & -0.039 & 0.347 & 2.072 & 0.075 &  & -0.048 & 0.353
& 1.178 & 0.077 &  & -0.094 & 0.335 & 1.294 & 0.094 \\ 
$20\%$ observations &  & -0.093 & 0.243 & 1.545 & 0.080 &  & -0.070 & 0.251
& 0.860 & 0.071 &  & -0.159 & 0.230 & 1.040 & 0.136 \\ 
$30\%$ observations &  & -0.155 & 0.195 & 1.480 & 0.150 &  & -0.098 & 0.203
& 0.744 & 0.088 &  & -0.209 & 0.181 & 1.027 & 0.233 \\ 
$50\%$ observations &  & -0.282 & 0.144 & 1.878 & 0.526 &  & -0.162 & 0.152
& 0.736 & 0.199 &  & -0.291 & 0.131 & 1.187 & 0.610 \\ 
\multicolumn{16}{l}{Kernel regression (local constant) estimator with
various kernel functions} \\ 
Gaussian &  & -0.142 & 0.216 & 1.535 & 0.207 &  & -0.110 & 0.193 & 0.734 & 
0.156 &  & -0.201 & 0.198 & 1.050 & 0.349 \\ 
Epanechnikov &  & -0.052 & 0.306 & 1.843 & 0.062 &  & -0.055 & 0.296 & 0.994
& 0.066 &  & -0.115 & 0.294 & 1.174 & 0.098 \\ 
7th polynomial &  & -0.041 & 0.355 & 2.122 & 0.069 &  & -0.043 & 0.343 & 
1.142 & 0.062 &  & -0.094 & 0.359 & 1.380 & 0.094 \\ 
7th polyweight &  & -0.042 & 0.335 & 2.005 & 0.062 &  & -0.050 & 0.325 & 
1.087 & 0.060 &  & -0.098 & 0.329 & 1.276 & 0.089 \\ 
\multicolumn{16}{l}{Local linear estimator with various kernel functions} \\ 
Gaussian &  & -0.101 & 0.188 & 1.264 & 0.224 &  & -0.012 & 0.182 & 0.602 & 
0.139 &  & -0.169 & 0.181 & 0.920 & 0.328 \\ 
Epanechnikov &  & -0.035 & 0.244 & 1.467 & 0.087 &  & -0.023 & 0.230 & 0.764
& 0.068 &  & -0.134 & 0.217 & 0.947 & 0.136 \\ 
7th polynomial &  & -0.005 & 0.317 & 1.885 & 0.082 &  & -0.022 & 0.306 & 
1.015 & 0.078 &  & -0.087 & 0.306 & 1.182 & 0.104 \\ 
7th polyweight &  & -0.011 & 0.284 & 1.686 & 0.081 &  & -0.023 & 0.272 & 
0.904 & 0.067 &  & -0.106 & 0.263 & 1.054 & 0.112 \\ \hline\hline
\end{tabular}
}
\end{sidewaystable}

\section{Conclusion}

\label{sec:conclution}This paper rephrases the identification at infinity
into an identification at the boundary via a CDF transformation and
accordingly proposes a kernel approach to semiparametrically estimate the
intercept of the sample selection model. The proposed kernel regression
estimator with generic transformation is a generalization of the
identification-at-infinity estimators and thus inherits the disadvantage
that the asymptotic bias and variance are implicit functions of the
bandwidth parameter. To select a bandwidth that minimizes the asymptotic
mean squared error, I use a specific transformation, namely, the empirical
CDF of the selection index, under which the asymptotic bias and variance
become explicit with respect to the bandwidth. A plug-in bandwidth selection
algorithm with regularization is therefore suggested. For the purpose of
bias reduction, I further propose a local linear estimator and an associated
analogous bandwidth selection algorithm. A simulation study illustrates the
effectiveness of the selected bandwidths. Comparison of the finite sample
performance of the estimators indicates that the local linear estimator with
Epanechnikov kernel is superior to the parametric two-step estimator under
nonnormal disturbance and to the identification-at-infinity estimators in
most cases.

\renewcommand{\thetheorem}{A.\arabic{theorem}} \setcounter{theorem}{0}

\section{Appendix}

\subsection{Proofs of the Example, Lemma, and Corollary}

\begin{proof}[Proof of the Example]
By Assumption 2, there exists a $\delta \in \left( 0,1\right) $
and a positive $C_{\delta }$ such that $k\left( u\right) \geq C_{\delta }$
for any $u\in \left[ 0,\delta \right] $. Therefore,%
\begin{eqnarray*}
Ek_{ni}^{2} &=&\int_{0}^{1}k^{2}\left( \frac{1-t}{h_{n}}\right) f_{F\left(
W\right) }\left( t\right) dt=h_{n}\int_{0}^{1\left/ h_{n}\right.
}k^{2}\left( u\right) f_{F\left( W\right) }\left( 1-uh_{n}\right) du \\
&=&h_{n}\int_{0}^{1}k^{2}\left( u\right) f_{F\left( W\right) }\left(
1-uh_{n}\right) du\geq C_{\delta }^{2}h_{n}\int_{0}^{\delta }f_{F\left(
W\right) }\left( 1-uh_{n}\right) du \\
&=&C_{\delta }^{2}\int_{1-\delta h_{n}}^{1}f_{F\left( W\right) }\left(
s\right) ds=C_{\delta }^{2}\Pr \left( F\left( W_{i}\right) >1-\delta
h_{n}\right) .
\end{eqnarray*}%
Since $F\left( \cdot \right) $ is the standard Laplacian CDF, the ratio in
Assumption 5.(iii) is bounded by%
\begin{equation*}
\frac{\left[ \Pr \left( F\left( W_{i}\right) >1-h_{n}\right) \right] ^{1+c}}{%
Ek_{ni}^{2}}\leq \frac{\left[ \Pr \left( F\left( W_{i}\right)
>1-h_{n}\right) \right] ^{1+c}}{C_{\delta }^{2}\Pr \left( F\left(
W_{i}\right) >1-\delta h_{n}\right) }=\frac{\left[ \Pr \left( W_{i}>-\log
\left( 2h_{n}\right) \right) \right] ^{1+c}}{C_{\delta }^{2}\Pr \left(
W_{i}>-\log \left( 2\delta h_{n}\right) \right) }.
\end{equation*}%
Denote $t_{n}=-\log \left( 2h_{n}\right) \rightarrow +\infty $.

(i) If the distribution of $W_{i}$ has a power-type upper tail such that $%
\Pr \left( W_{i}>t\right) \sim t^{-\lambda }$ for some $\lambda >0$, we have%
\begin{equation*}
\frac{\left[ \Pr \left( W_{i}>-\log \left( 2h_{n}\right) \right) \right]
^{1+c}}{C_{\delta }^{2}\Pr \left( W_{i}>-\log \left( 2\delta h_{n}\right)
\right) }\sim \frac{\left( t_{n}-\log \delta \right) ^{\lambda }}{%
t_{n}^{\lambda \left( 1+c\right) }}\rightarrow 0
\end{equation*}%
for any $c>0$ and Assumption 5.(iii) holds.

(ii) If the distribution of $W_{i}$ has an exponential-type upper tail such
that $\Pr \left( W_{i}>t\right) \sim \exp \left( -c_{0}t^{\lambda }\right) $
for some $\lambda >0$ and $c_{0}>0$, we have%
\begin{eqnarray*}
\frac{\left[ \Pr \left( W_{i}>-\log \left( 2h_{n}\right) \right) \right]
^{1+c}}{C_{\delta }^{2}\Pr \left( W_{i}>-\log \left( 2\delta h_{n}\right)
\right) } &\sim &\frac{\exp \left\{ -\left( 1+c\right) c_{0}t_{n}^{\lambda
}\right\} }{\exp \left\{ -c_{0}\left( t_{n}-\log \delta \right) ^{\lambda
}\right\} } \\
&=&\exp \left\{ c_{0}t_{n}^{\lambda }\left[ \left( 1+\frac{-\log \delta }{%
t_{n}}\right) ^{\lambda }-\left( 1+c\right) \right] \right\} \\
&\rightarrow &0
\end{eqnarray*}%
for any $c>0$ and Assumption 5.(iii) holds.

(iii) If the upper tail of $W_{i}$'s distribution decays as rapidly as $\Pr
\left( W_{i}>t\right) \sim \exp \left( -\exp \left( c_{0}t^{\lambda }\right)
\right) $, however, Assumption 5.(iii) is guaranteed to hold only for $%
\lambda <1$, because in this case%
\begin{eqnarray*}
\frac{\left[ \Pr \left( W_{i}>-\log \left( 2h_{n}\right) \right) \right]
^{1+c}}{C_{\delta }^{2}\Pr \left( W_{i}>-\log \left( 2\delta h_{n}\right)
\right) } &\sim &\frac{\exp \left\{ -\left( 1+c\right) \exp \left(
c_{0}t_{n}^{\lambda }\right) \right\} }{\exp \left\{ -\exp \left[
c_{0}\left( t_{n}-\log \delta \right) ^{\lambda }\right] \right\} } \\
&=&\exp \left\{ \exp \left( c_{0}t_{n}^{\lambda }\right) \left[ \exp \left[
c_{0}t_{n}^{\lambda }\left( \left( 1+\frac{-\log \delta }{t_{n}}\right)
^{\lambda }-1\right) \right] -\left( 1+c\right) \right] \right\} \\
&\rightarrow &\left\{
\begin{array}{ll}
0 & \text{if }\left( \lambda <1\right) \text{ or }\left( \lambda =1\text{
and }c>\delta ^{-c_{0}}-1\right) \\
1 & \text{if }\lambda =1\text{ and }c=\delta ^{-c_{0}}-1 \\
+\infty & \text{if }\left( \lambda >1\right) \text{ or }\left( \lambda =1%
\text{ and }c<\delta ^{-c_{0}}-1\right)%
\end{array}%
\right.
\end{eqnarray*}
\end{proof}

\begin{proof}[Proof of the Lemma]
For the numerator of the bias term, we have%
\begin{eqnarray*}
EU_{i}D_{i}k_{ni} &=&E\left[ k_{ni}E\left[ U_{i}D_{i}\left| W_{i}\right. %
\right] \right] =E\left[ k_{ni}G\left( F_{W}\left( W_{i}\right) \right)
G_{1}\left( F_{W}\left( W_{i}\right) \right) \right]  \\
&=&\int_{0}^{1}k\left( \frac{1-t}{h_{n}}\right) G\left( t\right) G_{1}\left(
t\right) dt \\
&=&h_{n}\int_{0}^{1}k\left( s\right) G\left( 1-h_{n}s\right) G_{1}\left(
1-h_{n}s\right) ds \\
&=&h_{n}\int_{0}^{1}k\left( s\right) h_{n}s\left[ -g\left( 1\right) +o\left(
1\right) \right] ds \\
&=&\left[ -\kappa _{1}g\left( 1\right) +o\left( 1\right) \right] h_{n}^{2},
\end{eqnarray*}%
where $G_{1}\left( t\right) =E\left[ D_{i}\left| F_{W}\left( W_{i}\right)
=t\right. \right] =\Pr \left( F_{W}\left( \varepsilon _{i}\right) <t\right) $
and $G_{1}\left( 1\right) =1$. For the denominator, it follows from Lemma %
\ref{Lemma:E} and Equation (\ref{knir}) that%
\begin{equation*}
ED_{i}k_{ni}=Ek_{ni}\left( 1+o\left( 1\right) \right) =\kappa
_{0}h_{n}\left( 1+o\left( 1\right) \right) .
\end{equation*}%
Therefore,%
\begin{equation*}
\frac{EU_{i}D_{i}k_{ni}}{ED_{i}k_{ni}}=\frac{\left[ -\kappa _{1}g\left(
1\right) +o\left( 1\right) \right] h_{n}}{\kappa _{0}\left( 1+o\left(
1\right) \right) }=-\frac{\kappa _{1}g\left( 1\right) }{\kappa _{0}}%
h_{n}+o\left( h_{n}\right) .
\end{equation*}
\end{proof}

\begin{proof}[Proof of the Corollary]
It follows immediately from Theorem \ref{Thm:lite}, Equation (%
\ref{knir}), and the Lemma.
\end{proof}

\subsection{Proof of Theorem \protect\ref{Thm:lite}}

\medskip

Denote $\tilde{k}_{ni}=k\left( \left. \left( 1-F\left( \hat{W}_{i}\right)
\right) \right/ h_{n}\right) $ and $k_{ni}=k\left( \left. \left( 1-F\left(
W_{i}\right) \right) \right/ h_{n}\right) $. And define the infeasible
kernel estimator as%
\begin{equation*}
\mu _{n}=\frac{\sum_{i=1}^{n}\left( Y_{i}-Z_{i}^{\prime }\theta _{0}\right)
D_{i}k_{ni}}{\sum_{i=1}^{n}D_{i}k_{ni}}.
\end{equation*}%
The asymptotic normality of $\tilde{\mu}$ follows by first establishing%
\begin{equation}
\frac{\sqrt{n}Ek_{ni}}{\sqrt{Ek_{ni}^{2}}}\left( \mu _{n}-\mu _{0}-\frac{%
EU_{i}D_{i}k_{ni}}{ED_{i}k_{ni}}\right) \rightarrow N\left( 0,\sigma
_{U}^{2}\right)  \label{NormalofInfeasible}
\end{equation}%
and then proving the asymptotic negligibility of $\tilde{\mu}-\mu _{n}$.

\begin{itemize}
\item \textbf{First step: proving}%
\begin{equation*}
\frac{\sqrt{n}Ek_{ni}}{\sqrt{Ek_{ni}^{2}}}\left( \mu _{n}-\mu _{0}-\frac{%
EU_{i}D_{i}k_{ni}}{ED_{i}k_{ni}}\right) \rightarrow N\left( 0,\sigma
_{U}^{2}\right) ,
\end{equation*}%
where%
\begin{equation*}
\mu _{n}=\frac{\sum_{i=1}^{n}\left( Y_{i}-Z_{i}^{\prime }\theta _{0}\right)
D_{i}k_{ni}}{\sum_{i=1}^{n}D_{i}k_{ni}}=\mu _{0}+\frac{%
\sum_{i=1}^{n}U_{i}D_{i}k_{ni}}{\sum_{i=1}^{n}D_{i}k_{ni}}.
\end{equation*}
\end{itemize}

For the triangular array $\left\{ U_{i}D_{i}k_{ni}:i\leq n,n\geq 1\right\} $%
, since%
\begin{equation*}
Var\left( \sum_{i=1}^{n}U_{i}D_{i}k_{ni}\right) =nVar\left(
U_{i}D_{i}k_{ni}\right) \leq nE\left[ U_{i}^{2}D_{i}k_{ni}^{2}\right] \leq
n\sigma _{U}^{2}\bar{k}^{2}<\infty ,
\end{equation*}%
it follows from Lindeberg's central limit theorem 
\citep[e.g.,][Theorem
1.15]{shao2003} and Lemma \ref{Lemma:Lindeberg} that 
\begin{equation*}
\frac{\sum_{i=1}^{n}\left( U_{i}D_{i}k_{ni}-EU_{i}D_{i}k_{ni}\right) }{\sqrt{%
nVar\left( U_{i}D_{i}k_{ni}\right) }}\rightarrow N\left( 0,1\right) .
\end{equation*}%
Therefore, by Lemmas \ref{Lemma:E}, \ref{Lemma:var} and \ref{Lemma:sample}%
\begin{eqnarray*}
&&\frac{\sqrt{n}Ek_{ni}}{\sqrt{Ek_{ni}^{2}}}\left( \mu _{n}-\mu _{0}-\frac{%
EU_{i}D_{i}k_{ni}}{\left( \left. 1\right/ n\right) \sum_{i=1}^{n}D_{i}k_{ni}}%
\right) \\
&=&\sigma _{U}\left( \frac{Ek_{ni}}{\left( \left. 1\right/ n\right)
\sum_{i=1}^{n}D_{i}k_{ni}}\right) \sqrt{\frac{Var\left(
U_{i}D_{i}k_{ni}\right) }{\sigma _{U}^{2}Ek_{ni}^{2}}}\frac{%
\sum_{i=1}^{n}\left( U_{i}D_{i}k_{ni}-EU_{i}D_{i}k_{ni}\right) }{\sqrt{%
nVar\left( U_{i}D_{i}k_{ni}\right) }} \\
&=&\sigma _{U}\left( 1+o_{p}\left( 1\right) \right) \frac{%
\sum_{i=1}^{n}\left( U_{i}D_{i}k_{ni}-EU_{i}D_{i}k_{ni}\right) }{\sqrt{%
nVar\left( U_{i}D_{i}k_{ni}\right) }}\rightarrow N\left( 0,\sigma
_{U}^{2}\right) .
\end{eqnarray*}%
It remains to show that%
\begin{equation*}
\frac{\sqrt{n}Ek_{ni}}{\sqrt{Ek_{ni}^{2}}}\left( \frac{EU_{i}D_{i}k_{ni}}{%
\left( \left. 1\right/ n\right) \sum_{i=1}^{n}D_{i}k_{ni}}-\frac{%
EU_{i}D_{i}k_{ni}}{ED_{i}k_{ni}}\right) \overset{p}{\rightarrow }0
\end{equation*}%
for this step of proof. Note that $\sum_{i=1}^{n}\left(
D_{i}k_{ni}-ED_{i}k_{ni}\right) \left/ \sqrt{nEk_{ni}^{2}}\right. $ is
bounded in probability because it has mean zero and variance%
\begin{equation*}
\frac{Var\left( \sum_{i=1}^{n}D_{i}k_{ni}\right) }{nEk_{ni}^{2}}=\frac{%
Var\left( D_{i}k_{ni}\right) }{Ek_{ni}^{2}}\leq \frac{ED_{i}k_{ni}^{2}}{%
Ek_{ni}^{2}}\leq 1.
\end{equation*}%
As a result, we have%
\begin{eqnarray*}
&&\frac{\sqrt{n}Ek_{ni}}{\sqrt{Ek_{ni}^{2}}}\left( \frac{EU_{i}D_{i}k_{ni}}{%
\left( \left. 1\right/ n\right) \sum_{i=1}^{n}D_{i}k_{ni}}-\frac{%
EU_{i}D_{i}k_{ni}}{ED_{i}k_{ni}}\right) \\
&=&-\frac{Ek_{ni}}{\left( \left. 1\right/ n\right) \sum_{i=1}^{n}D_{i}k_{ni}}%
\cdot \frac{EU_{i}D_{i}k_{ni}}{ED_{i}k_{ni}}\cdot \frac{\sum_{i=1}^{n}\left(
D_{i}k_{ni}-ED_{i}k_{ni}\right) }{\sqrt{nEk_{ni}^{2}}} \\
&=&-\left( 1+o_{p}\left( 1\right) \right) \cdot o\left( 1\right) \cdot
O_{p}\left( 1\right) =o_{p}\left( 1\right) ,
\end{eqnarray*}%
where the second equality follows from Lemmas \ref{Lemma:E}, \ref{Lemma:bias}
and \ref{Lemma:sample}.

\begin{itemize}
\item \textbf{Second step: proving}%
\begin{equation*}
\frac{\sqrt{n}Ek_{ni}}{\sqrt{Ek_{ni}^{2}}}\left( \tilde{\mu}-\mu _{n}\right) 
\overset{p}{\rightarrow }0.
\end{equation*}
\end{itemize}

To this end, write the left-hand side as%
\begin{equation*}
C_{n}\left( \frac{\tilde{A}_{n}}{\tilde{B}_{n}}-\frac{A_{n}}{B_{n}}\right)
=C_{n}\frac{\tilde{A}_{n}-A_{n}}{B_{n}}\cdot \frac{B_{n}}{\tilde{B}_{n}}%
-C_{n}\frac{\tilde{B}_{n}-B_{n}}{B_{n}}\cdot \frac{A_{n}}{B_{n}}\cdot \frac{%
B_{n}}{\tilde{B}_{n}},
\end{equation*}%
where%
\begin{eqnarray*}
C_{n} &=&\sqrt{n}Ek_{ni}\left/ \sqrt{Ek_{ni}^{2}}\right. , \\
\tilde{A}_{n} &=&\sum_{i=1}^{n}\left( Y_{i}-Z_{i}^{\prime }\hat{\theta}%
\right) D_{i}\tilde{k}_{ni},\text{ }\tilde{B}_{n}=\sum_{i=1}^{n}D_{i}\tilde{k%
}_{ni}, \\
A_{n} &=&\sum_{i=1}^{n}\left( Y_{i}-Z_{i}^{\prime }\theta _{0}\right)
D_{i}k_{ni},\text{ }B_{n}=\sum_{i=1}^{n}D_{i}k_{ni}.
\end{eqnarray*}%
It is sufficient to show that%
\begin{equation}
\text{(a) }\frac{\tilde{B}_{n}}{B_{n}}\overset{p}{\rightarrow }1\text{, (b) }%
C_{n}\frac{\tilde{A}_{n}-A_{n}}{B_{n}}\overset{p}{\rightarrow }0\text{, (c) }%
\frac{A_{n}}{B_{n}}=O_{p}\left( 1\right) \text{, (d) }C_{n}\frac{\tilde{B}%
_{n}-B_{n}}{B_{n}}\overset{p}{\rightarrow }0.  \label{abcd}
\end{equation}%
Note that (d) implies (a) because $C_{n}\geq \sqrt{n}E\left[ k_{ni}\cdot
k_{ni}\left/ \bar{k}\right. \right] \left/ \sqrt{Ek_{ni}^{2}}\right. =\left. 
\sqrt{nEk_{ni}^{2}}\right/ \bar{k}\rightarrow \infty $ by Assumption 5.(ii).
Next I will prove (b), (c) and (d), respectively.

For (b), the left-hand side can be further decomposed as%
\begin{equation*}
C_{n}\frac{\tilde{A}_{n}-A_{n}}{B_{n}}=\left( \frac{Ek_{ni}}{\left( \left.
1\right/ n\right) \sum_{i}D_{i}k_{ni}}\right) \left[ \underset{\text{(I)}}{%
\underbrace{\frac{\sum_{i}\left( \mu _{0}+U_{i}\right) D_{i}\left( \tilde{k}%
_{ni}-k_{ni}\right) }{\sqrt{nEk_{ni}^{2}}}}}-\underset{\text{(II)}}{%
\underbrace{\frac{\left( \hat{\theta}-\theta _{0}\right) ^{\prime
}\sum_{i}Z_{i}D_{i}\tilde{k}_{ni}}{\sqrt{nEk_{ni}^{2}}}}}\right] .
\end{equation*}%
It follows from Lemmas \ref{Lemma:E} and \ref{Lemma:sample} that $%
Ek_{ni}\left/ \left[ \left( \left. 1\right/ n\right) \sum_{i}D_{i}k_{ni}%
\right] \right. \overset{p}{\rightarrow }1$. For the first term in the
bracket, a Taylor expansion about $\beta _{0}$ yields%
\begin{eqnarray*}
\left| \text{(I)}\right| &\leq &\left| \sum_{i}\frac{\left( \mu
_{0}+U_{i}\right) D_{i}}{\sqrt{nEk_{ni}^{2}}}k^{\prime }\left( \frac{%
1-F\left( W_{i}\right) }{h_{n}}\right) \frac{f\left( W_{i}\right) }{h_{n}}%
\left( \hat{W}_{i}-W_{i}\right) \right| \\
&&+\left| \frac{1}{2}\sum_{i}\frac{\left( \mu _{0}+U_{i}\right) D_{i}}{\sqrt{%
nEk_{ni}^{2}}}k^{\prime \prime }\left( \frac{1-F\left( W_{i}^{\ast }\right) 
}{h_{n}}\right) \left[ \frac{f\left( W_{i}^{\ast }\right) }{h_{n}}\right]
^{2}\left( \hat{W}_{i}-W_{i}\right) ^{2}\right| \\
&&+\left| \frac{1}{2}\sum_{i}\frac{\left( \mu _{0}+U_{i}\right) D_{i}}{\sqrt{%
nEk_{ni}^{2}}}k^{\prime }\left( \frac{1-F\left( W_{i}^{\ast }\right) }{h_{n}}%
\right) \frac{f^{\prime }\left( W_{i}^{\ast }\right) }{h_{n}}\left( \hat{W}%
_{i}-W_{i}\right) ^{2}\right| \\
&\leq &\frac{\sqrt{n}\left\| \hat{\beta}-\beta _{0}\right\| }{\sqrt{%
Ek_{ni}^{2}}}\frac{1}{n}\sum_{i}\left( \left| \mu _{0}\right| +\left|
U_{i}\right| \right) \left| k^{\prime }\left( \frac{1-F\left( W_{i}\right) }{%
h_{n}}\right) \right| \frac{f\left( W_{i}\right) }{h_{n}}\left\|
X_{i}\right\| \\
&&+\frac{\left( \sqrt{n}\left\| \hat{\beta}-\beta _{0}\right\| \right) ^{2}}{%
2\sqrt{nEk_{ni}^{2}}}\frac{1}{n}\sum_{i}\left( \left| \mu _{0}\right|
+\left| U_{i}\right| \right) \left| k^{\prime \prime }\left( \frac{1-F\left(
W_{i}^{\ast }\right) }{h_{n}}\right) \right| \left[ \frac{f\left(
W_{i}^{\ast }\right) }{h_{n}}\right] ^{2}\left\| X_{i}\right\| ^{2} \\
&&+\frac{\left( \sqrt{n}\left\| \hat{\beta}-\beta _{0}\right\| \right) ^{2}}{%
2\sqrt{nEk_{ni}^{2}}}\frac{1}{n}\sum_{i}\left( \left| \mu _{0}\right|
+\left| U_{i}\right| \right) \left| k^{\prime }\left( \frac{1-F\left(
W_{i}^{\ast }\right) }{h_{n}}\right) \right| \frac{\left| f^{\prime }\left(
W_{i}^{\ast }\right) \right| }{h_{n}}\left\| X_{i}\right\| ^{2},
\end{eqnarray*}%
where $W_{i}^{\ast }=X_{i}^{\prime }\beta ^{\ast }$ with $\beta ^{\ast }$\
lying on the line segment joining $\hat{\beta}$ and $\beta _{0}$. Note that $%
\left| k^{\prime }\left( u\right) \right| \leq \bar{k^{\prime }}1\left\{
0\leq u\leq 1\right\} $ and $\left| k^{\prime \prime }\left( u\right)
\right| \leq \bar{k^{\prime \prime }}1\left\{ 0\leq u\leq 1\right\} $ by
Assumption 2. For large enough $n$ such that $h_{n}\leq 1-F\left( C\right) $%
, we have%
\begin{eqnarray*}
\left| \text{(I)}\right| &\leq &\frac{\bar{k^{\prime }}\sqrt{n}\left\| \hat{%
\beta}-\beta _{0}\right\| }{\sqrt{Ek_{ni}^{2}}}\frac{1}{n}\sum_{i}\left(
\left| \mu _{0}\right| +\left| U_{i}\right| \right) 1\left\{ F\left(
W_{i}\right) >1-h_{n}\right\} H_{C}\left( W_{i}\right) \left\| X_{i}\right\|
\\
&&+\frac{\left( \sqrt{n}\left\| \hat{\beta}-\beta _{0}\right\| \right) ^{2}}{%
2\sqrt{nEk_{ni}^{2}}}\frac{1}{n}\sum_{i}\left( \left| \mu _{0}\right|
+\left| U_{i}\right| \right) \left[ \bar{k^{\prime \prime }}H_{C}^{2}\left(
W_{i}^{\ast }\right) +\bar{k^{\prime }}\left| S_{C}\left( W_{i}^{\ast
}\right) \right| H_{C}\left( W_{i}^{\ast }\right) \right] \left\|
X_{i}\right\| ^{2},
\end{eqnarray*}%
where $H_{C}\left( \cdot \right) $ and $S_{C}\left( \cdot \right) $ are
defined in Assumption 3.(iii). Note that $\left. \left( \left. 1\right/
n\right) \sum_{i=1}^{n}\left| G_{ni}\right| \right/ E\left| G_{ni}\right|
=O_{p}\left( 1\right) $ by Markov's inequality for any i.i.d. random
variables $\left\{ G_{ni}:i\leq n\right\} $ with $0<E\left| G_{ni}\right|
<\infty $. Thus, it follows from Assumption 4 that, with probability tending
to one,%
\begin{eqnarray*}
\left| \text{(I)}\right| &\leq &\frac{O_{p}\left( 1\right) }{\sqrt{%
Ek_{ni}^{2}}}\left( \left| \mu _{0}\right| +E\left| U_{i}\right| \right) E%
\left[ 1\left\{ F\left( W_{i}\right) >1-h_{n}\right\} H_{C}\left(
W_{i}\right) \left\| X_{i}\right\| \right] \\
&&+\frac{O_{p}\left( 1\right) }{\sqrt{nEk_{ni}^{2}}}\left( \left| \mu
_{0}\right| +E\left| U_{i}\right| \right) E\left[ \sup_{\beta \in \mathcal{N}%
\left( \beta _{0}\right) }\left[ H_{C}^{2}\left( X_{i}\beta \right) +\left|
S_{C}\left( X_{i}\beta \right) \right| H_{C}\left( X_{i}\beta \right) \right]
\left\| X_{i}\right\| ^{2}\right] \\
&\leq &O_{p}\left( 1\right) \sqrt{\frac{\left[ \Pr \left( F\left(
W_{i}\right) >1-h_{n}\right) \right] ^{1+c_{1}\left/ \left( 8+2c_{1}\right)
\right. }}{Ek_{ni}^{2}}}\left[ EH_{C}^{4}\left( W_{i}\right) \right]
^{1\left/ 4\right. }\left( E\left\| X_{i}\right\| ^{4+c_{1}}\right)
^{1\left/ \left( 4+c_{1}\right) \right. } \\
&&+\frac{O_{p}\left( 1\right) }{\sqrt{nEk_{ni}^{2}}}\left( E\sup_{\beta \in 
\mathcal{N}\left( \beta _{0}\right) }H_{C}^{4}\left( X_{i}\beta \right)
\right) ^{1\left/ 2\right. }\left( E\left\| X_{i}\right\| ^{4}\right)
^{1\left/ 2\right. } \\
&&+\frac{O_{p}\left( 1\right) }{\sqrt{nEk_{ni}^{2}}}\left( E\sup_{\beta \in 
\mathcal{N}\left( \beta _{0}\right) }S_{C}^{4}\left( X_{i}\beta \right)
\right) ^{1\left/ 4\right. }\left( E\sup_{\beta \in \mathcal{N}\left( \beta
_{0}\right) }H_{C}^{4}\left( X_{i}\beta \right) \right) ^{1\left/ 4\right.
}\left( E\left\| X_{i}\right\| ^{4}\right) ^{1\left/ 2\right. } \\
&=&o_{p}\left( 1\right) ,
\end{eqnarray*}%
where the second inequality follows from H\"{o}lder's inequality and the
last equality follows from Assumptions 1.(iv), 3.(iii) and 5.

For the second term of (b), similarly, we have%
\begin{eqnarray*}
\left| \text{(II)}\right| &\leq &\frac{\sqrt{n}\left\| \hat{\theta}-\theta
_{0}\right\| }{\sqrt{Ek_{ni}^{2}}}\frac{1}{n}\sum_{i=1}^{n}\left\|
Z_{i}\right\| \left( k_{ni}+\left| \tilde{k}_{ni}-k_{ni}\right| \right) \\
&\leq &\frac{O_{p}\left( 1\right) }{\sqrt{Ek_{ni}^{2}}}\frac{1}{n}%
\sum_{i=1}^{n}\left\| Z_{i}\right\| \left[ \bar{k}1\left\{ F\left(
W_{i}\right) >1-h_{n}\right\} +\bar{k^{\prime }}H_{C}\left( W_{i}^{\ast
}\right) \left| \hat{W}_{i}-W_{i}\right| \right] \\
&\leq &\frac{O_{p}\left( 1\right) }{\sqrt{Ek_{ni}^{2}}}E\left[ \left\|
Z_{i}\right\| 1\left\{ F\left( W_{i}\right) >1-h_{n}\right\} \right] +\frac{%
O_{p}\left( 1\right) }{\sqrt{nEk_{ni}^{2}}}E\left[ \left\| Z_{i}\right\|
\left\| X_{i}\right\| \sup_{\beta \in \mathcal{N}\left( \beta _{0}\right)
}H_{C}\left( X_{i}\beta \right) \right] \\
&\leq &O_{p}\left( 1\right) \sqrt{\frac{\left[ \Pr \left( F\left(
W_{i}\right) >1-h_{n}\right) \right] ^{1+c_{1}\left/ \left( 2+c_{1}\right)
\right. }}{Ek_{ni}^{2}}}\left( E\left\| Z_{i}\right\| ^{2+c_{1}}\right)
^{1\left/ \left( 2+c_{1}\right) \right. } \\
&&+\frac{O_{p}\left( 1\right) }{\sqrt{nEk_{ni}^{2}}}\left( E\left\|
Z_{i}\right\| ^{2}\right) ^{1\left/ 2\right. }\left( E\left\| X_{i}\right\|
^{4}\right) ^{1\left/ 4\right. }\left( E\sup_{\beta \in \mathcal{N}\left(
\beta _{0}\right) }H_{C}^{4}\left( X_{i}\beta \right) \right) ^{1\left/
4\right. } \\
&=&o_{p}\left( 1\right) ,
\end{eqnarray*}%
which completes the proof of (b).

For (c) in (\ref{abcd}), we know by (\ref{NormalofInfeasible}) that%
\begin{equation*}
\frac{A_{n}}{B_{n}}=\mu _{n}=\mu _{0}+\frac{EU_{i}D_{i}k_{ni}}{ED_{i}k_{ni}}%
+O_{p}\left( \frac{1}{C_{n}}\right) .
\end{equation*}%
It follows from Lemma \ref{Lemma:bias} and $C_{n}\rightarrow \infty $ that $%
\left. A_{n}\right/ B_{n}=\mu _{0}+o\left( 1\right) +o_{p}\left( 1\right)
=O_{p}\left( 1\right) $. The proof of (d) in (\ref{abcd}) follows exactly
the same line as the analysis for the first term (I) of (b), by replacing $%
\mu _{0}+U_{i}$ with $1$. Consequently, we have $C_{n}\left( \tilde{\mu}-\mu
_{n}\right) \overset{p}{\rightarrow }0$, which in combination of (\ref%
{NormalofInfeasible}) leads to the conclusion.

\subsection{Proof of Theorem \protect\ref{Thm:FLC}}

\medskip

Denote $\hat{k}_{ni}=k\left( \left. \left( 1-\hat{F}_{n}\left( \hat{W}%
_{i}\right) \right) \right/ h_{n}\right) $ and $k_{ni}=k\left( \left. \left(
1-F_{W}\left( W_{i}\right) \right) \right/ h_{n}\right) $. By (\ref%
{NormalofInfeasible}), Equation (\ref{knir}), the Lemma, and Assumption
5'.(ii), we have%
\begin{equation}
\sqrt{nh_{n}}\left( \mu _{n}-\mu _{0}+\frac{\kappa _{1}g\left( 1\right) }{%
\kappa _{0}}h_{n}\right) \rightarrow N\left( 0,\frac{\chi _{0}\sigma _{U}^{2}%
}{\kappa _{0}^{2}}\right) ,  \label{NormalofInfeasible2}
\end{equation}%
where%
\begin{equation*}
\mu _{n}=\frac{\sum_{i=1}^{n}\left( Y_{i}-Z_{i}^{\prime }\theta _{0}\right)
D_{i}k_{ni}}{\sum_{i=1}^{n}D_{i}k_{ni}}.
\end{equation*}%
It remains to show that $\sqrt{nh_{n}}\left( \hat{\mu}-\mu _{n}\right) 
\overset{p}{\rightarrow }0$.

Denote%
\begin{eqnarray*}
\hat{A}_{n} &=&\sum_{i=1}^{n}\left( Y_{i}-Z_{i}^{\prime }\hat{\theta}\right)
D_{i}\hat{k}_{ni},\text{ }\hat{B}_{n}=\sum_{i=1}^{n}D_{i}\hat{k}_{ni}, \\
A_{n} &=&\sum_{i=1}^{n}\left( Y_{i}-Z_{i}^{\prime }\theta _{0}\right)
D_{i}k_{ni},\text{ }B_{n}=\sum_{i=1}^{n}D_{i}k_{ni},
\end{eqnarray*}%
then%
\begin{eqnarray*}
\sqrt{nh_{n}}\left( \hat{\mu}-\mu _{n}\right) &=&\sqrt{nh_{n}}\left( \frac{%
\hat{A}_{n}}{\hat{B}_{n}}-\frac{A_{n}}{B_{n}}\right) \\
&=&\sqrt{nh_{n}}\left( \frac{\hat{A}_{n}-A_{n}}{B_{n}}-\frac{\hat{B}%
_{n}-B_{n}}{B_{n}}\cdot \frac{A_{n}}{B_{n}}\right) \frac{B_{n}}{\hat{B}_{n}}
\\
&=&\sqrt{nh_{n}}\left( \frac{\hat{A}_{n}-A_{n}}{B_{n}}-\frac{\hat{B}%
_{n}-B_{n}}{B_{n}}\mu _{0}-o_{p}\left( \frac{\hat{B}_{n}-B_{n}}{B_{n}}%
\right) \right) \frac{B_{n}}{\hat{B}_{n}},
\end{eqnarray*}%
where the last equality follows from the consistency of $\left. A_{n}\right/
B_{n}=\mu _{n}$ for $\mu _{0}$ implied by (\ref{NormalofInfeasible2}). It is
sufficient to show that%
\begin{equation*}
\text{(a) }\frac{\hat{B}_{n}}{B_{n}}\overset{p}{\rightarrow }1\text{, (b) }%
\sqrt{nh_{n}}\left( \frac{\hat{A}_{n}-A_{n}}{B_{n}}-\frac{\hat{B}_{n}-B_{n}}{%
B_{n}}\mu _{0}\right) \overset{p}{\rightarrow }0\text{, (c) }\sqrt{nh_{n}}%
\left( \frac{\hat{B}_{n}-B_{n}}{B_{n}}\right) =O_{p}\left( 1\right) \text{.}
\end{equation*}%
Note that (c) implies (a). For (b) and (c), the left-hand sides can be
written as%
\begin{eqnarray}
&&\sqrt{nh_{n}}\left( \frac{\hat{A}_{n}-A_{n}}{B_{n}}-\frac{\hat{B}_{n}-B_{n}%
}{B_{n}}\mu _{0}\right)  \notag \\
&=&\left( \frac{1}{\kappa _{0}}+o_{p}\left( 1\right) \right) \left[ \underset%
{\text{(b1)}}{\underbrace{\frac{1}{\sqrt{nh_{n}}}\sum_{i=1}^{n}U_{i}D_{i}%
\left( \hat{k}_{ni}-k_{ni}\right) }}-\underset{\text{(b2)}}{\underbrace{%
\frac{\left( \hat{\theta}-\theta _{0}\right) ^{\prime }}{\sqrt{nh_{n}}}%
\sum_{i=1}^{n}Z_{i}D_{i}\hat{k}_{ni}}}\right]  \label{term_b}
\end{eqnarray}%
and%
\begin{equation}
\sqrt{nh_{n}}\left( \frac{\hat{B}_{n}-B_{n}}{B_{n}}\right) =\left( \frac{1}{%
\kappa _{0}}+o_{p}\left( 1\right) \right) \left[ \underset{\text{(c1)}}{%
\underbrace{\frac{1}{\sqrt{nh_{n}}}\sum_{i=1}^{n}D_{i}\left( \hat{k}%
_{ni}-k_{ni}\right) }}\right] ,  \label{term_c}
\end{equation}%
because $B_{n}=nh_{n}\left( \kappa _{0}+o_{p}\left( 1\right) \right) $ by
Lemma \ref{Lemma:Lambda}.

First consider $\left( 1\left/ \sqrt{nh_{n}}\right. \right) \sum_{i}\pi
_{i}D_{i}\left( \hat{k}_{ni}-k_{ni}\right) $, where $\pi _{i}$ is either $%
U_{i}$ or $1$, corresponding to (b1) and (c1), respectively. A sixth-order
Taylor expansion yields%
\begin{eqnarray*}
\hat{k}_{ni}-k_{ni} &=&\sum_{r=1}^{5}\frac{1}{r!h_{n}^{r}}k^{\left( r\right)
}\left( \frac{1-F_{W}\left( W_{i}\right) }{h_{n}}\right) \left[ \left( 1-%
\hat{F}_{n}\left( \hat{W}_{i}\right) \right) -\left( 1-F_{W}\left(
W_{i}\right) \right) \right] ^{r} \\
&&+\frac{1}{6!h_{n}^{6}}k^{\left( 6\right) }\left( \frac{1-F_{i}^{\ast }}{%
h_{n}}\right) \left[ \left( 1-\hat{F}_{n}\left( \hat{W}_{i}\right) \right)
-\left( 1-F_{W}\left( W_{i}\right) \right) \right] ^{6},
\end{eqnarray*}%
where $F_{i}^{\ast }$ lies between $\hat{F}_{n}\left( \hat{W}_{i}\right) $
and $F_{W}\left( W_{i}\right) $. Denote $\alpha _{ij}\left( \beta \right)
=1\left\{ X_{ij}^{\prime }\beta <0\right\} -1\left\{ X_{ij}^{\prime }\beta
_{0}<0\right\} $, where $X_{ij}=X_{i}-X_{j}$, then%
\begin{equation}
\left( 1-\hat{F}_{n}\left( \hat{W}_{i}\right) \right) -\left( 1-F_{W}\left(
W_{i}\right) \right) =\frac{1}{n-1}\sum_{j\neq i}\left\{ \left[ 1\left\{
W_{j}>W_{i}\right\} -\left( 1-F_{W}\left( W_{i}\right) \right) \right]
+\alpha _{ij}\left( \hat{\beta}\right) \right\} .  \label{survival_decom}
\end{equation}%
It follows from the Taylor expansion and the above equality that%
\begin{equation}
\frac{1}{\sqrt{nh_{n}}}\sum_{i=1}^{n}\pi _{i}D_{i}\left( \hat{k}%
_{ni}-k_{ni}\right) =\text{(I)}+\text{(II)}+\text{(III)}+\text{(IV)},
\label{total}
\end{equation}%
where 
\begin{eqnarray*}
\text{(I)} &=&\sqrt{\frac{n}{h_{n}^{3}}}\frac{1}{n\left( n-1\right) }%
\sum_{i}\sum_{j\neq i}\pi _{i}D_{i}k^{\prime }\left( \frac{1-F_{W}\left(
W_{i}\right) }{h_{n}}\right) \left[ 1\left\{ W_{j}>W_{i}\right\} -\left(
1-F_{W}\left( W_{i}\right) \right) \right] \\
\text{(II)} &=&\frac{1}{\left( n-1\right) \sqrt{nh_{n}^{3}}}%
\sum_{i}\sum_{j\neq i}\pi _{i}D_{i}k^{\prime }\left( \frac{1-F_{W}\left(
W_{i}\right) }{h_{n}}\right) \alpha _{ij}\left( \hat{\beta}\right) \\
\text{(III)} &=&\sum_{r=2}^{5}\sum_{s=0}^{r}\frac{1}{s!\left( r-s\right) !}%
\frac{1}{\sqrt{nh_{n}^{2r+1}}}\sum_{i=1}^{n}\pi _{i}D_{i}k^{\left( r\right)
}\left( \frac{1-F_{W}\left( W_{i}\right) }{h_{n}}\right) \\
&&\cdot \left\{ \frac{1}{n-1}\sum_{j\neq i}\left[ 1\left\{
W_{j}>W_{i}\right\} -\left( 1-F_{W}\left( W_{i}\right) \right) \right]
\right\} ^{r-s}\left\{ \frac{1}{n-1}\sum_{j\neq i}\alpha _{ij}\left( \hat{%
\beta}\right) \right\} ^{s} \\
\text{(IV)} &=&\frac{1}{6!\sqrt{nh_{n}^{13}}}\sum_{i=1}^{n}\pi
_{i}D_{i}k^{\left( 6\right) }\left( \frac{1-F_{i}^{\ast }}{h_{n}}\right) %
\left[ \left( 1-\hat{F}_{n}\left( \hat{W}_{i}\right) \right) -\left(
1-F_{W}\left( W_{i}\right) \right) \right] ^{6}.
\end{eqnarray*}%
I will consider these four terms, respectively.

Denote $\xi _{i}$ as the $i$-th observation and%
\begin{equation*}
m_{n}\left( \xi _{i},\xi _{j}\right) =h_{n}^{-3/2}\pi _{i}D_{i}k^{\prime
}\left( \frac{1-F_{W}\left( W_{i}\right) }{h_{n}}\right) \left[ 1\left\{
W_{j}>W_{i}\right\} -\left( 1-F_{W}\left( W_{i}\right) \right) \right] ,
\end{equation*}%
then%
\begin{equation*}
\text{(I)}=\frac{\sqrt{n}}{n\left( n-1\right) }\sum_{i}\sum_{j\neq
i}m_{n}\left( \xi _{i},\xi _{j}\right) .
\end{equation*}%
Note that $E\left[ \left. m_{n}\left( \xi _{i},\xi _{j}\right) \right| \xi
_{i}\right] =0$ and $E\left[ m_{n}\left( \xi _{i},\xi _{j}\right) \right] =0$%
. Since%
\begin{equation*}
E\left[ m_{n}^{2}\left( \xi _{i},\xi _{j}\right) \right] \leq
h_{n}^{-3}\left( \bar{k}^{\prime }\right) ^{2}E\pi _{i}^{2}\Pr \left(
F_{W}\left( W_{i}\right) >1-h_{n}\right) =O\left( h_{n}^{-2}\right) =o\left(
n\right)
\end{equation*}%
by Assumption 5', it follows from the projection method for U-statistics %
\citep[e.g.,][Lemma 3.1]{powell1989semiparametric} that%
\begin{equation*}
\text{(I)}=\frac{1}{\sqrt{n}}\sum_{j=1}^{n}E\left[ \left. m_{n}\left( \xi
_{i},\xi _{j}\right) \right| \xi _{j}\right] +o_{p}\left( 1\right) .
\end{equation*}%
Denote $G_{\pi }\left( t\right) =E\left[ \pi _{i}D_{i}\left| F_{W}\left(
W_{i}\right) =t\right. \right] =E\left[ \pi _{i}1\left\{ F_{W}\left(
\varepsilon _{i}\right) <t\right\} \right] $ with $\pi $ being either $U$ or 
$1$, then%
\begin{eqnarray*}
E\left[ \left. m_{n}\left( \xi _{i},\xi _{j}\right) \right| \xi _{j}\right]
&=&h_{n}^{-3/2}E\left[ \left. G_{\pi }\left( F_{W}\left( W_{i}\right)
\right) k^{\prime }\left( \frac{1-F_{W}\left( W_{i}\right) }{h_{n}}\right) %
\left[ 1\left\{ W_{j}>W_{i}\right\} -\left( 1-F_{W}\left( W_{i}\right)
\right) \right] \right| \xi _{j}\right] \\
&=&h_{n}^{-3/2}\int_{0}^{1}G_{\pi }\left( t\right) k^{\prime }\left( \frac{%
1-t}{h_{n}}\right) \left[ 1\left\{ t<F_{W}\left( W_{j}\right) \right\}
-\left( 1-t\right) \right] dt \\
&=&h_{n}^{-1/2}\int_{0}^{1/h_{n}}G_{\pi }\left( 1-h_{n}s\right) k^{\prime
}\left( s\right) \left[ 1\left\{ 1-h_{n}s<F_{W}\left( W_{j}\right) \right\}
-h_{n}s\right] ds \\
&=&h_{n}^{-1/2}\int_{0}^{1}\left[ G_{\pi }\left( 1-h_{n}s\right) -G_{\pi
}\left( 1\right) \right] k^{\prime }\left( s\right) \left[ 1\left\{
1-F_{W}\left( W_{j}\right) <h_{n}s\right\} -h_{n}s\right] ds \\
&&-h_{n}^{-1/2}G_{\pi }\left( 1\right) \left( k_{nj}-h_{n}\kappa _{0}\right)
,
\end{eqnarray*}%
therefore%
\begin{equation}
\text{(I)}=-\frac{G_{\pi }\left( 1\right) }{\sqrt{nh_{n}}}%
\sum_{i=1}^{n}\left( k_{ni}-h_{n}\kappa _{0}\right) +R_{1n}+o_{p}\left(
1\right) ,  \label{(I)}
\end{equation}%
where%
\begin{eqnarray*}
R_{1n} &=&\int_{0}^{1}\left[ G_{\pi }\left( 1-h_{n}s\right) -G_{\pi }\left(
1\right) \right] k^{\prime }\left( s\right) \frac{1}{\sqrt{nh_{n}}}%
\sum_{i=1}^{n}\left[ 1\left\{ 1-F_{W}\left( W_{i}\right) <h_{n}s\right\}
-h_{n}s\right] ds, \\
\left| R_{1n}\right| &\leq &\bar{k}^{\prime }\int_{0}^{1}\left| G_{\pi
}\left( 1-h_{n}s\right) -G_{\pi }\left( 1\right) \right| ds\cdot \sup_{s\in %
\left[ 0,1\right] }\left| \frac{1}{\sqrt{nh_{n}}}\sum_{i=1}^{n}\left[
1\left\{ 1-F_{W}\left( W_{i}\right) <h_{n}s\right\} -h_{n}s\right] \right| .
\end{eqnarray*}%
Since $\left| G_{\pi }\left( 1-h_{n}s\right) -G_{\pi }\left( 1\right)
\right| \leq E\left| \pi _{i}\right| =O\left( 1\right) $ and $%
\lim_{n\rightarrow \infty }\left| G_{\pi }\left( 1-h_{n}s\right) -G_{\pi
}\left( 1\right) \right| =0$ for any $s\in \left[ 0,1\right] $, it follows
from the dominated convergence theorem 
\citep[e.g.,][Theorem
1.1.(ii)]{shao2003} that $\lim_{n\rightarrow \infty }\int_{0}^{1}\left|
G_{\pi }\left( 1-h_{n}s\right) -G_{\pi }\left( 1\right) \right| ds=0$. For
the supremum term, since $1-F_{W}\left( W_{i}\right) $ follows a uniform
distribution on the unit interval, it follows from 
\citet[][Lemma
2.3]{stute1982oscillation} that%
\begin{equation*}
\sup_{u\in \left[ 0,1\right] }\left| \frac{1}{\sqrt{nh_{n}}}\sum_{i=1}^{n}%
\left[ 1\left\{ 1-F_{W}\left( W_{i}\right) <h_{n}s\right\} -h_{n}s\right]
\right| =O_{p}\left( 1\right) ,
\end{equation*}%
thus the remainder term satisfies%
\begin{equation}
\left| R_{1n}\right| =o\left( 1\right) \cdot O_{p}\left( 1\right)
=o_{p}\left( 1\right) .  \label{R1n}
\end{equation}%
For the first term of (\ref{(I)}), since $Ek_{ni}=h_{n}\kappa _{0}$ and $%
Var\left( k_{ni}\right) =Ek_{ni}^{2}-\left( Ek_{ni}\right) ^{2}=h_{n}\chi
_{0}-\left( h_{n}\kappa _{0}\right) ^{2}$, we have%
\begin{equation}
E\left[ \frac{1}{\sqrt{nh_{n}}}\sum_{i=1}^{n}\left( k_{ni}-h_{n}\kappa
_{0}\right) \right] ^{2}=\frac{Var\left( k_{ni}\right) }{h_{n}}=\chi
_{0}-h_{n}\kappa _{0}^{2}=O\left( 1\right) .  \label{Var_kni}
\end{equation}%
Substituting (\ref{R1n}) and (\ref{Var_kni}) into (\ref{(I)}) gives%
\begin{equation}
\text{(I)}=O_{p}\left( 1\right) \cdot G_{\pi }\left( 1\right) +o_{p}\left(
1\right) =\left\{ 
\begin{array}{cc}
o_{p}\left( 1\right) & \text{for }\pi _{i}=U_{i} \\ 
O_{p}\left( 1\right) & \text{for }\pi _{i}=1,%
\end{array}%
\right.  \label{part1}
\end{equation}%
by noting that $G_{\pi }\left( 1\right) =E\pi _{i}$.

Next consider (II) and (III). For (III), we have 
\begin{eqnarray*}
\left| \text{(III)}\right| &\leq &\sum_{r=2}^{5}\sum_{s=0}^{r}\frac{1}{%
s!\left( r-s\right) !}\frac{\overline{k^{\left( r\right) }}}{\sqrt{%
nh_{n}^{2r+1}}}\left( \frac{n}{n-1}\right) ^{r-s} \\
&&\cdot \left( \sup_{0<w<h_{n}}\left| \frac{1}{n}\sum_{j=1}^{n}\left[
1\left\{ 1-F_{W}\left( W_{j}\right) <w\right\} -w\right] \right| \right)
^{r-s} \\
&&\cdot \sum_{i=1}^{n}\left| \pi _{i}\right| 1\left\{ F_{W}\left(
W_{i}\right) >1-h_{n}\right\} \left( \frac{1}{n-1}\sum_{j\neq i}\left|
\alpha _{ij}\left( \hat{\beta}\right) \right| \right) ^{s}.
\end{eqnarray*}%
It follows from \citet[][Lemma 2.3]{stute1982oscillation} that%
\begin{equation*}
\sup_{0<w<h_{n}}\left| \frac{1}{n}\sum_{j=1}^{n}\left[ 1\left\{
1-F_{W}\left( W_{j}\right) <w\right\} -w\right] \right| =O_{p}\left( \sqrt{%
\frac{h_{n}}{n}}\right) ,
\end{equation*}%
so that%
\begin{equation*}
\left| \text{(III)}\right| \leq \sum_{r=2}^{5}\sum_{s=0}^{r}O_{p}\left(
\left( nh_{n}\right) ^{-\left. \left( r+s+1\right) \right/ 2}\right)
\sum_{i=1}^{n}\left| \pi _{i}\right| 1\left\{ F_{W}\left( W_{i}\right)
>1-h_{n}\right\} \left( \sum_{j\neq i}\left| \alpha _{ij}\left( \hat{\beta}%
\right) \right| \right) ^{s}.
\end{equation*}%
For the terms of $s=0$, since $\left. \left( \left. 1\right/ n\right)
\sum_{i=1}^{n}\left| G_{ni}\right| \right/ E\left| G_{ni}\right|
=O_{p}\left( 1\right) $ for any i.i.d. random variables $\left\{
G_{ni}:i\leq n\right\} $ with $0<E\left| G_{ni}\right| <\infty $,%
\begin{eqnarray*}
&&O_{p}\left( \left( nh_{n}\right) ^{-\left. \left( r+1\right) \right/
2}\right) \sum_{i=1}^{n}\left| \pi _{i}\right| 1\left\{ F_{W}\left(
W_{i}\right) >1-h_{n}\right\} \\
&=&O_{p}\left( \left( nh_{n}\right) ^{-\left. \left( r+1\right) \right/
2}\right) \cdot O_{p}\left( n\right) E\left| \pi _{i}\right| \Pr \left(
F_{W}\left( W_{i}\right) >1-h_{n}\right) \\
&=&O_{p}\left( \left( nh_{n}\right) ^{-\left. \left( r-1\right) \right/
2}\right) =o_{p}\left( 1\right) .
\end{eqnarray*}%
For the terms of $s\geq 1$, since $\left| \alpha _{ij}\left( \hat{\beta}%
\right) \right| ^{s}=\left| \alpha _{ij}\left( \hat{\beta}\right) \right| $,%
\begin{eqnarray*}
&&\left( \sum_{j\neq i}\left| \alpha _{ij}\left( \hat{\beta}\right) \right|
\right) ^{s}=\sum_{j_{1}\neq i}\left| \alpha _{ij_{1}}\left( \hat{\beta}%
\right) \right| \left[ \left| \alpha _{ij_{1}}\left( \hat{\beta}\right)
\right| +\sum_{j\neq i,j\neq j_{1}}\left| \alpha _{ij}\left( \hat{\beta}%
\right) \right| \right] ^{s-1} \\
&=&\sum_{j_{1}\neq i}\left| \alpha _{ij_{1}}\left( \hat{\beta}\right)
\right| \sum_{t_{1}=0}^{s-1}\binom{s-1}{t_{1}}\left| \alpha _{ij_{1}}\left( 
\hat{\beta}\right) \right| ^{s-1-t_{1}}\left[ \sum_{j\neq i,j_{1}}\left|
\alpha _{ij}\left( \hat{\beta}\right) \right| \right] ^{t_{1}} \\
&=&\sum_{j_{1}\neq i}\left| \alpha _{ij_{1}}\left( \hat{\beta}\right)
\right| +\sum_{t_{1}=1}^{s-1}\binom{s-1}{t_{1}}\sum_{j_{1}\neq i}\left|
\alpha _{ij_{1}}\left( \hat{\beta}\right) \right| \sum_{j_{2}\neq
i,j_{1}}\left| \alpha _{ij_{2}}\left( \hat{\beta}\right) \right| \left[
\left| \alpha _{ij_{2}}\left( \hat{\beta}\right) \right| +\sum_{j\neq
i,j_{1},j_{2}}\left| \alpha _{ij}\left( \hat{\beta}\right) \right| \right]
^{t_{1}-1} \\
&=&\sum_{j_{1}\neq i}\left| \alpha _{ij_{1}}\left( \hat{\beta}\right)
\right| +\sum_{t_{1}=1}^{s-1}\binom{s-1}{t_{1}}\sum_{j_{1}\neq i}\left|
\alpha _{ij_{1}}\left( \hat{\beta}\right) \right| \sum_{j_{2}\neq
i,j_{1}}\left| \alpha _{ij_{2}}\left( \hat{\beta}\right) \right| \\
&&+\sum_{t_{1}=1}^{s-1}\sum_{t_{2}=1}^{t_{1}-1}\binom{s-1}{t_{1}}\binom{%
t_{1}-1}{t_{2}}\sum_{j_{1}\neq i}\left| \alpha _{ij_{1}}\left( \hat{\beta}%
\right) \right| \sum_{j_{2}\neq i,j_{1}}\left| \alpha _{ij_{2}}\left( \hat{%
\beta}\right) \right| \left[ \sum_{j\neq i,j_{1},j_{2}}\left| \alpha
_{ij}\left( \hat{\beta}\right) \right| \right] ^{t_{2}} \\
&=&\cdots =\sum_{t=1}^{s}C_{t}\prod_{t^{\prime }=1}^{t}\sum_{j_{t^{\prime
}}\neq i,j_{1},\cdots ,j_{t^{\prime }-1}}\left| \alpha _{ij_{t^{\prime
}}}\left( \hat{\beta}\right) \right| ,
\end{eqnarray*}%
where $C_{t}$ is a constant depending only on $t$. Therefore,%
\begin{equation*}
\left| \text{(III)}\right| \leq
\sum_{r=2}^{5}\sum_{s=1}^{r}\sum_{t=1}^{s}O_{p}\left( \left( nh_{n}\right)
^{-\frac{r+s+1}{2}}\right) \sum_{i=1}^{n}\left| \pi _{i}\right| 1\left\{
F_{W}\left( W_{i}\right) >1-h_{n}\right\} \prod_{t^{\prime
}=1}^{t}\sum_{j_{t^{\prime }}\neq i,j_{1},\cdots ,j_{t^{\prime }-1}}\left|
\alpha _{ij_{t^{\prime }}}\left( \hat{\beta}\right) \right| +o_{p}\left(
1\right) .
\end{equation*}%
Denote%
\begin{equation}
J_{rsti}\left( \beta \right) =\left( nh_{n}\right) ^{-\left. \left(
r+s+1\right) \right/ 2}\left| \pi _{i}\right| 1\left\{ F_{W}\left(
W_{i}\right) >1-h_{n}\right\} \prod_{t^{\prime }=1}^{t}\sum_{j_{t^{\prime
}}\neq i,j_{1},\cdots ,j_{t^{\prime }-1}}\left| \alpha _{ij_{t^{\prime
}}}\left( \beta \right) \right|  \label{Jrsti}
\end{equation}%
for any $1\leq t\leq s\leq r\leq 5$ and $i=1,\cdots ,n$, then%
\begin{equation*}
\left| \text{(II)}\right| +\left| \text{(III)}\right| \leq O_{p}\left(
1\right)
\sum_{r=1}^{5}\sum_{s=1}^{r}\sum_{t=1}^{s}\sum_{i=1}^{n}J_{rsti}\left( \hat{%
\beta}\right) +o_{p}\left( 1\right) .
\end{equation*}

Following the method of \cite{schafgans2002intercept}, I will prove that for
any $1\leq t\leq s\leq r\leq 5$, we have $\sum_{i}J_{rsti}\left( \hat{\beta}%
\right) =o_{p}\left( 1\right) $, namely $\Pr \left( \sum_{i}J_{rsti}\left( 
\hat{\beta}\right) >\delta \right) \rightarrow 0$ for any $\delta >0$. Let $%
M_{n}$ be a slowly divergent sequence, then%
\begin{eqnarray*}
&&\Pr \left( \sum_{i=1}^{n}J_{rsti}\left( \hat{\beta}\right) >\delta \right)
\\
&\leq &\Pr \left( \sum_{i=1}^{n}J_{rsti}\left( \hat{\beta}\right) >\delta
,\left\Vert \hat{\beta}-\beta _{0}\right\Vert \leq \frac{M_{n}}{\sqrt{n}}%
\right) +\Pr \left( \left\Vert \hat{\beta}-\beta _{0}\right\Vert >\frac{M_{n}%
}{\sqrt{n}}\right) \\
&\leq &\sup_{\left\Vert \beta -\beta _{0}\right\Vert \leq M_{n}\left/ \sqrt{n%
}\right. }\Pr \left( \sum_{i=1}^{n}J_{rsti}\left( \beta \right) >\delta
\right) +o\left( 1\right) \\
&\leq &\frac{n}{\delta }\left( nh_{n}\right) ^{-\frac{r+s+1}{2}}E\left\vert
\pi _{i}\right\vert \sup_{\left\Vert \beta -\beta _{0}\right\Vert \leq \frac{%
M_{n}}{\sqrt{n}}}E\left[ 1\left\{ F_{W}\left( W_{i}\right) >1-h_{n}\right\}
\prod_{t^{\prime }=1}^{t}\sum_{j_{t^{\prime }}\neq i,j_{1},\cdots
,j_{t^{\prime }-1}}\left\vert \alpha _{ij_{t^{\prime }}}\left( \beta \right)
\right\vert \right] +o\left( 1\right) ,
\end{eqnarray*}%
where the second inequality follows from the $\sqrt{n}$-consistency of $\hat{%
\beta}$ and the third inequality follows from Markov's inequality. As in the
proof of Lemma \ref{Lemma:alpha_ij}, we can show that under $\left\Vert
\beta -\beta _{0}\right\Vert \leq M_{n}\left/ \sqrt{n}\right. $, for $%
j_{t^{\prime }}\neq i,j_{1},\cdots ,j_{t^{\prime }-1}$,%
\begin{eqnarray*}
&&E\left[ \left. \left\vert \alpha _{ij_{t^{\prime }}}\left( \beta \right)
\right\vert \right\vert X_{i},X_{j_{1}},\cdots ,X_{j_{t^{\prime }-1}}\right]
=E\left[ \left. E\left[ \left. \left\vert \alpha _{ij}\left( \beta \right)
\right\vert \right\vert X_{i},X_{j,(-1)}\right] \right\vert X_{i}\right] \\
&=&E\left[ \left. 
\begin{array}{c}
1\left\{ X_{ij,(-1)}^{\prime }\left( \beta _{(-1)}-\beta _{0,(-1)}\right)
>0\right\} \int_{W_{i}}^{W_{i}+X_{ij,(-1)}^{\prime }\left( \beta
_{(-1)}-\beta _{0,(-1)}\right) }f_{W\left\vert X_{(-1)}\right. }\left(
w\left\vert X_{j,(-1)}\right. \right) dw \\ 
+1\left\{ X_{ij,(-1)}^{\prime }\left( \beta _{(-1)}-\beta _{0,(-1)}\right)
<0\right\} \int_{W_{i}+X_{ij,(-1)}^{\prime }\left( \beta _{(-1)}-\beta
_{0,(-1)}\right) }^{W_{i}}f_{W\left\vert X_{(-1)}\right. }\left( w\left\vert
X_{j,(-1)}\right. \right) dw%
\end{array}%
\right\vert X_{i}\right] \\
&\leq &\overline{f_{W\left\vert X_{(-1)}\right. }}E\left[ \left. \left\vert
X_{ij,(-1)}^{\prime }\left( \beta _{(-1)}-\beta _{0,(-1)}\right) \right\vert
\right\vert X_{i}\right] \\
&\leq &\frac{\overline{f_{W\left\vert X_{(-1)}\right. }}M_{n}}{\sqrt{n}}%
\left( \left\Vert X_{i}\right\Vert +E\left\Vert X_{i}\right\Vert \right) .
\end{eqnarray*}%
Therefore, under $\left\Vert \beta -\beta _{0}\right\Vert \leq M_{n}\left/ 
\sqrt{n}\right. $, 
\begin{eqnarray*}
&&E\left[ 1\left\{ F_{W}\left( W_{i}\right) >1-h_{n}\right\}
\prod_{t^{\prime }=1}^{t}\sum_{j_{t^{\prime }}\neq i,j_{1},\cdots
,j_{t^{\prime }-1}}\left\vert \alpha _{ij_{t^{\prime }}}\left( \beta \right)
\right\vert \right] \\
&\leq &\frac{\overline{f_{W\left\vert X_{(-1)}\right. }}M_{n}}{\sqrt{n}}%
\left( n-t\right) E\left[ 1\left\{ F_{W}\left( W_{i}\right) >1-h_{n}\right\}
\left( \left\Vert X_{i}\right\Vert +E\left\Vert X_{i}\right\Vert \right)
\prod_{t^{\prime }=1}^{t-1}\sum_{j_{t^{\prime }}\neq i,j_{1},\cdots
,j_{t^{\prime }-1}}\left\vert \alpha _{ij_{t^{\prime }}}\left( \beta \right)
\right\vert \right] \\
&\leq &\cdots \leq \left( \frac{\overline{f_{W\left\vert X_{(-1)}\right. }}%
M_{n}}{\sqrt{n}}\right) ^{t-1}\prod_{t^{\prime }=1}^{t}\left( n-t^{\prime
}\right) E\left[ 1\left\{ F_{W}\left( W_{i}\right) >1-h_{n}\right\} \left(
\left\Vert X_{i}\right\Vert +E\left\Vert X_{i}\right\Vert \right)
^{t-1}\left\vert \alpha _{ij_{1}}\left( \beta \right) \right\vert \right] .
\end{eqnarray*}%
It follows from Lemma \ref{Lemma:alpha_ij} that by properly restricting the
divergence rate of $M_{n}$,%
\begin{eqnarray*}
&&\Pr \left( \sum_{i=1}^{n}J_{rsti}\left( \hat{\beta}\right) >\delta \right)
\\
&\leq &O\left( n\left( nh_{n}\right) ^{-\left. \left( r+s+1\right) \right/
2}\left( \frac{M_{n}}{\sqrt{n}}\right) ^{t-1}n^{t}\right) \left( \frac{M_{n}%
}{\sqrt{nh_{n}^{2}}}\right) ^{-\left( t-1\right) }\left( \sqrt{\frac{n}{%
h_{n}^{3}}}\right) ^{-1}J_{n}^{\left( t-1\right) }\left( E\left\Vert
X_{i}\right\Vert \right) \\
&=&O\left( \left( nh_{n}\right) ^{-\left. \left( r+s-2t\right) \right/
2}\right) \cdot o\left( 1\right) \\
&=&o\left( 1\right)
\end{eqnarray*}%
for any $1\leq t\leq s\leq r\leq 5$. Consequently,%
\begin{equation}
\left\vert \text{(II)}\right\vert +\left\vert \text{(III)}\right\vert
=o_{p}\left( 1\right) .  \label{part23}
\end{equation}

Now consider (IV). Write%
\begin{eqnarray*}
\left( 1-\hat{F}_{n}\left( \hat{W}_{i}\right) \right) -\left( 1-F_{W}\left(
W_{i}\right) \right) &=&\left\{ \frac{1}{n-1}\sum_{j\neq i}1\left\{
X_{ij}^{\prime }\hat{\beta}<0\right\} -\left. E\left[ \left. 1\left\{
X_{ij}^{\prime }\beta <0\right\} \right| X_{i}\right] \right| _{\beta =\hat{%
\beta}}\right\} \\
&&+\left\{ \left. E\left[ \left. 1\left\{ X_{ij}^{\prime }\beta <0\right\}
\right| X_{i}\right] \right| _{\beta =\hat{\beta}}-E\left[ \left. 1\left\{
X_{ij}^{\prime }\beta _{0}<0\right\} \right| X_{i}\right] \right\} .
\end{eqnarray*}%
For the first term, it follows from the VC property of $\left\{ 1\left\{
X_{j}^{\prime }\beta >t\right\} :\beta \in R^{d_{x}},t\in R\right\} $ and
the Donsker theorem that%
\begin{eqnarray*}
&&\sup_{i}\left| \frac{1}{n-1}\sum_{j\neq i}1\left\{ X_{ij}^{\prime }\hat{%
\beta}<0\right\} -\left. E\left[ \left. 1\left\{ X_{ij}^{\prime }\beta
<0\right\} \right| X_{i}\right] \right| _{\beta =\hat{\beta}}\right| \\
&\leq &\sup_{\beta \in R^{d_{x}},t\in R}\left| \frac{1}{n}%
\sum_{j=1}^{n}1\left\{ X_{j}^{\prime }\beta >t\right\} -E\left[ 1\left\{
X_{j}^{\prime }\beta >t\right\} \right] \right| +\frac{1}{n} \\
&=&O_{p}\left( \frac{1}{\sqrt{n}}\right) .
\end{eqnarray*}%
For the second term, since%
\begin{eqnarray*}
E\left[ \left. 1\left\{ X_{ij}^{\prime }\beta <0\right\} \right| X_{i}\right]
&=&E\left[ \left. 1\left\{ W_{j}>W_{i}+X_{ij,(-1)}^{\prime }\left( \beta
_{(-1)}-\beta _{0,(-1)}\right) \right\} \right| X_{i}\right] \\
&=&E\left[ \left. 1-F_{\left. W\right| X_{(-1)}}\left( \left.
W_{i}+X_{ij,(-1)}^{\prime }\left( \beta _{(-1)}-\beta _{0,(-1)}\right)
\right| X_{j,(-1)}\right) \right| X_{i}\right] ,
\end{eqnarray*}%
we have%
\begin{equation*}
\partial _{\beta _{(-1)}}E\left[ \left. 1\left\{ X_{ij}^{\prime }\beta
<0\right\} \right| X_{i}\right] =-E\left[ \left. f_{\left. W\right|
X_{(-1)}}\left( \left. W_{i}+X_{ij,(-1)}^{\prime }\left( \beta _{(-1)}-\beta
_{0,(-1)}\right) \right| X_{j,(-1)}\right) X_{ij,(-1)}\right| X_{i}\right] .
\end{equation*}%
It follows that%
\begin{eqnarray*}
&&\left| \left. E\left[ \left. 1\left\{ X_{ij}^{\prime }\beta <0\right\}
\right| X_{i}\right] \right| _{\beta =\hat{\beta}}-E\left[ \left. 1\left\{
X_{ij}^{\prime }\beta _{0}<0\right\} \right| X_{i}\right] \right| \\
&\leq &\overline{f_{\left. W\right| X_{(-1)}}}\left\| \hat{\beta}%
_{(-1)}-\beta _{0,(-1)}\right\| E\left[ \left. \left\| X_{ij,(-1)}\right\|
\right| X_{i}\right] \\
&\leq &O_{p}\left( \frac{1}{\sqrt{n}}\right) \left( \left\| X_{i}\right\|
+E\left\| X_{i}\right\| \right) .
\end{eqnarray*}%
uniformly over $i=1,\cdots ,n$. Therefore,%
\begin{equation}
\left| \left( 1-\hat{F}_{n}\left( \hat{W}_{i}\right) \right) -\left(
1-F_{W}\left( W_{i}\right) \right) \right| \leq O_{p}\left( \frac{1}{\sqrt{n}%
}\right) \left( \left\| X_{i}\right\| +1\right)  \label{11111}
\end{equation}%
uniformly over $i=1,\cdots ,n$, and by Assumption 5',%
\begin{eqnarray}
\left| \text{(IV)}\right| &\leq &\frac{\overline{k^{\left( 6\right) }}}{6!%
\sqrt{nh_{n}^{13}}}\sum_{i=1}^{n}\left| \pi _{i}\right| \left| \left( 1-\hat{%
F}_{n}\left( \hat{W}_{i}\right) \right) -\left( 1-F_{W}\left( W_{i}\right)
\right) \right| ^{6}  \notag \\
&\leq &O_{p}\left( \frac{1}{n^{3}\sqrt{nh_{n}^{13}}}\right)
\sum_{i=1}^{n}\left| \pi _{i}\right| \left( \left\| X_{i}\right\| +1\right)
^{6}  \notag \\
&=&O_{p}\left( \frac{n}{n^{3}\sqrt{nh_{n}^{13}}}\right) E\left| \pi
_{i}\right| E\left[ \left( \left\| X_{i}\right\| +1\right) ^{6}\right] 
\notag \\
&=&o_{p}\left( 1\right) .  \label{part4}
\end{eqnarray}%
Inserting (\ref{part1}), (\ref{part23}), and (\ref{part4}) into (\ref{total}%
) obtains%
\begin{equation*}
\text{(b1)}=o_{p}\left( 1\right) ,\text{\ (c1)}=O_{p}\left( 1\right) ,
\end{equation*}%
where (b1) and (c1) are defined in (\ref{term_b}) and (\ref{term_c}),
respectively.

It remains to show that (b2) converges to zero in probability. It follows
from (\ref{11111}) and Assumption 5' that%
\begin{eqnarray*}
\left| \text{(b-2)}\right| &\leq &O_{p}\left( \frac{1}{\sqrt{h_{n}}}\right) 
\frac{1}{n}\sum_{i=1}^{n}\left\| Z_{i}\right\| \left( \left| k_{ni}\right|
+\left| \hat{k}_{ni}-k_{ni}\right| \right) \\
&\leq &O_{p}\left( \frac{1}{\sqrt{h_{n}}}\right) \frac{1}{n}%
\sum_{i=1}^{n}\left\| Z_{i}\right\| 1\left\{ F_{W}\left( W_{i}\right)
>1-h_{n}\right\} \\
&&+O_{p}\left( \frac{1}{\sqrt{h_{n}^{3}}}\right) \frac{1}{n}%
\sum_{i=1}^{n}\left\| Z_{i}\right\| 1\left\{ F_{W}\left( W_{i}\right)
>1-h_{n}\right\} \left| \left( 1-\hat{F}_{n}\left( \hat{W}_{i}\right)
\right) -\left( 1-F_{W}\left( W_{i}\right) \right) \right| \\
&&+O_{p}\left( \frac{1}{\sqrt{h_{n}^{5}}}\right) \frac{1}{n}%
\sum_{i=1}^{n}\left\| Z_{i}\right\| \left| \left( 1-\hat{F}_{n}\left( \hat{W}%
_{i}\right) \right) -\left( 1-F_{W}\left( W_{i}\right) \right) \right| ^{2}
\\
&\leq &O_{p}\left( \frac{1}{\sqrt{h_{n}}}\right) \cdot O_{p}\left( E\left[
\left\| Z_{i}\right\| 1\left\{ F_{W}\left( W_{i}\right) >1-h_{n}\right\} %
\right] \right) \\
&&+O_{p}\left( \frac{1}{\sqrt{nh_{n}^{3}}}\right) \frac{1}{n}%
\sum_{i=1}^{n}\left\| Z_{i}\right\| 1\left\{ F_{W}\left( W_{i}\right)
>1-h_{n}\right\} \left( \left\| X_{i}\right\| +1\right) \\
&&+O_{p}\left( \frac{1}{\sqrt{n^{2}h_{n}^{5}}}\right) \frac{1}{n}%
\sum_{i=1}^{n}\left\| Z_{i}\right\| \left( \left\| X_{i}\right\| +1\right)
^{2} \\
&\leq &O_{p}\left( \frac{1}{\sqrt{h_{n}}}\right) \cdot O_{p}\left(
h_{n}^{1-1\left/ \left( 2+c_{1}\right) \right. }\right) +O_{p}\left( \frac{1%
}{\sqrt{nh_{n}^{3}}}\right) \cdot O_{p}\left( h_{n}^{1-1/2-1/6}\right)
+O_{p}\left( \frac{1}{\sqrt{n^{2}h_{n}^{5}}}\right) \\
&=&O_{p}\left( h_{n}^{c_{1}\left/ \left( 4+2c_{1}\right) \right. }\right)
+O_{p}\left( \left( nh_{n}^{7/3}\right) ^{-1/2}\right) +O_{p}\left( \left(
nh_{n}^{5/2}\right) ^{-1}\right) \\
&=&o_{p}\left( 1\right) ,
\end{eqnarray*}%
which completes the proof.

\subsection{Proof of Theorem \protect\ref{Thm:FLL}}

\medskip

Denote $\hat{k}_{ni}=k\left( \left. \left( 1-\hat{F}_{n}\left( \hat{W}%
_{i}\right) \right) \right/ h_{n}\right) $ and $k_{ni}=k\left( \left. \left(
1-F_{W}\left( W_{i}\right) \right) \right/ h_{n}\right) $. And define the
infeasible local linear estimator as%
\begin{equation}
\left( \mu _{n}^{L},b_{n}^{L}\right) =\arg \min_{\mu ,b}\sum_{i=1}^{n}\left[
Y_{i}-Z_{i}^{\prime }\theta _{0}-\mu -\left( F_{W}\left( W_{i}\right)
-1\right) b\right] ^{2}D_{i}k_{ni}.  \label{InfeasibleLL}
\end{equation}%
I will first establish the asymptotic normality of $\left( \mu
_{n}^{L},b_{n}^{L}\right) $ by the standard argument for the nonparametric
local linear estimation, and then prove the asymptotic negligibility of the
difference of $\left( \hat{\mu}^{L},\hat{b}^{L}\right) $\ and $\left( \mu
_{n}^{L},b_{n}^{L}\right) $ in a way similar to the proof of Theorem \ref%
{Thm:FLC}.

\begin{itemize}
\item \textbf{First step: proving}%
\begin{equation*}
\left( 
\begin{array}{cc}
\sqrt{nh_{n}} &  \\ 
& \sqrt{nh_{n}^{3}}%
\end{array}%
\right) \left[ \left( 
\begin{array}{c}
\mu _{n}^{L} \\ 
b_{n}^{L}%
\end{array}%
\right) -\left( 
\begin{array}{c}
\mu _{0} \\ 
g\left( 1\right)%
\end{array}%
\right) -\left( 
\begin{array}{c}
\frac{\left( \kappa _{2}^{2}-\kappa _{1}\kappa _{3}\right) g^{\prime }\left(
1\right) }{2\left( \kappa _{0}\kappa _{2}-\kappa _{1}^{2}\right) }h_{n}^{2}
\\ 
\frac{\left( \kappa _{1}\kappa _{2}-\kappa _{0}\kappa _{3}\right) g^{\prime
}\left( 1\right) }{2\left( \kappa _{0}\kappa _{2}-\kappa _{1}^{2}\right) }%
h_{n}%
\end{array}%
\right) \right] \rightarrow N\left( 0,\sigma _{U}^{2}\Omega ^{L}\right) ,
\end{equation*}%
where the infeasible local linear estimator $\left( \mu
_{n}^{L},b_{n}^{L}\right) $ is defined in (\ref{InfeasibleLL}).
\end{itemize}

Denote $\Gamma =\left( \mathbf{1}_{n\times 1},F_{W}\left( \mathbf{W}\right) -%
\mathbf{1}_{n\times 1}\right) $ and $\mathbf{K}=diag\left(
D_{1}k_{n1},\cdots ,D_{n}k_{nn}\right) $, then%
\begin{eqnarray*}
\left( 
\begin{array}{c}
\mu _{n}^{L} \\ 
b_{n}^{L}%
\end{array}%
\right) &=&\arg \min_{\gamma \in R^{2}}\left( \mu _{0}\mathbf{1}_{n\times 1}+%
\mathbf{U}-\Gamma \gamma \right) ^{\prime }\mathbf{K}\left( \mu _{0}\mathbf{1%
}_{n\times 1}+\mathbf{U}-\Gamma \gamma \right) \\
&=&\left( \Gamma ^{\prime }\mathbf{K}\Gamma \right) ^{-1}\Gamma ^{\prime }%
\mathbf{K}\left( \mu _{0}\mathbf{1}_{n\times 1}+\mathbf{U}\right) .
\end{eqnarray*}%
Denote $V_{i}=U_{i}-G\left( F_{W}\left( W_{i}\right) \right) $, where $%
G\left( t\right) $ is defined in (\ref{G(t)}), then we have%
\begin{equation}
E\left[ V_{i}D_{i}\varphi \left( W_{i}\right) \right] =E\left[ \left( U_{i}-E%
\left[ U_{i}\left| D_{i}=1,W_{i}\right. \right] \right) D_{i}\varphi \left(
W_{i}\right) \right] =0  \label{2222}
\end{equation}%
for any function $\varphi \left( \cdot \right) $. Since $G\left( 1\right) =0$%
, a Taylor expansion of $G\left( F_{W}\left( W_{i}\right) \right) $ about $%
F_{W}\left( W_{i}\right) =1$ shows that%
\begin{equation*}
G\left( F_{W}\left( W_{i}\right) \right) =g\left( 1\right) \left(
F_{W}\left( W_{i}\right) -1\right) +\frac{1}{2}g^{\prime }\left( 1+\delta
_{i}\left( F_{W}\left( W_{i}\right) -1\right) \right) \left( F_{W}\left(
W_{i}\right) -1\right) ^{2}
\end{equation*}%
for some $\delta _{i}\in \left[ 0,1\right] $. Hence,%
\begin{equation}
\left( 
\begin{array}{c}
\mu _{n}^{L} \\ 
b_{n}^{L}%
\end{array}%
\right) =\left( 
\begin{array}{c}
\mu _{0} \\ 
g\left( 1\right)%
\end{array}%
\right) +\left( \Gamma ^{\prime }\mathbf{K}\Gamma \right) ^{-1}\Gamma
^{\prime }\mathbf{K}\left( \mathbf{V}+\mathbf{M}+\mathbf{T}\right) ,
\label{00}
\end{equation}%
where%
\begin{eqnarray*}
M_{i} &=&\frac{1}{2}g^{\prime }\left( 1\right) \left( 1-F_{W}\left(
W_{i}\right) \right) ^{2}, \\
T_{i} &=&\frac{1}{2}\left[ g^{\prime }\left( 1-\delta _{i}\left(
1-F_{W}\left( W_{i}\right) \right) \right) -g^{\prime }\left( 1\right) %
\right] \left( 1-F_{W}\left( W_{i}\right) \right) ^{2}.
\end{eqnarray*}

\textbf{Firstly}, consider the denominator term $\Gamma ^{\prime }\mathbf{K}%
\Gamma $ and the bias term $\Gamma ^{\prime }\mathbf{KM}$. It follows from
Lemma \ref{Lemma:Lambda} that%
\begin{eqnarray*}
\left( \Gamma ^{\prime }\mathbf{K}\Gamma \right) ^{-1} &=&\left( 
\begin{array}{cc}
\Lambda _{0}\left( 1\right) & -\Lambda _{1}\left( 1\right) \\ 
-\Lambda _{1}\left( 1\right) & \Lambda _{2}\left( 1\right)%
\end{array}%
\right) ^{-1}=\frac{1}{\Lambda _{0}\left( 1\right) \Lambda _{2}\left(
1\right) -\Lambda _{1}^{2}\left( 1\right) }\left( 
\begin{array}{cc}
\Lambda _{2}\left( 1\right) & \Lambda _{1}\left( 1\right) \\ 
\Lambda _{1}\left( 1\right) & \Lambda _{0}\left( 1\right)%
\end{array}%
\right) \\
&=&\frac{1}{\kappa _{0}\kappa _{2}-\kappa _{1}^{2}}\left( 
\begin{array}{cc}
\left( nh_{n}\right) ^{-1}\left( \kappa _{2}+o_{p}\left( 1\right) \right) & 
\left( nh_{n}^{2}\right) ^{-1}\left( \kappa _{1}+o_{p}\left( 1\right) \right)
\\ 
\left( nh_{n}^{2}\right) ^{-1}\left( \kappa _{1}+o_{p}\left( 1\right) \right)
& \left( nh_{n}^{3}\right) ^{-1}\left( \kappa _{0}+o_{p}\left( 1\right)
\right)%
\end{array}%
\right) \\
&=&\left( 
\begin{array}{cc}
O_{p}\left( \left( nh_{n}\right) ^{-1}\right) & O_{p}\left( \left(
nh_{n}^{2}\right) ^{-1}\right) \\ 
O_{p}\left( \left( nh_{n}^{2}\right) ^{-1}\right) & O_{p}\left( \left(
nh_{n}^{3}\right) ^{-1}\right)%
\end{array}%
\right) ,
\end{eqnarray*}%
and that%
\begin{equation*}
\Gamma ^{\prime }\mathbf{KM}=\frac{g^{\prime }\left( 1\right) }{2}\left( 
\begin{array}{c}
\Lambda _{2}\left( 1\right) \\ 
-\Lambda _{3}\left( 1\right)%
\end{array}%
\right) =\frac{g^{\prime }\left( 1\right) }{2}\left( 
\begin{array}{c}
\left( nh_{n}^{3}\right) \left( \kappa _{2}+o_{p}\left( 1\right) \right) \\ 
-\left( nh_{n}^{4}\right) \left( \kappa _{3}+o_{p}\left( 1\right) \right)%
\end{array}%
\right) ,
\end{equation*}%
where $\Lambda _{r}\left( 1\right) $ is defined in (\ref{Lambdar}). So%
\begin{equation}
\left( \Gamma ^{\prime }\mathbf{K}\Gamma \right) ^{-1}\Gamma ^{\prime }%
\mathbf{KM}=\frac{g^{\prime }\left( 1\right) }{2\left( \kappa _{0}\kappa
_{2}-\kappa _{1}^{2}\right) }\left( 
\begin{array}{c}
h_{n}^{2}\left( \kappa _{2}^{2}-\kappa _{1}\kappa _{3}\right) \\ 
h_{n}\left( \kappa _{1}\kappa _{2}-\kappa _{0}\kappa _{3}\right)%
\end{array}%
\right) +\left( 
\begin{array}{c}
o_{p}\left( h_{n}^{2}\right) \\ 
o_{p}\left( h_{n}\right)%
\end{array}%
\right) .  \label{11}
\end{equation}

\textbf{Secondly}, consider the remainder term%
\begin{equation*}
\Gamma ^{\prime }\mathbf{KT}=\frac{1}{2}\left( 
\begin{array}{c}
\sum_{i}\left[ g^{\prime }\left( 1-\delta _{i}\left( 1-F_{W}\left(
W_{i}\right) \right) \right) -g^{\prime }\left( 1\right) \right] \left(
1-F_{W}\left( W_{i}\right) \right) ^{2}D_{i}k_{ni} \\ 
-\sum_{i}\left[ g^{\prime }\left( 1-\delta _{i}\left( 1-F_{W}\left(
W_{i}\right) \right) \right) -g^{\prime }\left( 1\right) \right] \left(
1-F_{W}\left( W_{i}\right) \right) ^{3}D_{i}k_{ni}%
\end{array}%
\right) .
\end{equation*}%
Since $k_{ni}=k_{ni}1\left\{ 1-F_{W}\left( W_{i}\right) <h_{n}\right\} $, we
have%
\begin{eqnarray*}
\left\vert \left( \Gamma ^{\prime }\mathbf{KT}\right) _{j}\right\vert &\leq &%
\frac{h_{n}^{j+1}}{2}\sup_{\delta \in \left[ 0,h_{n}\right] }\left\vert
g^{\prime }\left( 1-\delta \right) -g^{\prime }\left( 1\right) \right\vert
\sum_{i}D_{i}k_{ni} \\
&=&\frac{h_{n}^{j+1}}{2}\cdot o\left( 1\right) \cdot O_{p}\left(
nh_{n}\right) \\
&=&o_{p}\left( nh_{n}^{j+2}\right) .
\end{eqnarray*}%
Thus,%
\begin{equation}
\left( \Gamma ^{\prime }\mathbf{K}\Gamma \right) ^{-1}\Gamma ^{\prime }%
\mathbf{KT}=\left( 
\begin{array}{cc}
O_{p}\left( \left( nh_{n}\right) ^{-1}\right) & O_{p}\left( \left(
nh_{n}^{2}\right) ^{-1}\right) \\ 
O_{p}\left( \left( nh_{n}^{2}\right) ^{-1}\right) & O_{p}\left( \left(
nh_{n}^{3}\right) ^{-1}\right)%
\end{array}%
\right) \left( 
\begin{array}{c}
o_{p}\left( nh_{n}^{3}\right) \\ 
o_{p}\left( nh_{n}^{4}\right)%
\end{array}%
\right) =\left( 
\begin{array}{c}
o_{p}\left( h_{n}^{2}\right) \\ 
o_{p}\left( h_{n}\right)%
\end{array}%
\right) .  \label{22}
\end{equation}

\textbf{Thirdly}, consider the variance term $\Gamma ^{\prime }\mathbf{KV}%
=\left( \Xi _{0},-\Xi _{1}\right) ^{\prime }$, where%
\begin{equation*}
\Xi _{r}=\sum_{i}V_{i}\left( 1-F_{W}\left( W_{i}\right) \right)
^{r}D_{i}k_{ni}.
\end{equation*}%
By (\ref{2222}), we have $E\left[ \Xi _{r}\right] =0$ and%
\begin{eqnarray*}
Var\left( \Xi _{r}\right) &=&nE\left[ V_{i}^{2}\left( 1-F_{W}\left(
W_{i}\right) \right) ^{2r}D_{i}k_{ni}^{2}\right] \\
&=&nE\left[ U_{i}^{2}\left( 1-F_{W}\left( W_{i}\right) \right)
^{2r}D_{i}k_{ni}^{2}\right] \\
&&-nE\left[ G^{2}\left( F_{W}\left( W_{i}\right) \right) G_{1}\left(
F_{W}\left( W_{i}\right) \right) \left( 1-F_{W}\left( W_{i}\right) \right)
^{2r}k_{ni}^{2}\right] ,
\end{eqnarray*}%
where $G_{1}\left( t\right) =\Pr \left( F_{W}\left( \varepsilon _{i}\right)
<t\right) $ with $G_{1}\left( 1\right) =1$. As in Lemma B.1, we can show that%
\begin{equation*}
E\left[ U_{i}^{2}\left( 1-F_{W}\left( W_{i}\right) \right)
^{2r}D_{i}k_{ni}^{2}\right] =E\left[ U_{i}^{2}\left( 1-F_{W}\left(
W_{i}\right) \right) ^{2r}k_{ni}^{2}\right] \left( 1+o\left( 1\right)
\right) =h_{n}^{2r+1}\sigma _{U}^{2}\chi _{2r}\left( 1+o\left( 1\right)
\right) ,
\end{equation*}%
and%
\begin{eqnarray*}
&&E\left[ G^{2}\left( F_{W}\left( W_{i}\right) \right) G_{1}\left(
F_{W}\left( W_{i}\right) \right) \left( 1-F_{W}\left( W_{i}\right) \right)
^{2r}k_{ni}^{2}\right] \\
&=&h_{n}^{2r+1}\int_{0}^{1}G^{2}\left( 1-h_{n}t\right) G_{1}\left(
1-h_{n}t\right) t^{2r}k^{2}\left( t\right) dt=o\left( h_{n}^{2r+1}\right) .
\end{eqnarray*}%
Hence, $Var\left( \Xi _{r}\right) =nh_{n}^{2r+1}\sigma _{U}^{2}\left( \chi
_{2r}+o\left( 1\right) \right) $. Similarly, we have%
\begin{equation*}
Cov\left( \Xi _{0},\Xi _{1}\right) =E\left[ \Xi _{0}\Xi _{1}\right] =nE\left[
V_{i}^{2}\left( 1-F_{W}\left( W_{i}\right) \right) D_{i}k_{ni}^{2}\right]
=nh_{n}^{2}\sigma _{U}^{2}\left( \chi _{1}+o\left( 1\right) \right) .
\end{equation*}%
It follows from Lindeberg's central limit theorem and the Cram\'{e}r-Wold
device \citep[e.g.,][Theorem 1.9(iii)]{shao2003} that%
\begin{equation*}
\left( 
\begin{array}{c}
\left( nh_{n}\right) ^{-1/2}\Xi _{0} \\ 
\left( nh_{n}^{3}\right) ^{-1/2}\Xi _{1}%
\end{array}%
\right) \rightarrow N\left( 0,\sigma _{U}^{2}\left( 
\begin{array}{cc}
\chi _{0} & \chi _{1} \\ 
\chi _{1} & \chi _{2}%
\end{array}%
\right) \right) ,
\end{equation*}%
because the Lindeberg's condition holds under $nh_{n}\rightarrow \infty $,
which can be shown as in Lemma \ref{Lemma:Lindeberg}.

\textbf{Fourthly}, consider%
\begin{eqnarray*}
\left( \Gamma ^{\prime }\mathbf{K}\Gamma \right) ^{-1}\Gamma ^{\prime }%
\mathbf{KV} &=&\frac{1}{\kappa _{0}\kappa _{2}-\kappa _{1}^{2}}\left( 
\begin{array}{cc}
\left( nh_{n}\right) ^{-1}\left( \kappa _{2}+o_{p}\left( 1\right) \right) & 
\left( nh_{n}^{2}\right) ^{-1}\left( \kappa _{1}+o_{p}\left( 1\right) \right)
\\ 
\left( nh_{n}^{2}\right) ^{-1}\left( \kappa _{1}+o_{p}\left( 1\right) \right)
& \left( nh_{n}^{3}\right) ^{-1}\left( \kappa _{0}+o_{p}\left( 1\right)
\right)%
\end{array}%
\right) \left( 
\begin{array}{c}
\Xi _{0} \\ 
-\Xi _{1}%
\end{array}%
\right) \\
&=&\frac{1}{\kappa _{0}\kappa _{2}-\kappa _{1}^{2}}\left( 
\begin{array}{c}
\kappa _{2}\left( nh_{n}\right) ^{-1}\Xi _{0}-\kappa _{1}\left(
nh_{n}^{2}\right) ^{-1}\Xi _{1} \\ 
\kappa _{1}\left( nh_{n}^{2}\right) ^{-1}\Xi _{0}-\kappa _{0}\left(
nh_{n}^{3}\right) ^{-1}\Xi _{1}%
\end{array}%
\right) +\left( 
\begin{array}{c}
o_{p}\left( \left( nh_{n}\right) ^{-1/2}\right) \\ 
o_{p}\left( \left( nh_{n}^{3}\right) ^{-1/2}\right)%
\end{array}%
\right) .
\end{eqnarray*}%
It follows from Lindeberg's central limit theorem and the Cram\'{e}r-Wold
device that%
\begin{eqnarray}
&&\left( 
\begin{array}{cc}
\sqrt{nh_{n}} &  \\ 
& \sqrt{nh_{n}^{3}}%
\end{array}%
\right) \left( \Gamma ^{\prime }\mathbf{K}\Gamma \right) ^{-1}\Gamma
^{\prime }\mathbf{KV}  \notag \\
&\mathbf{=}&\frac{1}{\kappa _{0}\kappa _{2}-\kappa _{1}^{2}}\left( 
\begin{array}{c}
\kappa _{2}\left( nh_{n}\right) ^{-1/2}\Xi _{0}-\kappa _{1}\left(
nh_{n}^{3}\right) ^{-1/2}\Xi _{1} \\ 
\kappa _{1}\left( nh_{n}\right) ^{-1/2}\Xi _{0}-\kappa _{0}\left(
nh_{n}^{3}\right) ^{-1/2}\Xi _{1}%
\end{array}%
\right) +\left( 
\begin{array}{c}
o_{p}\left( 1\right) \\ 
o_{p}\left( 1\right)%
\end{array}%
\right)  \notag \\
&\rightarrow &N\left( 0,\sigma _{U}^{2}\Omega ^{L}\right) .  \label{33}
\end{eqnarray}

In summary, the asymptotic normality of $\left( \mu
_{n}^{L},b_{n}^{L}\right) $ follows from replacing (\ref{00}) with (\ref{11}%
), (\ref{22}), (\ref{33}), and then normalizing it with $diag\left( \sqrt{%
nh_{n}},\sqrt{nh_{n}^{3}}\right) $.

\begin{itemize}
\item \textbf{Second step: proving}%
\begin{equation*}
\left( 
\begin{array}{cc}
\sqrt{nh_{n}} &  \\ 
& \sqrt{nh_{n}^{3}}%
\end{array}%
\right) \left[ \left( 
\begin{array}{c}
\hat{\mu}^{L} \\ 
\hat{b}^{L}%
\end{array}%
\right) -\left( 
\begin{array}{c}
\mu _{n}^{L} \\ 
b_{n}^{L}%
\end{array}%
\right) \right] \overset{p}{\rightarrow }\left( 
\begin{array}{c}
0 \\ 
0%
\end{array}%
\right) .
\end{equation*}
\end{itemize}

Denote $\hat{\Gamma}=\left( \mathbf{1}_{n\times 1},\hat{F}_{n}\left( \mathbf{%
\hat{W}}\right) -\mathbf{1}_{n\times 1}\right) $, $\mathbf{\hat{K}}%
=diag\left( D_{1}\hat{k}_{n1},\cdots ,D_{n}\hat{k}_{nn}\right) $, and $\hat{%
\Lambda}_{r}\left( \pi \right) =\sum_{i=1}^{n}\left( 1-\hat{F}_{n}\left( 
\hat{W}_{i}\right) \right) ^{r}\pi _{i}D_{i}\hat{k}_{ni}$ with $\pi $ being
either $U$ or $1$, then 
\begin{equation*}
\left( 
\begin{array}{c}
\hat{\mu}^{L} \\ 
\hat{b}^{L}%
\end{array}%
\right) =\left( \hat{\Gamma}^{\prime }\mathbf{\hat{K}}\hat{\Gamma}\right)
^{-1}\hat{\Gamma}^{\prime }\mathbf{\hat{K}}\left( \mathbf{Y}-\mathbf{Z}\hat{%
\theta}\right) =\frac{1}{\hat{\Lambda}_{0}\left( 1\right) \hat{\Lambda}%
_{2}\left( 1\right) -\hat{\Lambda}_{1}^{2}\left( 1\right) }\left( 
\begin{array}{cc}
\hat{\Lambda}_{2}\left( 1\right) & \hat{\Lambda}_{1}\left( 1\right) \\ 
\hat{\Lambda}_{1}\left( 1\right) & \hat{\Lambda}_{0}\left( 1\right)%
\end{array}%
\right) \hat{\Gamma}^{\prime }\mathbf{\hat{K}}\left( \mathbf{Y}-\mathbf{Z}%
\hat{\theta}\right) .
\end{equation*}%
Further denote%
\begin{eqnarray*}
\hat{A}_{rn} &=&\sum_{i}\left( Y_{i}-Z_{i}^{\prime }\hat{\theta}\right) %
\left[ \hat{\Lambda}_{r+1}\left( 1\right) -\left( 1-\hat{F}_{n}\left( \hat{W}%
_{i}\right) \right) \hat{\Lambda}_{r}\left( 1\right) \right] D_{i}\hat{k}%
_{ni}, \\
A_{rn} &=&\sum_{i}\left( Y_{i}-Z_{i}^{\prime }\theta _{0}\right) \left[
\Lambda _{r+1}\left( 1\right) -\left( 1-F_{W}\left( W_{i}\right) \right)
\Lambda _{r}\left( 1\right) \right] D_{i}k_{ni}, \\
\hat{B}_{n} &=&\hat{\Lambda}_{0}\left( 1\right) \hat{\Lambda}_{2}\left(
1\right) -\hat{\Lambda}_{1}^{2}\left( 1\right) , \\
B_{n} &=&\Lambda _{0}\left( 1\right) \Lambda _{2}\left( 1\right) -\Lambda
_{1}^{2}\left( 1\right) ,
\end{eqnarray*}%
then%
\begin{equation*}
\left( 
\begin{array}{c}
\hat{\mu}^{L} \\ 
\hat{b}^{L}%
\end{array}%
\right) =\frac{1}{\hat{B}_{n}}\left( 
\begin{array}{c}
\hat{A}_{1n} \\ 
\hat{A}_{0n}%
\end{array}%
\right) ,\text{ \ }\left( 
\begin{array}{c}
\mu _{n}^{L} \\ 
b_{n}^{L}%
\end{array}%
\right) =\frac{1}{B_{n}}\left( 
\begin{array}{c}
A_{1n} \\ 
A_{0n}%
\end{array}%
\right) .
\end{equation*}%
It follows from the first part of the proof that%
\begin{eqnarray*}
\frac{A_{1n}}{B_{n}} &=&\mu _{0}+O_{p}\left( h_{n}^{2}+\frac{1}{\sqrt{nh_{n}}%
}\right) =\mu _{0}+o_{p}\left( h_{n}\right) , \\
\frac{A_{0n}}{B_{n}} &=&g\left( 1\right) +O_{p}\left( h_{n}+\frac{1}{\sqrt{%
nh_{n}^{3}}}\right) =O_{p}\left( 1\right) ,
\end{eqnarray*}%
and from Lemma \ref{Lemma:Lambda} that%
\begin{equation*}
B_{n}=n^{2}h_{n}^{4}\left( \kappa _{0}\kappa _{2}-\kappa
_{1}^{2}+o_{p}\left( 1\right) \right) .
\end{equation*}%
Also note that%
\begin{eqnarray*}
\frac{\hat{A}_{1n}}{\hat{B}_{n}}-\frac{A_{1n}}{B_{n}} &=&\left[ \left( \hat{A%
}_{1n}-A_{1n}\right) -\left( \hat{B}_{n}-B_{n}\right) \frac{A_{1n}}{B_{n}}%
\right] \frac{1}{\hat{B}_{n}} \\
&=&\left[ \left( \hat{A}_{1n}-A_{1n}\right) -\left( \hat{B}_{n}-B_{n}\right)
\mu _{0}+o_{p}\left( h_{n}\left( \hat{B}_{n}-B_{n}\right) \right) \right] 
\frac{B_{n}}{\hat{B}_{n}}\frac{1}{B_{n}}, \\
\frac{\hat{A}_{0n}}{\hat{B}_{n}}-\frac{A_{0n}}{B_{n}} &=&\left[ \left( \hat{A%
}_{0n}-A_{0n}\right) -\left( \hat{B}_{n}-B_{n}\right) \frac{A_{0n}}{B_{n}}%
\right] \frac{1}{\hat{B}_{n}} \\
&=&\left[ \left( \hat{A}_{0n}-A_{0n}\right) -O_{p}\left( \hat{B}%
_{n}-B_{n}\right) \right] \frac{B_{n}}{\hat{B}_{n}}\frac{1}{B_{n}}.
\end{eqnarray*}%
Therefore, in order to prove%
\begin{equation*}
\left( 
\begin{array}{cc}
\sqrt{nh_{n}} &  \\ 
& \sqrt{nh_{n}^{3}}%
\end{array}%
\right) \left( 
\begin{array}{c}
\hat{\mu}^{L}-\mu _{n}^{L} \\ 
\hat{b}^{L}-b_{n}^{L}%
\end{array}%
\right) =\left( 
\begin{array}{c}
\sqrt{nh_{n}}\left( \left. \hat{A}_{1n}\right/ \hat{B}_{n}-\left.
A_{1n}\right/ B_{n}\right) \\ 
\sqrt{nh_{n}^{3}}\left( \left. \hat{A}_{0n}\right/ \hat{B}_{n}-\left.
A_{0n}\right/ B_{n}\right)%
\end{array}%
\right) \overset{p}{\rightarrow }\left( 
\begin{array}{c}
0 \\ 
0%
\end{array}%
\right) ,
\end{equation*}%
it is sufficient to prove%
\begin{equation}
\text{(a) }\frac{\hat{B}_{n}}{B_{n}}\overset{p}{\rightarrow }1\text{, (b) }%
\frac{\left( \hat{A}_{1n}-A_{1n}\right) -\left( \hat{B}_{n}-B_{n}\right) \mu
_{0}}{\sqrt{n^{3}h_{n}^{7}}}\overset{p}{\rightarrow }0\text{, (c) }\frac{%
\hat{A}_{0n}-A_{0n}}{\sqrt{n^{3}h_{n}^{5}}}\overset{p}{\rightarrow }0\text{,
(d) }\frac{\hat{B}_{n}-B_{n}}{\sqrt{n^{3}h_{n}^{5}}}\overset{p}{\rightarrow }%
0\text{.}  \label{abcdL}
\end{equation}%
Note that (d) implies $\left. \sqrt{nh_{n}^{3}}\left( \hat{B}%
_{n}-B_{n}\right) \right/ B_{n}\overset{p}{\rightarrow }0$, which in turn
implies (a). Some calculations show that%
\begin{eqnarray}
\left( \hat{A}_{1n}\!-\!A_{1n}\right) \!-\!\left( \hat{B}_{n}\!-\!B_{n}%
\right) \mu _{0} &=&\left[ \hat{\Lambda}_{0}\left( U\right) \hat{\Lambda}%
_{2}\left( 1\right) -\Lambda _{0}\left( U\right) \Lambda _{2}\left( 1\right) %
\right] -\left[ \hat{\Lambda}_{1}\left( U\right) \hat{\Lambda}_{1}\left(
1\right) -\Lambda _{1}\left( U\right) \Lambda _{1}\left( 1\right) \right] 
\notag \\
&-&\left( \hat{\theta}-\theta _{0}\right) ^{\prime }\left[ \hat{\Lambda}%
_{2}\left( 1\right) \sum_{i}Z_{i}D_{i}\hat{k}_{ni}-\hat{\Lambda}_{1}\left(
1\right) \sum_{i}Z_{i}\left( 1-\hat{F}_{n}\left( \hat{W}_{i}\right) \right)
D_{i}\hat{k}_{ni}\right]  \notag \\
\hat{A}_{0n}-A_{0n} &=&\left[ \hat{\Lambda}_{0}\left( U\right) \hat{\Lambda}%
_{1}\left( 1\right) -\Lambda _{0}\left( U\right) \Lambda _{1}\left( 1\right) %
\right] -\left[ \hat{\Lambda}_{1}\left( U\right) \hat{\Lambda}_{0}\left(
1\right) -\Lambda _{1}\left( U\right) \Lambda _{0}\left( 1\right) \right] 
\notag \\
&-&\left( \hat{\theta}-\theta _{0}\right) ^{\prime }\left[ \hat{\Lambda}%
_{1}\left( 1\right) \sum_{i}Z_{i}D_{i}\hat{k}_{ni}-\hat{\Lambda}_{0}\left(
1\right) \sum_{i}Z_{i}\left( 1-\hat{F}_{n}\left( \hat{W}_{i}\right) \right)
D_{i}\hat{k}_{ni}\right]  \notag \\
\hat{B}_{n}-B_{n} &=&\left[ \hat{\Lambda}_{0}\left( 1\right) \hat{\Lambda}%
_{2}\left( 1\right) -\Lambda _{0}\left( 1\right) \Lambda _{2}\left( 1\right) %
\right] -\left[ \hat{\Lambda}_{1}^{2}\left( 1\right) -\Lambda _{1}^{2}\left(
1\right) \right] .  \label{diff3}
\end{eqnarray}%
I will analyze the asymptotic behaviors of the terms appeared in the above
equations in order.

\textbf{Firstly}, consider $\hat{\Lambda}_{0}\left( \pi \right) -\Lambda
_{0}\left( \pi \right) =\sum_{i}\pi _{i}D_{i}\left( \hat{k}%
_{ni}-k_{ni}\right) $. It follows from the proof of Theorem \ref{Thm:FLC}
that%
\begin{equation}
\hat{\Lambda}_{0}\left( \pi \right) -\Lambda _{0}\left( \pi \right) =\left\{ 
\begin{array}{cc}
o_{p}\left( \sqrt{nh_{n}}\right) & \text{for }\pi \text{ being }U\text{
(therein the term (b1)),} \\ 
O_{p}\left( \sqrt{nh_{n}}\right) & \text{for }\pi \text{ being }1\text{
(therein the term (c1)).}%
\end{array}%
\right.  \label{Lambda0_diff}
\end{equation}

\textbf{Secondly}, consider%
\begin{equation*}
\hat{\Lambda}_{1}\left( \pi \right) -\Lambda _{1}\left( \pi \right)
=\sum_{i}\pi _{i}\left[ \left( 1-\hat{F}_{n}\left( \hat{W}_{i}\right)
\right) D_{i}\hat{k}_{ni}-\left( 1-F_{W}\left( W_{i}\right) \right)
D_{i}k_{ni}\right] .
\end{equation*}%
By (\ref{survival_decom}), we can decompose it into four terms as%
\begin{equation}
\hat{\Lambda}_{1}\left( \pi \right) -\Lambda _{1}\left( \pi \right) =\left( 
\text{I}\right) +\left( \text{II}\right) +\left( \text{III}\right) +\left( 
\text{IV}\right) ,  \label{Lambda1_decom}
\end{equation}%
where%
\begin{eqnarray*}
\left( \text{I}\right) &=&\frac{1}{n-1}\sum_{i}\sum_{j\neq i}\left[ 1\left\{
W_{j}>W_{i}\right\} -\left( 1-F_{W}\left( W_{i}\right) \right) \right] \pi
_{i}D_{i}k_{ni}, \\
\left( \text{II}\right) &=&\frac{1}{n-1}\sum_{i}\sum_{j\neq i}\alpha
_{ij}\left( \hat{\beta}\right) \pi _{i}D_{i}k_{ni}, \\
\left( \text{III}\right) &=&\sum_{i}\left( 1-F_{W}\left( W_{i}\right)
\right) \pi _{i}D_{i}\left( \hat{k}_{ni}-k_{ni}\right) , \\
\left( \text{IV}\right) &=&\sum_{i}\left[ \left( 1-\hat{F}_{n}\left( \hat{W}%
_{i}\right) \right) -\left( 1-F_{W}\left( W_{i}\right) \right) \right] \pi
_{i}D_{i}\left( \hat{k}_{ni}-k_{ni}\right) .
\end{eqnarray*}

Denote $m_{1n}\left( \xi _{i},\xi _{j}\right) =h_{n}^{-3/2}\left[ 1\left\{
W_{j}>W_{i}\right\} -\left( 1-F_{W}\left( W_{i}\right) \right) \right] \pi
_{i}D_{i}k_{ni}$, then $E\left[ \left. m_{1n}\left( \xi _{i},\xi _{j}\right)
\right| \xi _{i}\right] =0$, and%
\begin{equation*}
\left( \text{I}\right) =\sqrt{nh_{n}^{3}}\frac{\sqrt{n}}{n\left( n-1\right) }%
\sum_{i}\sum_{j\neq i}m_{1n}\left( \xi _{i},\xi _{j}\right) .
\end{equation*}%
Since $E\left[ m_{1n}^{2}\left( \xi _{i},\xi _{j}\right) \right] =o\left(
n\right) $ by Assumption 5'', it follows from the projection method for
U-statistics \citep[e.g.,][Lemma 3.1]{powell1989semiparametric} that%
\begin{equation*}
\left( \text{I}\right) =h_{n}^{3/2}\sum_{j=1}^{n}E\left[ \left. m_{1n}\left(
\xi _{i},\xi _{j}\right) \right| \xi _{j}\right] +o_{p}\left( \sqrt{%
nh_{n}^{3}}\right) .
\end{equation*}%
Denote $G_{\pi }\left( t\right) =E\left[ \pi _{i}1\left\{ F_{W}\left(
\varepsilon _{i}\right) <t\right\} \right] $, then%
\begin{eqnarray*}
E\left[ \left. m_{1n}\left( \xi _{i},\xi _{j}\right) \right| \xi _{j}\right]
&=&h_{n}^{-3/2}E\left[ \left. \left[ 1\left\{ W_{j}>W_{i}\right\} -\left(
1-F_{W}\left( W_{i}\right) \right) \right] G_{\pi }\left( F_{W}\left(
W_{i}\right) \right) k_{ni}\right| W_{j}\right] \\
&=&h_{n}^{-1/2}\int_{0}^{1}\left[ 1\left\{ F_{W}\left( W_{j}\right)
>1-h_{n}t\right\} -h_{n}t\right] G_{\pi }\left( 1-h_{n}t\right) k\left(
t\right) dt \\
&=&h_{n}^{-1/2}\int_{0}^{1}\left[ 1\left\{ 1-F_{W}\left( W_{j}\right)
<h_{n}t\right\} -h_{n}t\right] \left[ G_{\pi }\left( 1-h_{n}t\right) -G_{\pi
}\left( 1\right) \right] k\left( t\right) dt \\
&&+h_{n}^{-1/2}G_{\pi }\left( 1\right) \left( \int_{0}^{1}1\left\{
1-F_{W}\left( W_{j}\right) <h_{n}t\right\} k\left( t\right) dt-h_{n}\kappa
_{1}\right) .
\end{eqnarray*}%
It follows from Fubini's theorem \citep[e.g.,][Theorem 1.3]{shao2003} that%
\begin{eqnarray*}
E\left[ \int_{0}^{1}1\left\{ 1-F_{W}\left( W_{j}\right) <h_{n}t\right\}
k\left( t\right) dt\right] &=&\int_{0}^{1}h_{n}tk\left( t\right)
dt=h_{n}\kappa _{1}, \\
Var\left( \int_{0}^{1}1\left\{ 1-F_{W}\left( W_{j}\right) <h_{n}t\right\}
k\left( t\right) dt\right) &\leq &E\left[ \left( \int_{0}^{1}1\left\{
1-F_{W}\left( W_{j}\right) <h_{n}t\right\} k\left( t\right) dt\right) ^{2}%
\right] \\
&\leq &E\left[ \int_{0}^{1}1\left\{ 1-F_{W}\left( W_{j}\right)
<h_{n}t\right\} k^{2}\left( t\right) dt\right] \\
&=&\int_{0}^{1}h_{n}tk^{2}\left( t\right) dt=h_{n}\chi _{1}.
\end{eqnarray*}%
Therefore, we have%
\begin{eqnarray*}
&&\sum_{j=1}^{n}\left( \int_{0}^{1}1\left\{ 1-F_{W}\left( W_{j}\right)
<h_{n}t\right\} k\left( t\right) dt-h_{n}\kappa _{1}\right) \\
&=&O_{p}\left( \sqrt{nVar\left( \int_{0}^{1}1\left\{ 1-F_{W}\left(
W_{j}\right) <h_{n}t\right\} k\left( t\right) dt\right) }\right) \\
&=&O_{p}\left( \sqrt{nh_{n}}\right) ,
\end{eqnarray*}%
and thus%
\begin{equation*}
\left( \text{I}\right) =\sqrt{nh_{n}^{3}}R_{1n}+O_{p}\left( \sqrt{nh_{n}^{3}}%
\right) G_{\pi }\left( 1\right) +o_{p}\left( \sqrt{nh_{n}^{3}}\right) ,
\end{equation*}%
where%
\begin{equation*}
\left| R_{1n}\right| \leq \bar{k}\sup_{s\in \left[ 0,h_{n}\right] }\left| 
\frac{1}{\sqrt{nh_{n}}}\sum_{j=1}^{n}\left[ 1\left\{ 1-F_{W}\left(
W_{j}\right) <s\right\} -s\right] \right| \int_{0}^{1}\left| G_{\pi }\left(
1-h_{n}t\right) -G_{\pi }\left( 1\right) \right| dt.
\end{equation*}%
If follows from \citet[Lemma 2.3]{stute1982oscillation} and the dominated
convergence theorem that%
\begin{equation*}
\left| R_{1n}\right| =O_{p}\left( 1\right) \cdot o\left( 1\right)
=o_{p}\left( 1\right) .
\end{equation*}%
By noting that $G_{\pi }\left( 1\right) =E\left[ \pi _{i}\right] $, we have%
\begin{equation}
\left( \text{I}\right) =O_{p}\left( \sqrt{nh_{n}^{3}}\right) G_{\pi }\left(
1\right) +o_{p}\left( \sqrt{nh_{n}^{3}}\right) =\left\{ 
\begin{array}{cc}
o_{p}\left( \sqrt{nh_{n}^{3}}\right) & \text{for }\pi \text{ being }U\text{,}
\\ 
O_{p}\left( \sqrt{nh_{n}^{3}}\right) & \text{for }\pi \text{ being }1\text{.}%
\end{array}%
\right.  \label{Lambda1_I}
\end{equation}

For the second term $\left( \text{II}\right) $, it follows from the proof of
Theorem \ref{Thm:FLC} (also $\left( \text{II}\right) $ therein) that%
\begin{equation}
\left( \text{II}\right) =o_{p}\left( \sqrt{nh_{n}^{3}}\right) .
\label{Lambda1_II}
\end{equation}

For the third term $\left( \text{III}\right) $, a sixth-order Taylor
expansion yields that%
\begin{equation*}
\left( \text{III}\right) =\left( \text{III-1}\right) +\left( \text{III-2}%
\right) +\left( \text{III-3}\right) +\left( \text{III-4}\right) ,
\end{equation*}%
where%
\begin{eqnarray*}
\left( \text{III-1}\right) &=&\frac{1}{h_{n}\left( n-1\right) }%
\sum_{i}\sum_{j\neq i}\left[ 1\left\{ W_{j}>W_{i}\right\} -\left(
1-F_{W}\left( W_{i}\right) \right) \right] \left( 1-F_{W}\left( W_{i}\right)
\right) \pi _{i}D_{i}k^{\prime }\left( \frac{1-F_{W}\left( W_{i}\right) }{%
h_{n}}\right) , \\
\left( \text{III-2}\right) &=&\frac{1}{h_{n}\left( n-1\right) }%
\sum_{i}\sum_{j\neq i}\alpha _{ij}\left( \hat{\beta}\right) \left(
1-F_{W}\left( W_{i}\right) \right) \pi _{i}D_{i}k^{\prime }\left( \frac{%
1-F_{W}\left( W_{i}\right) }{h_{n}}\right) , \\
\left( \text{III-3}\right) &=&\sum_{r=2}^{5}\frac{1}{r!h_{n}^{r}}\sum_{i}%
\left[ \left( 1-\hat{F}_{n}\left( \hat{W}_{i}\right) \right) -\left(
1-F_{W}\left( W_{i}\right) \right) \right] ^{r}\left( 1-F_{W}\left(
W_{i}\right) \right) \pi _{i}D_{i}k^{\left( r\right) }\left( \frac{%
1-F_{W}\left( W_{i}\right) }{h_{n}}\right) , \\
\left( \text{III-4}\right) &=&\frac{1}{6!h_{n}^{6}}\sum_{i}\left[ \left( 1-%
\hat{F}_{n}\left( \hat{W}_{i}\right) \right) -\left( 1-F_{W}\left(
W_{i}\right) \right) \right] ^{6}\left( 1-F_{W}\left( W_{i}\right) \right)
\pi _{i}D_{i}k^{\left( 6\right) }\left( \frac{1-F_{i}^{\ast }}{h_{n}}\right)
.
\end{eqnarray*}%
Denote%
\begin{equation*}
m_{2n}\left( \xi _{i},\xi _{j}\right) =h_{n}^{-5/2}\left[ 1\left\{
W_{j}>W_{i}\right\} -\left( 1-F_{W}\left( W_{i}\right) \right) \right]
\left( 1-F_{W}\left( W_{i}\right) \right) \pi _{i}D_{i}k^{\prime }\left( 
\frac{1-F_{W}\left( W_{i}\right) }{h_{n}}\right) ,
\end{equation*}%
then we have $E\left[ \left. m_{2n}\left( \xi _{i},\xi _{j}\right) \right|
\xi _{i}\right] =0$ and $E\left[ m_{2n}^{2}\left( \xi _{i},\xi _{j}\right) %
\right] =o\left( n\right) $. It follows from the projection method that%
\begin{equation*}
\left( \text{III-1}\right) =\sqrt{nh_{n}^{3}}\frac{\sqrt{n}}{n\left(
n-1\right) }\sum_{i}\sum_{j\neq i}m_{2n}\left( \xi _{i},\xi _{j}\right)
=h_{n}^{3/2}\sum_{j=1}^{n}E\left[ \left. m_{2n}\left( \xi _{i},\xi
_{j}\right) \right| \xi _{j}\right] +o_{p}\left( \sqrt{nh_{n}^{3}}\right) .
\end{equation*}%
Similar arguments show that%
\begin{eqnarray*}
&&h_{n}^{3/2}\sum_{j=1}^{n}E\left[ \left. m_{2n}\left( \xi _{i},\xi
_{j}\right) \right| \xi _{j}\right] \\
&=&h_{n}\sum_{j=1}^{n}\int_{0}^{1}\left( 1\left\{ F_{W}\left( W_{j}\right)
>1-h_{n}t\right\} -h_{n}t\right) G_{\pi }\left( 1-h_{n}t\right) tk^{\prime
}\left( t\right) dt \\
&=&\sqrt{nh_{n}^{3}}\int_{0}^{1}\left[ \frac{1}{\sqrt{nh_{n}}}%
\sum_{j=1}^{n}\left( 1\left\{ 1-F_{W}\left( W_{j}\right) <h_{n}t\right\}
-h_{n}t\right) \right] \left[ G_{\pi }\left( 1-h_{n}t\right) -G_{\pi }\left(
1\right) \right] tk^{\prime }\left( t\right) dt \\
&&+h_{n}G_{\pi }\left( 1\right) \sum_{j=1}^{n}\left[ \int_{0}^{1}1\left\{
1-F_{W}\left( W_{j}\right) <h_{n}t\right\} tk^{\prime }\left( t\right)
dt+2h_{n}\kappa _{1}\right] \\
&=&\sqrt{nh_{n}^{3}}\cdot O_{p}\left( 1\right) \cdot o\left( 1\right)
+h_{n}G_{\pi }\left( 1\right) \cdot O_{p}\left( \sqrt{nh_{n}}\right) \\
&=&O_{p}\left( \sqrt{nh_{n}^{3}}\right) E\left[ \pi _{i}\right] +o_{p}\left( 
\sqrt{nh_{n}^{3}}\right) .
\end{eqnarray*}%
Consequently,%
\begin{equation*}
\left( \text{III-1}\right) =O_{p}\left( \sqrt{nh_{n}^{3}}\right) E\left[ \pi
_{i}\right] +o_{p}\left( \sqrt{nh_{n}^{3}}\right) =\left\{ 
\begin{array}{cc}
o_{p}\left( \sqrt{nh_{n}^{3}}\right) & \text{for }\pi \text{ being }U\text{,}
\\ 
O_{p}\left( \sqrt{nh_{n}^{3}}\right) & \text{for }\pi \text{ being }1\text{.}%
\end{array}%
\right.
\end{equation*}%
For $\left( \text{III-2}\right) $ and $\left( \text{III-3}\right) $, it
follows from the proof of Theorem \ref{Thm:FLC} (specifically $\left( \text{%
II}\right) $ and $\left( \text{III}\right) $ therein), by inserting $%
1-F_{W}\left( W_{i}\right) <h_{n}$, that%
\begin{equation*}
\left| \left( \text{III-2}\right) \right| +\left| \left( \text{III-3}\right)
\right| \leq O_{p}\left( \sqrt{nh_{n}^{3}}\right)
\sum_{r=1}^{5}\sum_{s=1}^{r}\sum_{t=1}^{s}\sum_{i=1}^{n}J_{rsti}\left( \hat{%
\beta}\right) +o_{p}\left( \sqrt{nh_{n}^{3}}\right) =o_{p}\left( \sqrt{%
nh_{n}^{3}}\right) ,
\end{equation*}%
where $J_{rsti}\left( \beta \right) $ is defined in (\ref{Jrsti}). For $%
\left( \text{III-4}\right) $, it follows from (\ref{11111}) that%
\begin{equation*}
\left| \left( \text{III-4}\right) \right| \leq O_{p}\left( \frac{1}{%
n^{3}h_{n}^{6}}\right) \sum_{i}\left| \pi _{i}\right| \left( \left\|
X_{i}\right\| +1\right) ^{6}=O_{p}\left( \frac{1}{n^{2}h_{n}^{6}}\right)
E\left| \pi _{i}\right| E\left[ \left( \left\| X_{i}\right\| +1\right) ^{6}%
\right] =o_{p}\left( 1\right) .
\end{equation*}%
Collecting these terms, we obtain%
\begin{equation}
\left( \text{III}\right) =\left\{ 
\begin{array}{cc}
o_{p}\left( \sqrt{nh_{n}^{3}}\right) & \text{for }\pi \text{ being }U\text{,}
\\ 
O_{p}\left( \sqrt{nh_{n}^{3}}\right) & \text{for }\pi \text{ being }1\text{.}%
\end{array}%
\right.  \label{Lambda1_III}
\end{equation}

For the last term $\left( \text{IV}\right) $ in (\ref{Lambda1_decom}), a
third-order Taylor expansion will be enough, which gives%
\begin{eqnarray*}
\left( \text{IV}\right) &=&\sum_{r=2}^{3}\frac{1}{\left( r-1\right)
!h_{n}^{r-1}}\sum_{i}\left[ \left( 1-\hat{F}_{n}\left( \hat{W}_{i}\right)
\right) -\left( 1-F_{W}\left( W_{i}\right) \right) \right] ^{r}\pi
_{i}D_{i}k^{\left( r-1\right) }\left( \frac{1-F_{W}\left( W_{i}\right) }{%
h_{n}}\right) \\
&&+\frac{1}{6h_{n}^{3}}\sum_{i}\left[ \left( 1-\hat{F}_{n}\left( \hat{W}%
_{i}\right) \right) -\left( 1-F_{W}\left( W_{i}\right) \right) \right]
^{4}\pi _{i}D_{i}k^{\left( 3\right) }\left( \frac{1-F_{W}\left( W_{i}\right) 
}{h_{n}}\right) .
\end{eqnarray*}%
Similarly, we have%
\begin{eqnarray}
\left| \left( \text{IV}\right) \right| &\leq &O_{p}\left( \sqrt{nh_{n}^{3}}%
\right)
\sum_{r=2}^{3}\sum_{s=1}^{r}\sum_{t=1}^{s}\sum_{i=1}^{n}J_{rsti}\left( \hat{%
\beta}\right) +o_{p}\left( \sqrt{nh_{n}^{3}}\right) +O_{p}\left( \frac{1}{%
n^{2}h_{n}^{3}}\right) \sum_{i}\left| \pi _{i}\right| \left( \left\|
X_{i}\right\| +1\right) ^{4}  \notag \\
&&=o_{p}\left( \sqrt{nh_{n}^{3}}\right) +O_{p}\left( \frac{1}{nh_{n}^{3}}%
\right) =o_{p}\left( \sqrt{nh_{n}^{3}}\right) .  \label{Lambda1_IV}
\end{eqnarray}

Substituting (\ref{Lambda1_I}), (\ref{Lambda1_II}), (\ref{Lambda1_III}), and
(\ref{Lambda1_IV}) into (\ref{Lambda1_decom}) yields%
\begin{equation}
\hat{\Lambda}_{1}\left( \pi \right) -\Lambda _{1}\left( \pi \right) =\left\{ 
\begin{array}{cc}
o_{p}\left( \sqrt{nh_{n}^{3}}\right) & \text{for }\pi \text{ being }U\text{,}
\\ 
O_{p}\left( \sqrt{nh_{n}^{3}}\right) & \text{for }\pi \text{ being }1\text{.}%
\end{array}%
\right.  \label{Lambda1_diff}
\end{equation}

\textbf{Thirdly}, consider%
\begin{equation*}
\hat{\Lambda}_{2}\left( 1\right) -\Lambda _{2}\left( 1\right) =\sum_{i}\left[
\left( 1-\hat{F}_{n}\left( \hat{W}_{i}\right) \right) ^{2}D_{i}\hat{k}%
_{ni}-\left( 1-F_{W}\left( W_{i}\right) \right) ^{2}D_{i}k_{ni}\right] .
\end{equation*}%
I also decompose it into several terms as%
\begin{eqnarray*}
\hat{\Lambda}_{2}\left( 1\right) -\Lambda _{2}\left( 1\right)  &=&\sum_{i}%
\left[ \left( 1-\hat{F}_{n}\left( \hat{W}_{i}\right) \right) -\left(
1-F_{W}\left( W_{i}\right) \right) \right] ^{2}D_{i}k_{ni} \\
&&+2\sum_{i}\left[ \left( 1-\hat{F}_{n}\left( \hat{W}_{i}\right) \right)
-\left( 1-F_{W}\left( W_{i}\right) \right) \right] \left( 1-F_{W}\left(
W_{i}\right) \right) D_{i}k_{ni} \\
&&+\sum_{i}\left[ \left( 1-\hat{F}_{n}\left( \hat{W}_{i}\right) \right)
-\left( 1-F_{W}\left( W_{i}\right) \right) \right] ^{2}D_{i}\left( \hat{k}%
_{ni}-k_{ni}\right)  \\
&&+2\sum_{i}\left[ \left( 1-\hat{F}_{n}\left( \hat{W}_{i}\right) \right)
-\left( 1-F_{W}\left( W_{i}\right) \right) \right] \left( 1-F_{W}\left(
W_{i}\right) \right) D_{i}\left( \hat{k}_{ni}-k_{ni}\right)  \\
&&+\sum_{i}\left( 1-F_{W}\left( W_{i}\right) \right) ^{2}D_{i}\left( \hat{k}%
_{ni}-k_{ni}\right)  \\
&\triangleq &\left( \text{V}\right) +2\left( \text{VI}\right) +\left( \text{%
VII}\right) +2\left( \text{VIII}\right) +\left( \text{IX}\right) .
\end{eqnarray*}%
Analogously, we can show that{\small 
\begin{eqnarray*}
\left\vert \left( \text{V}\right) \right\vert  &\leq &O_{p}\left( \sqrt{%
nh_{n}^{5}}\right) \sum_{s=1}^{2}\sum_{t=1}^{s}\sum_{i=1}^{n}J_{2sti}\left( 
\hat{\beta}\right) +o_{p}\left( \sqrt{nh_{n}^{5}}\right) =o_{p}\left( \sqrt{%
nh_{n}^{5}}\right) , \\
\left\vert \left( \text{VI}\right) \right\vert  &\leq &\left\vert \frac{1}{%
n-1}\sum_{i}\sum_{j\neq i}\left[ 1\left\{ W_{j}>W_{i}\right\} -\left(
1-F_{W}\left( W_{i}\right) \right) \right] D_{i}k_{ni}\right\vert  \\
&&+\frac{\bar{k}h_{n}}{n-1}\sum_{i}\sum_{j\neq i}\left\vert \alpha
_{ij}\left( \hat{\beta}\right) \right\vert 1\left\{ F_{W}\left( W_{i}\right)
>1-h_{n}\right\}  \\
&=&\left[ O_{p}\left( \sqrt{nh_{n}^{5}}\right) G_{1}\left( 1\right)
+o_{p}\left( \sqrt{nh_{n}^{5}}\right) \right] +\left[ O_{p}\left( \sqrt{%
nh_{n}^{5}}\right) \sum_{i}J_{111i}\left( \hat{\beta}\right) +o_{p}\left( 
\sqrt{nh_{n}^{5}}\right) \right]  \\
&=&O_{p}\left( \sqrt{nh_{n}^{5}}\right) , \\
\left\vert \left( \text{VII}\right) \right\vert  &\leq &\left\vert \frac{1}{%
h_{n}}\sum_{i}\left[ \left( 1-\hat{F}_{n}\left( \hat{W}_{i}\right) \right)
-\left( 1-F_{W}\left( W_{i}\right) \right) \right] ^{3}D_{i}k^{\prime
}\left( \frac{1-F_{W}\left( W_{i}\right) }{h_{n}}\right) \right\vert  \\
&&+\left\vert \frac{1}{2h_{n}^{2}}\sum_{i}\left[ \left( 1-\hat{F}_{n}\left( 
\hat{W}_{i}\right) \right) -\left( 1-F_{W}\left( W_{i}\right) \right) \right]
^{4}D_{i}k^{\prime \prime }\left( \frac{1-F_{i}^{\ast }}{h_{n}}\right)
\right\vert  \\
&\leq &\left[ O_{p}\left( \sqrt{nh_{n}^{5}}\right)
\sum_{s=1}^{3}\sum_{t=1}^{s}\sum_{i=1}^{n}J_{3sti}\left( \hat{\beta}\right)
+o_{p}\left( \sqrt{nh_{n}^{5}}\right) \right] +\left[ O_{p}\left( \frac{1}{%
n^{2}h_{n}^{2}}\right) \sum_{i}\left( \left\Vert X_{i}\right\Vert +1\right)
^{4}\right]  \\
&=&o_{p}\left( \sqrt{nh_{n}^{5}}\right) +O_{p}\left( \frac{1}{nh_{n}^{2}}%
\right) =o_{p}\left( \sqrt{nh_{n}^{5}}\right) , \\
\left\vert \left( \text{VIII}\right) \right\vert  &\leq &\left\vert
\sum_{r=1}^{4}\frac{1}{r!h_{n}^{r}}\sum_{i}\left[ \left( 1-\hat{F}_{n}\left( 
\hat{W}_{i}\right) \right) -\left( 1-F_{W}\left( W_{i}\right) \right) \right]
^{r+1}\left( 1-F_{W}\left( W_{i}\right) \right) D_{i}k^{\left( r\right)
}\left( \frac{1-F_{W}\left( W_{i}\right) }{h_{n}}\right) \right\vert  \\
&&+\left\vert \frac{1}{5!h_{n}^{5}}\sum_{i}\left[ \left( 1-\hat{F}_{n}\left( 
\hat{W}_{i}\right) \right) -\left( 1-F_{W}\left( W_{i}\right) \right) \right]
^{6}\left( 1-F_{W}\left( W_{i}\right) \right) D_{i}k^{\left( 5\right)
}\left( \frac{1-F_{i}^{\ast }}{h_{n}}\right) \right\vert  \\
&\leq &\left[ O_{p}\left( \sqrt{nh_{n}^{5}}\right)
\sum_{r=2}^{5}\sum_{s=1}^{r}\sum_{t=1}^{s}\sum_{i=1}^{n}J_{rsti}\left( \hat{%
\beta}\right) +o_{p}\left( \sqrt{nh_{n}^{5}}\right) \right] +\left[
O_{p}\left( \frac{1}{n^{3}h_{n}^{5}}\right) \sum_{i}\left( \left\Vert
X_{i}\right\Vert +1\right) ^{6}\right]  \\
&=&o_{p}\left( \sqrt{nh_{n}^{5}}\right) +O_{p}\left( \frac{1}{n^{2}h_{n}^{5}}%
\right) =o_{p}\left( \sqrt{nh_{n}^{5}}\right) , \\
\left\vert \left( \text{IX}\right) \right\vert  &\leq &\left\vert \frac{1}{%
h_{n}\left( n-1\right) }\sum_{i}\sum_{j\neq i}\left[ 1\left\{
W_{j}>W_{i}\right\} -\left( 1-F_{W}\left( W_{i}\right) \right) \right]
\left( 1-F_{W}\left( W_{i}\right) \right) ^{2}D_{i}k^{\prime }\left( \frac{%
1-F_{W}\left( W_{i}\right) }{h_{n}}\right) \right\vert  \\
&&+\frac{\overline{k^{\prime }}h_{n}}{n-1}\sum_{i}\sum_{j\neq i}\left\vert
\alpha _{ij}\left( \hat{\beta}\right) \right\vert 1\left\{ F_{W}\left(
W_{i}\right) >1-h_{n}\right\}  \\
&&+\left\vert \sum_{r=2}^{5}\frac{1}{r!h_{n}^{r}}\sum_{i}\left[ \left( 1-%
\hat{F}_{n}\left( \hat{W}_{i}\right) \right) -\left( 1-F_{W}\left(
W_{i}\right) \right) \right] ^{r}\left( 1-F_{W}\left( W_{i}\right) \right)
^{2}D_{i}k^{\left( r\right) }\left( \frac{1-F_{W}\left( W_{i}\right) }{h_{n}}%
\right) \right\vert  \\
&&+\left\vert \frac{1}{6!h_{n}^{6}}\sum_{i}\left[ \left( 1-\hat{F}_{n}\left( 
\hat{W}_{i}\right) \right) -\left( 1-F_{W}\left( W_{i}\right) \right) \right]
^{6}\left( 1-F_{W}\left( W_{i}\right) \right) ^{2}D_{i}k^{\left( 6\right)
}\left( \frac{1-F_{i}^{\ast }}{h_{n}}\right) \right\vert  \\
&=&O_{p}\left( \sqrt{nh_{n}^{5}}\right) G_{1}\left( 1\right) +O_{p}\left( 
\sqrt{nh_{n}^{5}}\right)
\sum_{r=1}^{5}\sum_{s=1}^{r}\sum_{t=1}^{s}\sum_{i=1}^{n}J_{rsti}\left( \hat{%
\beta}\right) +o_{p}\left( \sqrt{nh_{n}^{5}}\right) =O_{p}\left( \sqrt{%
nh_{n}^{5}}\right) .
\end{eqnarray*}%
} Therefore, we have%
\begin{equation}
\hat{\Lambda}_{2}\left( 1\right) -\Lambda _{2}\left( 1\right) =O_{p}\left( 
\sqrt{nh_{n}^{5}}\right) .  \label{Lambda2_diff}
\end{equation}

\textbf{Fourthly}, consider $\left( \hat{\theta}-\theta _{0}\right) ^{\prime
}\sum_{i}Z_{i}D_{i}\hat{k}_{ni}$\ and $\left( \hat{\theta}-\theta
_{0}\right) ^{\prime }\sum_{i}Z_{i}\left( 1-\hat{F}_{n}\left( \hat{W}%
_{i}\right) \right) D_{i}\hat{k}_{ni}$. It follows from the proof of Theorem
2 (specifically (b2) therein) that%
\begin{equation}
\left( \hat{\theta}-\theta _{0}\right) ^{\prime }\sum_{i}Z_{i}D_{i}\hat{k}%
_{ni}=o_{p}\left( \sqrt{nh_{n}}\right) .  \label{thetaZ0}
\end{equation}%
For the second term, we have%
\begin{eqnarray}
\left\vert \left( \hat{\theta}-\theta _{0}\right) ^{\prime
}\sum_{i}Z_{i}\left( 1-\hat{F}_{n}\left( \hat{W}_{i}\right) \right) D_{i}%
\hat{k}_{ni}\right\vert &\leq &O_{p}\left( \frac{1}{\sqrt{n}}\right)
\sum_{i}\left\Vert Z_{i}\right\Vert \left( 1-\hat{F}_{n}\left( \hat{W}%
_{i}\right) \right) \hat{k}_{ni}  \notag \\
&\leq &O_{p}\left( \frac{h_{n}}{\sqrt{n}}\right) \sum_{i}\left\Vert
Z_{i}\right\Vert \hat{k}_{ni}  \notag \\
&=&o_{p}\left( \sqrt{nh_{n}^{3}}\right) ,  \label{thetaZ1}
\end{eqnarray}%
where the last equality follows exactly the same line as the proof for $%
\left( \hat{\theta}-\theta _{0}\right) ^{\prime }\sum_{i}Z_{i}D_{i}\hat{k}%
_{ni}$.

\textbf{Finally}, prove (b)-(d) in (\ref{abcdL}). It follows from Lemma \ref%
{Lemma:Lambda} and Equations (\ref{Lambda0_diff}), (\ref{Lambda1_diff}), and
(\ref{Lambda2_diff}) that%
\begin{eqnarray*}
\Lambda _{r}\left( \pi \right) &=&\left\{ 
\begin{array}{cc}
o_{p}\left( nh_{n}^{r+1}\right) & \text{for }\pi \text{ being }U\text{,} \\ 
O_{p}\left( nh_{n}^{r+1}\right) & \text{for }\pi \text{ being }1\text{,}%
\end{array}%
\right. \\
\hat{\Lambda}_{r}\left( \pi \right) -\Lambda _{r}\left( \pi \right)
&=&\left\{ 
\begin{array}{cc}
o_{p}\left( \sqrt{nh_{n}^{2r+1}}\right) & \text{for }\pi \text{ being }U%
\text{,} \\ 
O_{p}\left( \sqrt{nh_{n}^{2r+1}}\right) & \text{for }\pi \text{ being }1%
\text{,}%
\end{array}%
\right. \\
\hat{\Lambda}_{r}\left( \pi \right) &=&\Lambda _{r}\left( \pi \right) +\left[
\hat{\Lambda}_{r}\left( \pi \right) -\Lambda _{r}\left( \pi \right) \right]
=\left\{ 
\begin{array}{cc}
o_{p}\left( nh_{n}^{r+1}\right) & \text{for }\pi \text{ being }U\text{,} \\ 
O_{p}\left( nh_{n}^{r+1}\right) & \text{for }\pi \text{ being }1\text{.}%
\end{array}%
\right.
\end{eqnarray*}%
Combining with (\ref{diff3}), (\ref{thetaZ0}), and (\ref{thetaZ1}) obtains%
\begin{eqnarray*}
\left( \hat{A}_{1n}\!-\!A_{1n}\right) \!-\!\left( \hat{B}_{n}\!-\!B_{n}%
\right) \mu _{0} &=&\hat{\Lambda}_{0}\left( U\right) \left[ \hat{\Lambda}%
_{2}\left( 1\right) -\Lambda _{2}\left( 1\right) \right] +\Lambda _{2}\left(
1\right) \left[ \hat{\Lambda}_{0}\left( U\right) -\Lambda _{0}\left(
U\right) \right] \\
&&-\hat{\Lambda}_{1}\left( U\right) \left[ \hat{\Lambda}_{1}\left( 1\right)
-\Lambda _{1}\left( 1\right) \right] -\Lambda _{1}\left( 1\right) \left[ 
\hat{\Lambda}_{1}\left( U\right) -\Lambda _{1}\left( U\right) \right] \\
&&-\hat{\Lambda}_{2}\left( 1\right) \left[ \left( \hat{\theta}-\theta
_{0}\right) ^{\prime }\sum_{i}Z_{i}D_{i}\hat{k}_{ni}\right] \\
&&+\hat{\Lambda}_{1}\left( 1\right) \left[ \left( \hat{\theta}-\theta
_{0}\right) ^{\prime }\sum_{i}Z_{i}\left( 1-\hat{F}_{n}\left( \hat{W}%
_{i}\right) \right) D_{i}\hat{k}_{ni}\right] \\
&=&o_{p}\left( \sqrt{n^{3}h_{n}^{7}}\right) , \\
\hat{A}_{0n}-A_{0n} &=&\hat{\Lambda}_{0}\left( U\right) \left[ \hat{\Lambda}%
_{1}\left( 1\right) -\Lambda _{1}\left( 1\right) \right] +\Lambda _{1}\left(
1\right) \left[ \hat{\Lambda}_{0}\left( U\right) -\Lambda _{0}\left(
U\right) \right] \\
&&-\hat{\Lambda}_{1}\left( U\right) \left[ \hat{\Lambda}_{0}\left( 1\right)
-\Lambda _{0}\left( 1\right) \right] -\Lambda _{0}\left( 1\right) \left[ 
\hat{\Lambda}_{1}\left( U\right) -\Lambda _{1}\left( U\right) \right] \\
&&-\hat{\Lambda}_{1}\left( 1\right) \left[ \left( \hat{\theta}-\theta
_{0}\right) ^{\prime }\sum_{i}Z_{i}D_{i}\hat{k}_{ni}\right] \\
&&+\hat{\Lambda}_{0}\left( 1\right) \left[ \left( \hat{\theta}-\theta
_{0}\right) ^{\prime }\sum_{i}Z_{i}\left( 1-\hat{F}_{n}\left( \hat{W}%
_{i}\right) \right) D_{i}\hat{k}_{ni}\right] \\
&=&o_{p}\left( \sqrt{n^{3}h_{n}^{5}}\right) , \\
\hat{B}_{n}-B_{n} &=&\hat{\Lambda}_{0}\left( 1\right) \left[ \hat{\Lambda}%
_{2}\left( 1\right) -\Lambda _{2}\left( 1\right) \right] +\Lambda _{2}\left(
1\right) \left[ \hat{\Lambda}_{0}\left( 1\right) -\Lambda _{0}\left(
1\right) \right] \\
&&-\left[ \hat{\Lambda}_{1}\left( 1\right) -\Lambda _{1}\left( 1\right) %
\right] \left[ \hat{\Lambda}_{1}\left( 1\right) +\Lambda _{1}\left( 1\right) %
\right] \\
&=&O_{p}\left( \sqrt{n^{3}h_{n}^{7}}\right) \\
&=&o_{p}\left( \sqrt{n^{3}h_{n}^{5}}\right) .
\end{eqnarray*}%
Consequently, (a)-(d) in (\ref{abcdL}) hold, which completes the proof.

\subsection{Several lemmas}

\begin{lemma}
\label{Lemma:E}Under Assumptions 1, 2 and 5.(i), $\left.
ED_{i}k_{ni}^{r}\right/ Ek_{ni}^{r}\rightarrow 1$ as $n\rightarrow \infty $
for any $r>0$.
\end{lemma}

\begin{proof}
Since $k\left( u\right) =0$ for $u>1$, we have%
\begin{eqnarray*}
\left| 1-\frac{ED_{i}k_{ni}^{r}}{Ek_{ni}^{r}}\right|  &=&\frac{E\left[
1\left\{ \varepsilon _{i}\geq W_{i}\right\} k_{ni}^{r}\right] }{Ek_{ni}^{r}}=%
\frac{E\left[ 1\left\{ F\left( \varepsilon _{i}\right) \geq F\left(
W_{i}\right) >1-h_{n}\right\} k_{ni}^{r}\right] }{Ek_{ni}^{r}} \\
&\leq &\frac{E\left[ 1\left\{ F\left( \varepsilon _{i}\right)
>1-h_{n}\right\} k_{ni}^{r}\right] }{Ek_{ni}^{r}}=\Pr \left( F\left(
\varepsilon _{i}\right) >1-h_{n}\right) \rightarrow 0.
\end{eqnarray*}
\end{proof}

\begin{lemma}
\label{Lemma:bias}Under Assumptions 1, 2 and 5.(i),%
\begin{equation*}
\frac{EU_{i}D_{i}k_{ni}}{ED_{i}k_{ni}}\rightarrow 0.
\end{equation*}
\end{lemma}

\begin{proof}
Since $EU_{i}k_{ni}=EU_{i}Ek_{ni}=0$ by Assumption 1.(ii), we have%
\begin{eqnarray*}
\left| \frac{EU_{i}D_{i}k_{ni}}{ED_{i}k_{ni}}\right|  &=&\frac{\left|
EU_{i}\left( 1-D_{i}\right) k_{ni}\right| }{Ek_{ni}\left( 1+o\left( 1\right)
\right) }\leq \frac{E\left[ \left| U_{i}\right| 1\left\{ F\left( \varepsilon
_{i}\right) >1-h_{n}\right\} k_{ni}\right] }{Ek_{ni}\left( 1+o\left(
1\right) \right) } \\
&=&E\left[ \left| U_{i}\right| 1\left\{ F\left( \varepsilon _{i}\right)
>1-h_{n}\right\} \right] \left( 1+o\left( 1\right) \right) \rightarrow 0,
\end{eqnarray*}%
where the first equality follows from Lemma \ref{Lemma:E}.
\end{proof}

\begin{lemma}
\label{Lemma:var}Under Assumptions 1, 2 and 5,%
\begin{equation*}
\frac{Var\left( U_{i}D_{i}k_{ni}\right) }{EU_{i}^{2}k_{ni}^{2}}\rightarrow 1.
\end{equation*}
\end{lemma}

\begin{proof}
As in Lemma \ref{Lemma:E}, we have%
\begin{equation*}
\left| 1-\frac{E\left[ U_{i}^{2}D_{i}k_{ni}^{2}\right] }{E\left[
U_{i}^{2}k_{ni}^{2}\right] }\right| =\frac{E\left[ U_{i}^{2}1\left\{ F\left(
\varepsilon _{i}\right) \geq F\left( W_{i}\right) >1-h_{n}\right\} k_{ni}^{2}%
\right] }{E\left[ U_{i}^{2}k_{ni}^{2}\right] }\leq \frac{E\left[
U_{i}^{2}1\left\{ F\left( \varepsilon _{i}\right) >1-h_{n}\right\} \right] }{%
EU_{i}^{2}}\rightarrow 0.
\end{equation*}%
On the other hand,%
\begin{equation*}
\frac{\left( E\left[ U_{i}D_{i}k_{ni}\right] \right) ^{2}}{E\left[
U_{i}^{2}k_{ni}^{2}\right] }\leq \frac{\left( E\left[ \left| U_{i}\right|
k_{ni}\right] \right) ^{2}}{E\left[ U_{i}^{2}k_{ni}^{2}\right] }=\frac{%
\left( E\left| U_{i}\right| \right) ^{2}\left( Ek_{ni}\right) ^{2}}{%
EU_{i}^{2}\cdot Ek_{ni}^{2}}\leq \frac{\left( E\left| U_{i}\right| \right)
^{2}\left( \bar{k}\right) ^{2}\left[ \Pr \left( F\left( W_{i}\right)
>1-h_{n}\right) \right] ^{2}}{EU_{i}^{2}\cdot Ek_{ni}^{2}}\rightarrow 0
\end{equation*}%
using Assumption 5.(iii). Consequently,%
\begin{equation*}
\frac{Var\left( U_{i}D_{i}k_{ni}\right) }{EU_{i}^{2}k_{ni}^{2}}=\frac{E\left[
U_{i}^{2}D_{i}k_{ni}^{2}\right] }{EU_{i}^{2}k_{ni}^{2}}-\frac{\left( E\left[
U_{i}D_{i}k_{ni}\right] \right) ^{2}}{EU_{i}^{2}k_{ni}^{2}}\rightarrow 1.
\end{equation*}
\end{proof}

\begin{lemma}
\label{Lemma:sample}Under Assumptions 1, 2 and 5,%
\begin{equation*}
\frac{\left( \left. 1\right/ n\right) \sum_{i=1}^{n}D_{i}k_{ni}}{ED_{i}k_{ni}%
}\overset{p}{\rightarrow }1.
\end{equation*}
\end{lemma}

\begin{proof}
It is sufficient to show that%
\begin{equation*}
E\left[ \frac{\left( \left. 1\right/ n\right) \sum_{i=1}^{n}\left(
D_{i}k_{ni}-ED_{i}k_{ni}\right) }{ED_{i}k_{ni}}\right] ^{2}\rightarrow 0.
\end{equation*}%
The left-hand side equals%
\begin{equation*}
Var\left( \frac{\left( \left. 1\right/ n\right) \sum_{i=1}^{n}D_{i}k_{ni}}{%
ED_{i}k_{ni}}\right) =\frac{Var\left( D_{i}k_{ni}\right) }{n\left(
ED_{i}k_{ni}\right) ^{2}}\leq \frac{ED_{i}k_{ni}^{2}}{n\left(
ED_{i}k_{ni}\right) ^{2}}\leq \frac{\bar{k}}{nED_{i}k_{ni}^2}=\frac{\bar{k}}{%
nEk_{ni}^2}\left( 1+o_{p}\left( 1\right) \right) ,
\end{equation*}%
where the last equality follows from Lemma \ref{Lemma:E}. The result
follows from Assumption 5.(ii).
\end{proof}

\begin{lemma}
\label{Lemma:Lindeberg}Denote $J_{ni}=U_{i}D_{i}k_{ni}$. Under Assumptions
1, 2 and 5, the Lindeberg's condition for the triangular array $\left\{
J_{ni}:i\leq n,n\geq 1\right\} $ holds, namely for any $\varepsilon >0$,%
\begin{equation}
E\left[ \frac{\left( J_{ni}-EJ_{ni}\right) ^{2}}{Var\left( J_{ni}\right) }%
1\left\{ \frac{\left( J_{ni}-EJ_{ni}\right) ^{2}}{Var\left( J_{ni}\right) }%
>n\varepsilon \right\} \right] \rightarrow 0.  \label{LindebergCondition}
\end{equation}
\end{lemma}

\begin{proof}
Because%
\begin{equation*}
E\left[ \frac{\left( J_{ni}-EJ_{ni}\right) ^{2}}{Var\left( J_{ni}\right) }%
1\left\{ \frac{\left( J_{ni}-EJ_{ni}\right) ^{2}}{Var\left( J_{ni}\right) }%
>n\varepsilon \right\} \right] \leq E\left[ \frac{1}{\left( n\varepsilon
\right) ^{\delta }}\left( \frac{\left( J_{ni}-EJ_{ni}\right) ^{2}}{Var\left(
J_{ni}\right) }\right) ^{1+\delta }\right] ,
\end{equation*}%
a sufficient condition for (\ref{LindebergCondition}) is%
\begin{equation*}
\frac{E\left[ \left| J_{ni}-EJ_{ni}\right| ^{2\left( 1+\delta \right) }%
\right] }{n^{\delta }\left[ Var\left( J_{ni}\right) \right] ^{1+\delta }}%
\rightarrow 0
\end{equation*}%
for some $\delta >0$. It follows from Minkowski's inequality and
H\"{o}lder's inequality that
\begin{equation*}
E\left[ \left| J_{ni}-EJ_{ni}\right| ^{2\left( 1+\delta \right) }\right]
\leq 4^{\left( 1+\delta \right) }E\left[ \left| J_{ni}\right| ^{2\left(
1+\delta \right) }\right] \leq 4^{\left( 1+\delta \right) }E\left[ \left|
U_{i}\right| ^{2\left( 1+\delta \right) }\right] \left( \bar{k}\right)
^{2\delta }Ek_{ni}^{2}.
\end{equation*}%
Let $\delta =\left. c_{1}\right/ 2$, then $E\left[ \left| U_{i}\right|
^{2\left( 1+\delta \right) }\right] <\infty $ by Assumption 1.(iv). As a
result,
\begin{equation*}
\frac{E\left[ \left| J_{ni}-EJ_{ni}\right| ^{2\left( 1+\delta \right) }%
\right] }{n^{\delta }\left[ Var\left( J_{ni}\right) \right] ^{1+\delta }}=%
\frac{O\left( 1\right) Ek_{ni}^{2}}{n^{\delta }\left( Ek_{ni}^{2}\right)
^{1+\delta }}\left( 1+o\left( 1\right) \right) =\frac{O\left( 1\right) }{%
\left( nEk_{ni}^{2}\right) ^{\delta }}\left( 1+o\left( 1\right) \right)
=o\left( 1\right) ,
\end{equation*}%
where the first equality follows from Lemma \ref{Lemma:var} and the last
equality follows from Assumption 5.(ii).
\end{proof}

\begin{lemma}
\label{Lemma:Lambda}Denote%
\begin{equation}
\Lambda _{r}\left( \pi \right) =\sum_{i=1}^{n}\left( 1-F_{W}\left(
W_{i}\right) \right) ^{r}\pi _{i}D_{i}k\left( \frac{1-F_{W}\left(
W_{i}\right) }{h_{n}}\right)  \label{Lambdar}
\end{equation}%
for $r\geq 0$ and $\pi $ being either $U$ or $1$. Under Assumptions 1, 2,
and $nh_{n}\rightarrow \infty $, we have%
\begin{equation*}
\Lambda _{r}\left( \pi \right) =nh_{n}^{r+1}\left( \kappa _{r}E\left[ \pi
_{i}\right] +o_{p}\left( 1\right) \right) .
\end{equation*}
\end{lemma}

\begin{proof}
It is enough to consider the expectation and variance of $\Lambda
_{r}\left( \pi \right) $ because%
\begin{equation*}
\Lambda _{r}\left( \pi \right) =E\left[ \Lambda _{r}\left( \pi \right) %
\right] +O_{p}\left( \sqrt{Var\left[ \Lambda _{r}\left( \pi \right) \right] }%
\right) .
\end{equation*}%
Denote $k_{ni}=k\left( \left. \left( 1-F_{W}\left( W_{i}\right) \right)
\right/ h_{n}\right) $, then we have%
\begin{eqnarray*}
E\left[ \Lambda _{r}\left( \pi \right) \right]  &=&nE\left[ \left(
1-F_{W}\left( W_{i}\right) \right) ^{r}G_{\pi }\left( F_{W}\left(
W_{i}\right) \right) k_{ni}\right]  \\
&=&nh_{n}^{r+1}\int_{0}^{1}G_{\pi }\left( 1-h_{n}t\right) t^{r}k\left(
t\right) dt \\
&=&nh_{n}^{r+1}\left[ \kappa _{r}G_{\pi }\left( 1\right) +o\left( 1\right) %
\right] ,
\end{eqnarray*}%
where $G_{\pi }\left( t\right) =E\left[ \pi _{i}1\left\{ F_{W}\left(
\varepsilon _{i}\right) <t\right\} \right] $ with $G_{\pi }\left( 1\right)
=E\pi _{i}$, and%
\begin{eqnarray*}
Var\left[ \Lambda _{r}\left( \pi \right) \right]  &=&nVar\left[ \left(
1-F_{W}\left( W_{i}\right) \right) ^{r}\pi _{i}D_{i}k_{ni}\right]  \\
&\leq &nE\left[ \left( 1-F_{W}\left( W_{i}\right) \right) ^{2r}\pi
_{i}^{2}k_{ni}^{2}\right]  \\
&=&nh_{n}^{2r+1}\chi _{2r}E\pi _{i}^{2}.
\end{eqnarray*}%
Therefore,%
\begin{equation*}
\Lambda _{r}\left( \pi \right) =nh_{n}^{r+1}\left( \kappa _{r}E\pi
_{i}+o\left( 1\right) \right) +O_{p}\left( \sqrt{nh_{n}^{2r+1}}\right)
=nh_{n}^{r+1}\left( \kappa _{r}E\pi _{i}+o_{p}\left( 1\right) \right) .
\end{equation*}
\end{proof}

\begin{lemma}
\label{Lemma:alpha_ij}Denote $\alpha _{ij}\left( \beta \right) =1\left\{
X_{ij}^{\prime }\beta <0\right\} -1\left\{ X_{ij}^{\prime }\beta
_{0}<0\right\} $, where $X_{ij}=X_{i}-X_{j}$. Denote $\tilde{\gamma}%
_{n}=F_{W}^{-1}\left( 1-h_{n}\right) -n^{-1/10}h_{n}^{-1/18}$. Let $M_{n}$
be a positive sequence satisfying $M_{n}\rightarrow +\infty $, $M_{n}\leq
h_{n}^{-1/18}$, and $M_{n}=o\left( \left. h_{n}^{2/3}\right/ f_{W}\left( 
\tilde{\gamma}_{n}\right) \right) $.\ Then, under Assumptions 1', 3', 4',
5', and 6, we have%
\begin{equation*}
J_{n}^{\left( r\right) }\left( c\right) :=\left( \frac{M_{n}}{\sqrt{%
nh_{n}^{2}}}\right) ^{r}\sqrt{\frac{n}{h_{n}^{3}}}\sup_{\left\| \beta -\beta
_{0}\right\| \leq \frac{M_{n}}{\sqrt{n}}}E\left[ 1\left\{ F_{W}\left(
W_{i}\right) >1-h_{n}\right\} \left( \left\| X_{i}\right\| +c\right)
^{r}\left| \alpha _{ij}\left( \beta \right) \right| \right] \rightarrow 0
\end{equation*}%
for any $r\in \left[ 0,4\right] $ and $c\in R$.
\end{lemma}

\begin{proof}
Since $\left| \alpha _{ij}\left( \beta \right) \right| \leq 1$
and by Markov's inequality and Assumption 1',%
\begin{eqnarray*}
E\left[ \left( \left\| X_{i}\right\| +c\right) ^{r}1\left\{ \left\|
X_{ij,(-1)}\right\| >n^{2/5}\right\} \right] &=&E\left[ \left( \left\|
X_{i}\right\| +c\right) ^{r}\Pr \left( \left. \left\| X_{ij,(-1)}\right\|
>n^{2/5}\right| X_{i}\right) \right] \\
&\leq &E\left[ \left( \left\| X_{i}\right\| +c\right) ^{r}\frac{E\left[
\left. \left\| X_{ij,(-1)}\right\| ^{6-r}\right| X_{i}\right] }{n^{\left(
2/5\right) \left( 6-r\right) }}\right] \\
&=&O\left( n^{-\left( 2/5\right) \left( 6-r\right) }\right) ,
\end{eqnarray*}%
it follows that for $0\leq r\leq 4$,%
\begin{eqnarray*}
&&\left( \frac{M_{n}}{\sqrt{nh_{n}^{2}}}\right) ^{r}\sqrt{\frac{n}{h_{n}^{3}}%
}\sup_{\beta }E\left[ 1\left\{ F_{W}\left( W_{i}\right) >1-h_{n}\right\}
\left( \left\| X_{i}\right\| +c\right) ^{r}\left| \alpha _{ij}\left( \beta
\right) \right| 1\left\{ \left\| X_{ij,(-1)}\right\| >n^{\frac{2}{5}%
}\right\} \right] \\
&=&O\left( \left( \frac{M_{n}}{\sqrt{nh_{n}^{2}}}\right) ^{r}\sqrt{\frac{n}{%
h_{n}^{3}}}n^{-\frac{2}{5}\left( 6-r\right) }\right) =O\left( n^{-\frac{1}{10%
}r-\frac{19}{10}}h_{n}^{-\frac{19}{18}r-\frac{3}{2}}\right) =o\left( \left(
nh_{n}^{13/5}\right) ^{-\frac{1}{10}r-\frac{19}{10}}\right) =o\left(
1\right) .
\end{eqnarray*}%
Therefore, we have%
\begin{eqnarray*}
J_{n}^{\left( r\right) }\left( c\right) &=&\left( \frac{M_{n}}{\sqrt{%
nh_{n}^{2}}}\right) ^{r}\sqrt{\frac{n}{h_{n}^{3}}}\sup_{\left\| \beta -\beta
_{0}\right\| \leq M_{n}\left/ \sqrt{n}\right. }E\left[
\begin{array}{c}
1\left\{ F\left( W_{i}\right) >1-h_{n}\right\} \left( \left\| X_{i}\right\|
+c\right) ^{r} \\
\cdot \left| \alpha _{ij}\left( \beta \right) \right| 1\left\{ \left\|
X_{ij,(-1)}\right\| \leq n^{\frac{2}{5}}\right\}%
\end{array}%
\right] +o\left( 1\right) \\
&\leq &\left( \frac{M_{n}}{\sqrt{nh_{n}^{2}}}\right) ^{r}\sqrt{\frac{n}{%
h_{n}^{3}}}E\left[
\begin{array}{c}
1\left\{ F\left( W_{i}\right) >1-h_{n}\right\} \left( \left\| X_{i}\right\|
+c\right) ^{r} \\
\cdot \underset{\left\| \beta -\beta _{0}\right\| \leq M_{n}\left/ \sqrt{n}%
\right. }{\sup }E\left[ \left. \left| \alpha _{ij}\left( \beta \right)
\right| 1\left\{ \left\| X_{ij,(-1)}\right\| \leq n^{\frac{2}{5}}\right\}
\right| X_{i}\right]%
\end{array}%
\right] +o\left( 1\right) .
\end{eqnarray*}

Note that%
\begin{eqnarray*}
\left| \alpha _{ij}\left( \beta \right) \right| &=&1\left\{ -X_{ij}^{\prime
}\left( \beta -\beta _{0}\right) \leq W_{ij}<0\right\} +1\left\{ 0\leq
W_{ij}<-X_{ij}^{\prime }\left( \beta -\beta _{0}\right) \right\} \\
&=&1\left\{ W_{i}<W_{j}\leq W_{i}+X_{ij,(-1)}^{\prime }\left( \beta
_{(-1)}-\beta _{0,(-1)}\right) \right\} \\
&&+1\left\{ W_{i}+X_{ij,(-1)}^{\prime }\left( \beta _{(-1)}-\beta
_{0,(-1)}\right) <W_{j}\leq W_{i}\right\} ,
\end{eqnarray*}%
where the second equality follows from the scale normalization $\hat{\beta}%
_{1}=\beta _{01}=1$. Under $\left\| \beta -\beta _{0}\right\| \leq
M_{n}\left/ \sqrt{n}\right. $ and $\left\| X_{ij,(-1)}\right\| \leq n^{2/5}$%
, we have $\left| X_{ij,(-1)}^{\prime }\left( \beta _{(-1)}-\beta
_{0,(-1)}\right) \right| \leq M_{n}\left/ n^{1/10}\right. \leq
n^{-1/10}h_{n}^{-1/18}$, thus for large enough $n$ such that $\tilde{\gamma}%
_{n}\geq C$,%
\begin{eqnarray*}
&&E\left[ \left. \left| \alpha _{ij}\left( \beta \right) \right| \right|
X_{i},X_{j,(-1)}\right] \\
&=&1\left\{ X_{ij,(-1)}^{\prime }\left( \beta _{(-1)}-\beta _{0,(-1)}\right)
>0\right\} \int_{W_{i}}^{W_{i}+X_{ij,(-1)}^{\prime }\left( \beta
_{(-1)}-\beta _{0,(-1)}\right) }f_{W\left| X_{(-1)}\right. }\left( w\left|
X_{j,(-1)}\right. \right) dw \\
&&+1\left\{ X_{ij,(-1)}^{\prime }\left( \beta _{(-1)}-\beta _{0,(-1)}\right)
<0\right\} \int_{W_{i}+X_{ij,(-1)}^{\prime }\left( \beta _{(-1)}-\beta
_{0,(-1)}\right) }^{W_{i}}f_{W\left| X_{(-1)}\right. }\left( w\left|
X_{j,(-1)}\right. \right) dw \\
&\leq &f_{W\left| X_{(-1)}\right. }\left( \tilde{\gamma}_{n}\left|
X_{j,(-1)}\right. \right) \left| X_{ij,(-1)}^{\prime }\left( \beta
_{(-1)}-\beta _{0,(-1)}\right) \right| \\
&\leq &\frac{M_{n}}{\sqrt{n}}f_{W\left| X_{(-1)}\right. }\left( \tilde{\gamma%
}_{n}\left| X_{j,(-1)}\right. \right) \left\| X_{ij,(-1)}\right\| .
\end{eqnarray*}%
As a result,%
\begin{eqnarray*}
&&\sup_{\left\| \beta -\beta _{0}\right\| \leq M_{n}\left/ \sqrt{n}\right. }E%
\left[ \left. \left| \alpha _{ij}\left( \beta \right) \right| 1\left\{
\left\| X_{ij,(-1)}\right\| \leq n^{2/5}\right\} \right| X_{i}\right] \\
&\leq &\frac{M_{n}}{\sqrt{n}}E\left[ \left. f_{W\left| X_{(-1)}\right.
}\left( \tilde{\gamma}_{n}\left| X_{j,(-1)}\right. \right) \left\|
X_{ij,(-1)}\right\| \right| X_{i}\right] \\
&=&\frac{M_{n}}{\sqrt{n}}\int f_{W\left| X_{(-1)}\right. }\left( \tilde{%
\gamma}_{n}\left| x_{(-1)}\right. \right) f_{X_{(-1)}}\left( x_{(-1)}\right)
\left\| X_{i,(-1)}-x_{(-1)}\right\| dx_{(-1)} \\
&=&\frac{M_{n}}{\sqrt{n}}\int f_{\left. X_{(-1)}\right| W}\left(
x_{(-1)}\left| \tilde{\gamma}_{n}\right. \right) f_{W}\left( \tilde{\gamma}%
_{n}\right) \left\| X_{i,(-1)}-x_{(-1)}\right\| dx_{(-1)} \\
&=&\frac{M_{n}}{\sqrt{n}}f_{W}\left( \tilde{\gamma}_{n}\right) E\left[
\left. \left\| X_{ij,(-1)}\right\| \right| X_{i},W_{j}=\tilde{\gamma}_{n}%
\right] \\
&\leq &\frac{M_{n}}{\sqrt{n}}f_{W}\left( \tilde{\gamma}_{n}\right) \left(
\left\| X_{i}\right\| +\left( E\left[ \left. \left\| X_{j}\right\|
^{6}\right| W_{j}=\tilde{\gamma}_{n}\right] \right) ^{1/6}\right) .
\end{eqnarray*}%
Denote $J\left( w\right) =E\left[ \left. \left\| X_{j}\right\| ^{6}\right|
W_{j}=w\right] $. Since%
\begin{equation*}
E\left\| X_{j}\right\| ^{6}=E\left[ J\left( W_{j}\right) \right]
=\int_{-\infty }^{+\infty }J\left( w\right) f_{W}\left( w\right) dw<\infty ,
\end{equation*}%
we know that $J\left( w\right) f_{W}\left( w\right) \rightarrow 0$ as $%
w\rightarrow +\infty $. It then follows that $J^{1/6}\left( \tilde{\gamma}%
_{n}\right) =o\left( f_{W}^{-1/6}\left( \tilde{\gamma}_{n}\right) \right) $,
and that%
\begin{equation*}
\sup_{\left\| \beta -\beta _{0}\right\| \leq \frac{M_{n}}{\sqrt{n}}}E\left[
\left. \left| \alpha _{ij}\left( \beta \right) \right| 1\left\{ \left\|
X_{ij,(-1)}\right\| \leq n^{2/5}\right\} \right| X_{i}\right] \leq \frac{%
M_{n}}{\sqrt{n}}f_{W}\left( \tilde{\gamma}_{n}\right) \left\| X_{i}\right\|
+o\left( \frac{M_{n}}{\sqrt{n}}f_{W}^{5/6}\left( \tilde{\gamma}_{n}\right)
\right) .
\end{equation*}

Consequently, for $0\leq r\leq 4$,%
\begin{eqnarray*}
J_{n}^{\left( r\right) }\left( c\right) &\leq &\left( \frac{M_{n}}{\sqrt{%
nh_{n}^{2}}}\right) ^{r}\frac{M_{n}}{h_{n}^{3/2}}f_{W}\left( \tilde{\gamma}%
_{n}\right) E\left[ 1\left\{ F_{W}\left( W_{i}\right) >1-h_{n}\right\}
\left( \left\| X_{i}\right\| +c\right) ^{r}\left\| X_{i}\right\| \right] \\
&&+o\left( \left( \frac{M_{n}}{\sqrt{nh_{n}^{2}}}\right) ^{r}\frac{M_{n}}{%
h_{n}^{3/2}}f_{W}^{5/6}\left( \tilde{\gamma}_{n}\right) \right) E\left[
1\left\{ F_{W}\left( W_{i}\right) >1-h_{n}\right\} \left( \left\|
X_{i}\right\| +c\right) ^{r}\right] +o\left( 1\right) \\
&\leq &\left( \frac{M_{n}}{\sqrt{nh_{n}^{2}}}\right) ^{r}\frac{M_{n}}{%
h_{n}^{3/2}}f_{W}\left( \tilde{\gamma}_{n}\right) \left( E\left[ \left(
\left\| X_{i}\right\| +c\right) ^{6r\left/ \left( r+1\right) \right.
}\left\| X_{i}\right\| ^{6\left/ \left( r+1\right) \right. }\right] \right)
^{\left. \left( r+1\right) \right/ 6}h_{n}^{1-\left. \left( r+1\right)
\right/ 6} \\
&&+o\left( \left( \frac{M_{n}}{\sqrt{nh_{n}^{2}}}\right) ^{r}\frac{M_{n}}{%
h_{n}^{3/2}}f_{W}^{5/6}\left( \tilde{\gamma}_{n}\right) \right) \left( E%
\left[ \left( \left\| X_{i}\right\| +c\right) ^{6}\right] \right) ^{\left.
r\right/ 6}h_{n}^{1-\left. r\right/ 6}+o\left( 1\right) \\
&=&o\left( n^{-\left( 1/2\right) r}h_{n}^{-\left( 11/9\right) r}\right)
+o\left( n^{-\left( 1/2)\right) r}h_{n}^{-\left( 11/9\right) r+\left(
5/108\right) }\right) +o\left( 1\right) \\
&=&o\left( \left( nh_{n}^{22/9}\right) ^{-\left( 1/2\right) r}\right)
+o\left( 1\right) =o\left( 1\right) ,
\end{eqnarray*}%
where the second inequality follows from H\"{o}lder's inequality, the first
equality follows from $M_{n}\leq h_{n}^{-1/18}$ and $M_{n}f_{W}\left( \tilde{%
\gamma}_{n}\right) =o\left( h_{n}^{2/3}\right) $, and the last equality
follows from Assumption 5'.
\end{proof}

\section*{Acknowledgements}

This work was supported by the National Natural Science Foundation of China
[grant numbers 72173142, 71991474] and the Guangdong Basic and Applied Basic
Research Foundation [grant number 2022A1515010079].

\bibliographystyle{Chicago}
\bibliography{0RefIntercept}

\end{document}